\documentclass[11pt]{book} 

\usepackage[utf8]{inputenc} 
\usepackage{geometry} 
\usepackage{booktabs} 
\usepackage{array} 
\usepackage{paralist} 
\usepackage{verbatim} 
\usepackage{subfig} 
\usepackage{amssymb}
\usepackage{amsmath}

\usepackage{amsthm}

\usepackage{amsfonts}
\usepackage[sc]{mathpazo}
\usepackage{eurosym}
\usepackage{listings}
\usepackage{makeidx}
\usepackage{accents}
\usepackage{graphicx}
\usepackage{tikz}
\usepackage{verbatim}
\usetikzlibrary{intersections}
\usetikzlibrary{automata}
\usetikzlibrary{patterns}   

\usepackage{amsthm}
\usepackage{natbib}
\usepackage{float}

\newtheorem{theorem}{Theorem}
 \numberwithin{theorem}{chapter}

\newtheorem{axiom}{Axiom}
\newtheorem{assumption}{Assumption}

\newtheorem{claim}{Claim}

\newtheorem{corollary}{Corollary}
\numberwithin{corollary}{chapter}

\newtheorem{definition}{\normalfont\scshape Definition}
 \numberwithin{definition}{chapter}
\newtheorem{exercise}{\normalfont\scshape Exercise}
\newtheorem{hint}{\normalfont\scshape Hint}
 \numberwithin{exercise}{chapter}
\newtheorem{fact}{\normalfont\scshape Fact}
 \numberwithin{fact}{chapter}
\newtheorem{lemma}{Lemma}

\newtheorem{proposition}{Proposition}

\theoremstyle{remark}
\newtheorem{example}{\normalfont\scshape Example}
\numberwithin{example}{chapter}
\newtheorem{property}{Property}
\newtheorem{remark}{\normalfont\scshape Remark}
\numberwithin{remark}{chapter}

\renewenvironment{proof}[1][Proof]{\noindent\textbf{#1.} }{\ \rule{0.5em}{0.5em}}

\allowdisplaybreaks
\usepackage[subfigure]{tocloft}

\makeatletter
\@addtoreset{section}{part}
\makeatother
\newlength\mylen

\renewcommand\cftpartpresnum{Part~}
\newtheorem*{theorem*}{Theorem}
\newtheorem*{axiom*}{Axiom}
\newtheorem*{case*}{Case}
\newtheorem*{claim*}{Claim}
\newtheorem*{conclusion*}{Conclusion}
\newtheorem*{condition*}{Condition}
\newtheorem*{conjecture*}{Conjecture}
\newtheorem*{corollary*}{Corollary}
\newtheorem*{criterion*}{Criterion}
\newtheorem*{definition*}{Definition}
\newtheorem*{example*}{Example}
\newtheorem*{exercise*}{Exercise}
\newtheorem*{lemma*}{Lemma}
\newtheorem*{notation*}{Notation}
\newtheorem*{problem*}{Problem}
\newtheorem*{proposition*}{Proposition}
\newtheorem*{remark*}{Remark}
\newtheorem*{summary*}{Summary}
\newtheorem*{hint*}{\normalfont\scshape Hint}

\renewcommand{\arraystretch}{1.5}

\usepackage{scrextend}

\usepackage{accents}

\usepackage{graphicx} 
\usepackage{booktabs} 

\usepackage{caption}

\usepackage{xcolor}

\usepackage{pgfplots}

\pgfplotsset{compat=1.10}

\usepackage{tikz}

\usepackage{enumitem}

\usepackage{mathrsfs}
\usepackage{bbm}

\usetikzlibrary{quotes,positioning}

\newenvironment{itemize*}%
{\begin{itemize}%
		\setlength{\topsep}{1pt}%
		\setlength{\partopsep}{1pt}%
		\setlength{\itemsep}{1pt}%
		\setlength{\parskip}{1pt}}%
	{\end{itemize}}

\newenvironment{enumerate*}%
{\begin{itemize}%
		\setlength{\topsep}{1pt}%
		\setlength{\partopsep}{1pt}%
		\setlength{\itemsep}{1pt}%
		\setlength{\parskip}{1pt}}%
	{\end{itemize}}
	
\usepackage{hyperref}

\newtheorem{innercustomgeneric}{\customgenericname}
\providecommand{\customgenericname}{}
\newcommand{\newcustomtheorem}[2]{%
	\newenvironment{#1}[1]
	{%
		\renewcommand\customgenericname{#2}%
		\renewcommand\theinnercustomgeneric{##1}%
		\innercustomgeneric
	}
	{\endinnercustomgeneric}
}

\newcustomtheorem{customthm}{Theorem}
\newcustomtheorem{customlemma}{Lemma}

\usepackage{placeins}

\newenvironment{Tabular}{
	\setlength{\tabcolsep}{2pt} 
	\renewcommand{\arraystretch}{1.2} 
	\tabular%
}{\endtabular}

\usepackage{tikz}

\usepackage{enumitem}

\usepackage{mathrsfs}

\usetikzlibrary{quotes,positioning}

\usepackage{blkarray}

\usepackage{amsmath}
\usepackage{tikz}
\newcommand{\tikznode}[2]{%
	\ifmmode%
	\tikz[remember picture,baseline=(#1.base),inner sep=0pt] \node (#1) {$#2$};%
	\else
	\tikz[remember picture,baseline=(#1.base),inner sep=0pt] \node (#1) {#2};%
	\fi}

\usepackage{amsmath}
\usepackage{tikz}
\usetikzlibrary{tikzmark}

\usepackage{placeins}

\usepackage{tikz}

\usepackage{tikz}
\usetikzlibrary{matrix}

\newcommand{\meet}{\wedge}
 
\providecommand{\eps}{\varepsilon}
\newcommand{\indep}{\perp \!\!\! \perp}

\DeclareMathOperator*{\argmax}{arg\,max}
\DeclareMathOperator*{\argmin}{arg\,min}
\DeclareMathOperator*{\Var}{\text{Var}}

\newcommand*{\titleGP}{\begingroup
\centering
\vspace*{\baselineskip}
\rule{\textwidth}{1.6pt}\vspace*{-\baselineskip}\vspace*{2pt}
\rule{\textwidth}{0.4pt}\\[\baselineskip]
{\LARGE INFORMATION AND LEARNING \\[0.2\baselineskip]
IN\\[0.6\baselineskip]
ECONOMIC THEORY}\\[0.2\baselineskip]
\rule{\textwidth}{0.4pt}\vspace*{-\baselineskip}\vspace{3.2pt}
\rule{\textwidth}{1.6pt}\\[\baselineskip]
\scshape
(IN PROGRESS AND INCOMPLETE) \\
\vspace*{2\baselineskip}
{\Large ANNIE LIANG \par}
{\itshape annie.liang@northwestern.edu\par}
\endgroup}

\begin{document}



\pagenumbering{gobble}

\titleGP

\clearpage

\noindent \Large{\textbf{Preface}}
\bigskip 
\normalsize

These lecture notes are targeted towards graduate students. Several chapters are also appropriate for an advanced undergraduate audience (e.g., Chapters 1-4 and Chapter 8). Exercises labeled U indicate an advanced undergraduate level, exercises labeled G indicate graduate level, and exercises labeled G$^*$ are either more involved or rely on background knowledge that is not covered in these notes.

Parts of these lecture notes are based on slides by Drew Fudenberg and Xiaosheng Mu. Hershdeep Chopra, Andrei Iakovlev, and Tiago Cardoso Botelho provided very helpful proofreading. I am also grateful to Kyohei Okumura, DJ Thornton, and Jingyi Qiu for comments and corrections.  All errors are mine alone. Please feel free to email  annie.liang@northwestern.edu with any edits and/or suggestions.

\clearpage

\normalsize

\pagenumbering{arabic}
\tableofcontents

\part{\sc{Foundations of Information}}
		
\chapter{Information Partitions and Knowledge}

Section \ref{sec:Partition} introduces the partitional model of information and three definitions of common knowledge.  Sections \ref{sec:AgreetoDisagree} presents \citet{Aumann}'s result that agents cannot agree to disagree. Section \ref{sec:emailgame} presents \citet{Rubinstein}'s email game. Section \ref{sec:pBelief} defines common $p$-belief.

We assume a finite state space in Sections \ref{sec:Partition}-\ref{sec:pBelief} to ease exposition, and discuss in Section \ref{sec:General} how these results extend more generally.

\section{Common Knowledge} \label{sec:Partition}

An unknown state $\omega$ takes values in the finite set $\Omega$. Agents $i\in \mathcal{I}$ share a \emph{common prior} that the state $\omega$ is distributed according  $P \in \Delta(\Omega)$.  Each agent $i$'s \emph{information partition}  $\Pi_i$ is a partition of $\Omega$, with the property that for any realization of the state $\omega$, agent $i$ is informed that the state belongs to $\Pi_i(\omega)$. 

\begin{assumption} Every partition element has strictly positive probability under the prior; that is, $P(\Pi_i(\omega))>0$ for every agent $i \in \mathcal{I}$ and state $\omega \in \Omega$. 
\end{assumption}

\begin{definition}[Knowledge] \label{def:Knowledge} The set of states at which agent $i$ \emph{knows} the event $A \subseteq \Omega$ to be true is 
\[K_i(A)=\{\omega: \Pi_i(\omega)\subseteq A \}.\]
\end{definition}

\noindent No agent can think that an event is true if it is not; that is, $K_i(A) \subseteq A$ for every agent $i$ and event $A$. 

\begin{definition}[Mutual Knowledge] The set of states at which the event $A \subseteq \Omega$ is \emph{mutual knowledge} is
\[K(A)=\bigcap_{i \in \mathcal{I}} \{\omega: \Pi_i(\omega)\subseteq A \}\]
i.e., all agents know $A$ to be true.
\end{definition}

\begin{example} \label{ex:Partition} Suppose the set of states is $\Omega =\{1,2,3,4,5,6\}$, and there are two agents with information partitions 
$
\Pi_1 =\{\{1,2,3\},\{4,5\},\{6\}\}$ and 
$\Pi_2 = \{\{1,2\},\{3,4\},\{5\},\{6\}\}$. Let $A = \{3,4,5,6\}$. Then, the set of states at which agent 1 knows $A$ to be true is $K_1(A) = \{4,5,6\}$, the set of states at which agent 2 knows $A$ to be true is $K_2(A) = \{3,4,5,6\}$, and the set of states at which both agents know $A$ to be true is $K(A)=\{4,5,6\}$.
\end{example}

The knowledge operators $K_i$ and $K$ can be applied to events that themselves represent knowledge or mutual knowledge of a state, thus building up higher-order knowledge (agent 1 knows that agent 2 knows that\dots). 
 
\begin{exercise}[U] Suppose there are two agents indexed to $i=1,2$.
\begin{itemize}
\item[(a)] Prove that  $K_1(K_2(A)), K_2(K_1(A)) \subseteq K(A)$ for every event $A\subseteq \Omega$.
\item[(b)] Provide an example in which $K(A) \nsubseteq K_1(K_2(A))$, demonstrating that even if both players  know an event to be true, either can fail to know that the other knows it. 
\end{itemize}
\end{exercise}

\begin{exercise}[U] Prove that $\neg K_i(\neg K_i(A)) = K_i(A)$ for every event $A\subseteq \Omega$ (where $\neg A$ denotes the complement of $A$.)
\end{exercise}

An implicit assumption is made that all agents know the state space $\Omega$ and the information partitions $(\Pi_i)_{i \in \mathcal{I}}$. This assumption is less strong than it might initially seem, since we can always redefine states and expand the state space to accommodate uncertainty about other players' partitions, as in the following example.

\begin{example} \label{ex:ExpandOmega} Let $\Omega = \{1,2\}$, $\mathcal{I}=\{1,2\}$, and  $\Pi_1 = \Pi_2 = \{ \{1\},\{2\}\}$. Suppose we want to model the situation where agent 1 has uncertainty over whether agent 2's information is the complete partition $\Pi_2$ or the trivial partition $\Pi_2' = \{\{1,2\}\}$. One way to do this is to expand the state space: Define $\widetilde{\Omega}=\Omega \times \{c,t\} = \{\{1,c\},\{1,t\},\{2,c\},\{2,t\}\}$ and revise the agents' information partitions to be 
\begin{align*}
\widetilde{\Pi}_1 & = \{\{(1,c),(1,t)\},\{(2,c),(2,t)\}\} \\
\widetilde{\Pi}_2 & =\{\{(1,c)\},\{(2,c)\},\{(1,t),(2,t)\}\}
\end{align*}
Then, for example, at state $(1,c)$ both agents know $\omega=1$ to be true, but agent 1 does not know whether agent 2 knows it. 
\end{example}

The event $A$ is \emph{common knowledge} at state $\omega$ if both agents know it to be true, know the other to know it to be true, ad infinitum. There are at least three equivalent ways to define this.

\paragraph{The First Definition.} The most direct approach is to recursively define higher-order levels of knowledge. 

\begin{definition}[Common Knowledge, Definition 1] \label{def:CK1} For any event $A \subseteq \Omega$, define $
		\mathscr{A}^1:=\bigcap_{i\in \mathcal{I}}K_i(A)$ to be the set of states at which every agent knows $A$, and recursively define $$\mathscr{A}^k:= \bigcap_{i\in \mathcal{I}}K_i(\mathscr{A}^{k-1})$$ for each $k \geq 2$. (For example, $\mathscr{A}^2$ is the set of states at which every agent knows that every agent knows $A$.) The set of states at which $A$ is \emph{common knowledge} is $ \mathscr{A}^\infty := \bigcap_{n\geq1} \mathscr{A}^n $.
\end{definition}

\begin{exercise}[G] Consider the informational environment of Example \ref{ex:Partition}. Find the smallest value of $k$ with the property that $\mathscr{A}^{k'}=\mathscr{A}^{k'+1}$ for all $k' \geq k$. 
\end{exercise}
		
\paragraph{The Second Definition.} Alternatively, we can define common knowledge using the meet of the players' information partitions. If two partitions $\Pi$ and $\Pi'$ satisfy
\[\Pi'(\omega) \subseteq \Pi(\omega) \quad \forall \omega \in \Omega\]
then we say that $\Pi$ is a \emph{coarsening} of $\Pi'$ (corresponding to weakly less information at every state), and $\Pi'$ is a \emph{refinement} of $\Pi$ (corresponding to weakly more information at every state). If $\Pi'$  a coarsening of both partitions $\Pi_1$ and $\Pi_2$, then it is a \emph{common coarsening} of $\Pi_1, \Pi_2$.

\begin{definition} Let $ \Pi_1 \meet \Pi_2 $ denote the finest common coarsening of $ \Pi_1,\Pi_2 $, i.e., the common coarsening of these partitions that is moreover a refinement of every other common coarsening of $\Pi_1$ and  $\Pi_2$. 	
\end{definition}

\begin{definition} For any sequence of information partitions $(\Pi_1, \dots, \Pi_{\vert \mathcal{I} \vert})$, let $\mathscr{P}_2 =  \Pi_1 \meet \Pi_2 $, and for each $k>2$, recursively define $\mathscr{P}_k = \mathscr{P}_{k-1} \meet \Pi_k$. The  \emph{meet} of $(\Pi_1, \dots, \Pi_{\vert \mathcal{I} \vert})$ is $ \bigwedge_{i \in \mathcal{I}}  \Pi_i \equiv \mathscr{P}_{\vert \mathcal{I} \vert}$.
\end{definition}

\begin{exercise}[G] Prove that for any sequence of information partitions $(\Pi_1, \dots, \Pi_{\vert \mathcal{I} \vert})$, the meet $\mathscr{P}^n$ is invariant to permutations of players indices.
\end{exercise}

\begin{example} \label{example:Ant} Consider Example \ref{ex:Partition}. Stack the two information partitions on top of one another, and suppose an ant is placed on one of the states in an agent's partition (see Figure \ref{fig:Ant}).

\begin{figure}[H]
\begin{center}
\includegraphics[scale=0.85]{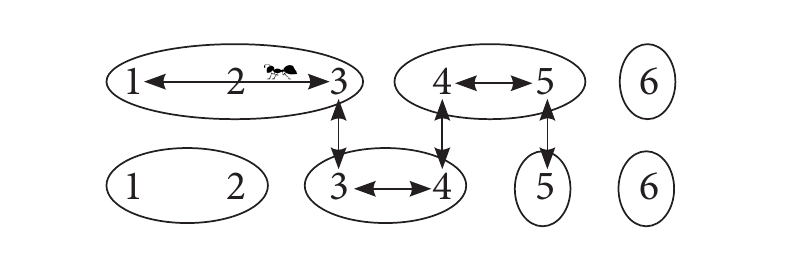}
\caption{$\Pi_1 \meet \Pi_2(\omega)$ includes all states that an ant seeded at $\omega$ can reach.}\label{fig:Ant}
\end{center} 
\end{figure}
\vspace{-4mm}
\noindent  The ant's movements obey two laws: The ant can move from side to side within an information partition element, and it can jump across the players' information partitions along the same state. The ant's full range of motion when seeded at any state $\omega$ then recovers the member of the meet that includes that state. So in this example, we have $ \Pi_1 \meet \Pi_2 = \{\{1,2,3,4,5\},\{6\}\}.$
\end{example}	

\begin{exercise}[G] Formalize the statements in the example above by proving that two points $x'$ and $x''$ belong to the same element of $\bigwedge_{i \in \mathcal{I}}  \Pi_i$ if and only if there is a sequence
$(x_0,x_1,x_2,…,x_n,x_{n+1})$, with $x_0=x'$ and $x_{n+1}=x''$, such that for every $0 \leq m \leq n$,  $x_m$ and $x_{m+1}$ belongs to the same element of $\Pi_i$ for some $i\in \mathcal{I}$.
\end{exercise}

\begin{definition}[Common Knowledge, Definition 2] \label{def:CK2}
An event $ A \subseteq \Omega $ is common knowledge at state $ \omega \in \Omega$ if $\bigwedge_{i \in \mathcal{I}}  \Pi_i(\omega) \subseteq A$.
\end{definition}   

\begin{remark} It is immediate that the set $\Omega$ is common knowledge at every $\omega \in \Omega$.
\end{remark}

\paragraph{The Third Definition.} Our final definition of common knowledge starts from the definition of an \emph{evident} event which, upon its occurrence, is known to all agents.

\begin{definition}[Evident Events] \label{def:Evident} The event $A \subseteq \Omega$ is \emph{evident} (or  \emph{public}) if $A \subseteq K(A)$.
\end{definition}


\begin{definition}[Common Knowledge, Definition 3] \label{def:CK3} The event $A \subseteq \Omega$ is common knowledge at $\omega$ if and only if there is an evident event $E$ such that $\omega \in E$ and $E \subseteq K(A)$.
\end{definition}

\begin{exercise}[G]
Let $\mathcal{I}=\{1,2\}$. Prove that an event $E\subseteq \Omega$ is evident if and only if it is a union of elements of the meet $\Pi_1 \wedge \Pi_2$. 
\end{exercise}

These three definitions of common knowledge are equivalent (see for example \citet{MondererSamet}).

\section{Agreeing to Disagree} \label{sec:AgreetoDisagree}

Often we are interested not only in agents' knowledge (which depends only on the agents' information partitions) but also in agents' posterior beliefs (which depend additionally on the prior $P$). At any state $\omega$ and for any event $A \subseteq \Omega$, agent $i$'s posterior probability of event $A$ is pinned down by Bayes' rule (see Section \ref{sec:Bayes}):
\[P(A \mid \Pi_i(\omega)) = \frac{P(A\cap \Pi_i(\omega))}{P(\Pi_i(\omega))}\]
Our assumption that every partition element has strictly positive prior probability ensures that this expression is well-defined.

One event of interest is the one in which a player's posterior belief takes on a particular value. Fixing an event $A$ and a number $p\in [0,1]$, define $A_p= \{ \omega : P(A \mid \Pi_i(\omega)) = p\}$ to be the set of states at which player $i$ assigns posterior probability $p$ to the event $A$ being true. If player 1 announces that he assigns probability $p$ to $A$, then all other agents know that the state must belong to $A_p$. 

\begin{example} Consider the informational environment of Example \ref{ex:Partition} with a uniform prior on $\Omega$, and define the event $A = \{2,3\}$. Agent 2 has four partition elements, $\{1,2\}$, $\{3,4\}$, $\{5\}$, $\{6\}$, and assigns to $A$ a posterior probability of $1/2$, $1/2$, $0$, and $0$ (respectively) on these partition elements. So the set of states $A_{1/2}$ at which agent 2 assigns probability $1/2$ to event $A$ being true, is $A_{1/2} = \{1,2,3,4\}$.
\end{example}

The following theorem shows that whenever players' posterior beliefs about an event are common knowledge (e.g., because players have publicly announced these beliefs), then these posterior beliefs must be identical. So disagreement cannot be sustained whenever players' beliefs are commonly known. 

\begin{theorem}[\citet{Aumann}] \label{thm:Aumann} Suppose $\mathcal{I}=\{1,2\}$. Fix any state $ \omega\in\Omega $ and event $A \subseteq \Omega$. If it is common knowledge at $ \omega $ that agent 1 assigns (posterior) probability $q_1$ to event $A$, while agent 2 assigns (posterior) probability $q_2$ to the same event, then $q_1=q_2$.
\end{theorem}

\noindent The result is stated  for two agents, but the proof below directly extends for an arbitrary finite number of players. 
\bigskip

\begin{proof}
	Let $\bf{P}$ be the element of $\Pi_1\meet\Pi_2 $ that contains $ \omega $. Then we can write $ \textbf P = \cup_k \mathcal{P}^k $ where $\mathcal{P}^k $ are elements of $ \Pi_1 $. Since the event  $\{$agent 1's posterior belief is $  q_1 \}$ is common knowledge at $ \omega $, agent 1 must assign probability $  q_1$ to event $A$ at every partition element $ \mathcal{P}^k $. So $ q_1=P(A\cap \mathcal{P}^k)/P(\mathcal{P}^k) $ for each $ k $. This implies $q_1 \cdot P(\mathcal{P}^k) = P(A\cap \mathcal{P}^k) $. Summing over each of player 1's partition elements, we have $ q_1 \sum_k P(\mathcal{P}^k)=\sum_k P(A\cap \mathcal{P}^k) $. Thus  $ q_1\cdot P(\textbf P) = P(A \cap \textbf P)$. But repeating  the same line of logic for player 2, we obtain $ q_2\cdot P(\textbf P) = P(A\cap \textbf P).$ So it must be that $q_1=q_2$.
\end{proof}

\medskip

The following example explains why it is important that players' posterior beliefs are common knowledge and not simply mutual knowledge.

\begin{example}
Let $ \Omega = \{1,2,3,4\} $ with a uniform prior, and define $ \Pi_1 =\{\{1,2\},\{3,4\}\} $ and $ \Pi_2  =\{\{1,2,3\},\{4\}\} $. Choose $A=\{1,4\} $ and $ \omega=2.$ Then agent 1 assigns posterior probability $1/2$ to $A$ while agent 2 assigns  posterior probability $1/3$. Moreover, each agent knows one another's posterior probability. But Theorem \ref{thm:Aumann} is not violated: agent $ 2 $ does not know that agent $ 1 $ knows his posterior probability to be $\frac13 $, so posterior beliefs are mutual knowledge but not common knowledge. \end{example}

The starting hypothesis of Theorem \ref{thm:Aumann}---that individuals have common knowledge of one anothers' beliefs---is strong. \citet{GeanakoplosPolemarchakis} show that the same result obtains under a more realistic process: Communication of posterior beliefs converges to common knowledge of identical posterior beliefs, where this convergence occurs in fewer than $ n_1+n_2 $ steps with $n_i$ the size of agent $i$'s partition. 

\begin{example} Let $\Omega =\{1,2,3,4,5,6,7,8,9\} $ with all states equally likely. There are two agents, Bob and Carly, with information partitions 
\[\Pi_B =\{\{1,2,3\},\{4,5,6\},\{7,8,9\}\}\]
and 
\[\Pi_C =\{\{1,2,3,4\},\{5,6,7,8\},\{9\}\}.\] Suppose the true state is $\omega=1$, and the agents repeatedly communicate their beliefs about the event $A=\{3,4\}$.

 \textbf{Round 1:} Bob's information partition reveals to him that the state belongs to $ \{1,2,3\} $, so he assigns posterior probability $1/3$ to the event $A$. Carly's information partition reveals to her that the state belongs to $ \{1,2,3,4\} $, so she assigns posterior probability $1/2$ to the event $A$. The two agents announce these posterior beliefs. 

\textbf{Round 2:} That Bob assigns probability $1/3$ to $A$ reveals to Carly that Bob was either informed that the state belongs to $\{1,2,3\}$ or informed that the state belongs to $\{4,5,6\}$. But Carly already knew in round 1 that these were the two partition elements that Bob might have been informed of (since she knew the state to be either 1, 2, 3, or 4), and so there is no information for her in this announcement. That Carly assigns probability $1/2$ to $A$ reveals to Bob that agent 2 knows $\{1,2,3,4\}$, but again Bob knew this in round 1. So both agents' posterior beliefs are unchanged. They again announce $1/3$ and $1/2$.
		
And now something interesting happens. That Bob sticks to his original belief of $1/3$ tells Carly that Bob must have observed $\{1,2,3\}$. If  instead Bob observed $ \{4,5,6\} $, then upon hearing that Carly's belief was $1/2$ (and thus learning that Carly observed $ \{1,2,3,4\}$), Bob would have deduced that the state was 4 with certainty, and hence revised his posterior belief of $A$ to 1. So Carly now knows that the state is in $ \{1,2,3\}$ and shares Bob's posterior belief, $ \frac 13 $. The two agents' beliefs have converged, and it is straightforward to show that these beliefs will not move after subsequent communication.
\end{example}

Although agents' beliefs must converge, the  belief that they converge to need not be the belief that agents would have held had they pooled their information:
		
\begin{example} Let $ \Omega = \{1,2,3,4\} $ with each state equally likely. Agents' partitions are given by 
			$ \Pi_1 =\{\{1,2\},\{3,4\}\} $ and 
			$ \Pi_2 =\{\{1,3\},\{2,4\}\} $. 
			Let $ \omega = 1 $ and $ A=\{1,4\} $. Both posteriors are $ 1/2 $ and the process of belief revision converges in one step. But had agents shared their information, they would have learned that $\omega \in \{1,2\}$ and also $\omega \in \{1,3\}$, leading to a (common) posterior belief that $A$ is true with probability 1.
\end{example}

\section{The Email Game} \label{sec:emailgame}

Common knowledge assumptions appear frequently in analyses of strategic environments; for example,  payoffs are assumed to be common knowledge in any complete-information game. Do strategic predictions made under an assumption of common knowledge approximately hold when we relax the assumption of common knowledge? \citet{Rubinstein}'s email game shows that for one formalism of what ``almost common knowledge" means, the answer is no: Strategic predictions can change discontinuously when we move from common knowledge to almost common knowledge. 

In this game, two agents each choose an action from $ \{A,B\}$. There are two possible payoff matrices indexed to $ \{a,b\} $ (depicted below with $a$ on the left and $b$ on the right).  The agents share a common prior that assigns probability  $ 1-p>\frac 12$ to the matrix indexed to $a$.  \\[-12mm]
\begin{center}
		\[\begin{array}{ccc} 
			 & A & B\\
			A& M,M & 0,-L\\
			B&-L,0 & 0,0
		\end{array}
		\hspace{5em}
		\begin{array}{ccc}
			& A & B\\
			A& 0,0 & 0,-L\\
			B&-L,0 & M,M
		\end{array}\]
\end{center}

We assume throughout that $ L>M>0 $. Thus $ (A,A) $ yields higher payoffs for both agents when the payoff parameter is $ a $ while $ (B,B) $ yields higher payoffs when the payoff parameter is $ b $. The action $ A $ is ``safe," in that it never yields a negative payoff.

\medskip

\textbf{Communication Protocol.} Both players have an automated email server, which is the only means by which the players can communicate. Agent $ 1 $ is informed of the payoff parameter. If (and only if) the parameter is $ b $, agent $ 1 $'s email server automatically sends an email to agent $ 2 $ announcing that the parameter is $b$. All emails are independently lost with probability $ \eps>0 $, so the agents' email servers are set up to automatically send back confirmations that emails have been received, and confirmations of confirmations, etc. Each agent $i$'s type is the number of emails that agent $i$'s computer sends, which is privately known to agent $i$.

In the special case $ T_1=T_2=\infty$, there is common knowledge that the parameter is $b$. But if for example $ T_1=2 $, then agent $ 1 $ knows the parameter is $ b $, and agent 1 knows that agent $ 2 $ knows that the parameter is $ b $, but agent $ 1 $ does not know that agent $ 2 $ knows that agent $ 1 $ knows that agent 2 knows that the parameter is $ b $. In general, so long as $T_1$ and $T_2$ are finite, then higher-order knowledge of parameter $b$ must break down at some stage.

\begin{remark} \label{remark:Partition} In the partitional framework of Section \ref{sec:Partition}, we would model this information environment as follows: The state space is
\[ \Omega = \left\{(a,0,0),(b,1,0),(b,1,1),(b,2,1),(b,2,2),\dots \right\}\]
and agents' information partitions are given by
\begin{align*}
\Pi_1 & =\left\{\{(a,0,0)\},\{(b,1,0),(b,1,1)\},\{(b,2,1),(b,2,2)\},\dots\right\} \\
\Pi_2 & =\left\{\{(a,0,0),(b,1,0)\},\{(b,1,1),(b,2,1)\},\dots\right\} 
\end{align*}
where, for example, $T_1=0$ reveals to player 1 the partition element $\{(a,0,0)\}$, while $T_2=0$ reveals to player 2 the partition element $\{(a,0,0),(b,1,0)\}$.
\end{remark}

		 \begin{proposition}
			There is a unique Bayesian Nash equilibrium in which agent $ 1 $ plays $ A $ when the payoff parameter is $ a $. In this equilibrium, both agents play $ A $ independently of the number of messages sent.
		\end{proposition}
		 
\begin{proof} Let $s_i : T_i \rightarrow \Delta(\{A,B\})$ denote player $i$'s equilibrium strategy. By assumption, $ s_1(0)=A $. We will show that also $s_2(0)=A$.  Agent 2 of type $T_2=0$ knows that either 
agent $ 1 $'s first message was never sent (the state is $(a,0,0)$), or agent $ 1 $'s first message was sent but lost (the state is $(b,1,0)$). Unconditionally, the probabilities of these states are  $(1-p)$ and $p\eps$. Conditional on $T_2=0$, agent $ 2 $ assigns a posterior probability of $\frac{1-p}{1-p+p\eps}$ to $(a,0,0)$, a posterior probability of $\frac{p\eps}{1-p+p\eps}$ to $(b,1,0)$ and zero probability to all other states.

So agent 2's expected payoff from playing $ A $ is at least \begin{equation} \label{eq:payoffA}
M \cdot \left(\frac{1-p}{1-p+p\eps}\right) + 0\cdot \left(\frac{p\eps}{1-p+p\eps}\right)
\end{equation}
 while agent 2's expected payoff from playing $ B $ is no more than 
 \begin{equation} \label{eq:payoffB}
 (-L) \cdot \left(\frac{1-p}{1-p+p\eps}\right) + M \cdot \left(\frac{p\eps}{1-p+p\eps}\right).
 \end{equation} Since $ 1-p>\frac 12 $ and $ L>M$ by assumption,  (\ref{eq:payoffA}) strictly exceeds (\ref{eq:payoffB}), and so agent 2's strategy must satisfy $s_2(0)=A $.

Now suppose $ s_i(T_i)=A $ for $ i=1,2 $ and all $ T_i<t $. We'll argue that $s_1(t)=s_2(t)=A$. Suppose first that agent 1's computer sends $t$ emails exactly, i.e., $T_1=t$. Since agent 1's computer did not send a $(t+1)$-th email, it must either be that agent $ 1 $'s $ t $-th message was lost  (the state is $(b,t,t-1)$), or that agent $ 1 $'s $ t $-th message was received, but its  confirmation was lost  (the state is $(b,t,t)$). Agent 2's posterior belief conditional on $T_1=t$ then assigns  probability $ z := \frac{\eps}{\eps+(1-\eps)\eps} >\frac 12 $ to $(b,t,t-1)$ and probability $1-z$ to $(b,t,t)$. So the expected payoff to playing $B$ is $z(-L)+(1-z)(M)<0$, while the payoff to playing $A$ is zero. We conclude that agent 1's strategy must satisfy $ s_1(t)=A $, with nearly identical reasoning yielding $s_2(t) = A$. 
\end{proof}

\medskip

This result shows a sharp discontinuity in strategic predictions at common knowledge. That is, $(B,B)$ is an equilibrium when agents have common knowledge of the payoff parameter $b$, but fails to be an equilibrium when players have knowledge of $b$ to arbitrarily high (finite) orders.

Whether this result is surprising depends on how natural we consider the relaxation of common knowledge to be. Rubinstein (1989) argues that ``high $ T_i $" is intuitively like common knowledge. Another view is that these are substantially different, since for arbitrarily small but strictly positive $\eps$ the informational model is the one described in Remark \ref{remark:Partition}, but for $\eps=0$ (corresponding to common knowledge of the state) the set of states with positive ex-ante probability is $\Omega = \{(a,0,0), (b,\infty,\infty)\}$ and the agents' information partitions are complete. So there is a discontinuity in the informational environments as $\eps \rightarrow 0$, and in this sense small $\eps$ may be quite unlike $\eps=0$. 

\section{(Common) $p$-Belief} \label{sec:pBelief}

We now consider an alternative approach to formalizing almost common knowledge, which defines common ``almost-knowledge" in contrast to the above ``almost-common" knowledge.

		\begin{definition} \label{def:pBelief} For any $p \in [0,1]$, say that
		agent $i$ \emph{$ p$-believes} $ A $ at $ \omega $ if $ P(A\mid \Pi_i(\omega))\geq p $. The set of states at which agent $ i $ $ p $-believes $ A $ is $$ \mathcal{B}_i^p(A) = \{\omega: P(A\mid \Pi_i(\omega))\geq p\}.$$
	\end{definition}

\begin{remark} Is the case $ p=1 $ equivalent to knowledge? Suppose $\Omega = \{1,2,3\}$ and the prior is $P = (0,1/2,1/2)$.  Agent 1's partition is $\{\{1,2\},\{3\}\}$ while agent 2's partition is $\{\{1\},\{2\},\{3\}\}$. The state is $\omega=2$. Then according to Definition \ref{def:Knowledge}, agent 2 knows $\{2\}$ but agent 1 does not, while according to Definition \ref{def:pBelief}, both agents have $1$-belief of $\{2\}$. Whether knowledge and 1-belief represent distinct modes of understanding is an interesting philosophical question, but we will not have more to say on it here.
\end{remark}
	
The following construction of \emph{common $p$-belief}, due to \citet{MondererSamet}, is parallel to Definition \ref{def:CK1} for common knowledge.

\begin{definition}[Common $p$-Belief] \label{def:CommonpBelief1}
For any $p\in [0,1]$ and  event $ A \subseteq \Omega$, define $
		\mathscr{A}^1 = \bigcap_{i\in \mathcal{I}} \mathcal{B}_i^p(A)$ to be the set of states at which every agent $p$-believes $A$ to be true, and recursively define
		$\mathscr{A}^k = \bigcap_{i\in \mathcal{I}} \mathcal{B}^p_i(\mathscr{A}^{k-1})$ for every $k\geq 2$. 
		Then $A$ is common $p$-belief at the set of states $\mathscr{A}^\infty = \cap_{n\geq 1} \mathscr{A}^n$. 
\end{definition}

We can also define common $p$-belief by generalizing the definition of an evident event (Definition \ref{def:Evident}) to events that are evident $p$-belief.

		\begin{definition} 
	For any $p\in [0,1]$, the event $A \subseteq \Omega$ is \emph{evident $p$-belief} if $A \subseteq \bigcap_{i \in \mathcal{I}} \mathcal{B}_i^p(A).$
	\end{definition}

\begin{definition} \label{def:CommonpBelief2} For any $p\in [0,1]$, the event $A \subseteq \Omega$ is common $p$-belief at $\omega$ if there exists an evident $p$-belief event $E$ such that
\[\omega \in E \subseteq \bigcap_{i \in \mathcal{I}} \mathcal{B}_i^p(A).\]
\end{definition}

Definitions \ref{def:CommonpBelief1} and \ref{def:CommonpBelief2} are introduced in \citet{MondererSamet} and shown to be equivalent.

\begin{exercise}[G$^*$]
Consider the email game of Rubinstein (1989). Let $P$ denote the common prior on $\Omega$ (as defined in Remark \ref{remark:Partition}), and define $\mathscr{C}^p$ to be the event that agents have common $p$-belief in parameter $b$. For each $\varepsilon \geq 0$, let
\[ \overline{p}(\varepsilon) = \sup_{p \in [0,1]} \{p: P(\mathscr{C}^p) > 0\}\]
be the supremum of the set of values of $p$ such that $\mathscr{C}^p$ has positive ex-ante probability. Is $\overline{p}(0)$  equal to the limit of $\overline{p}(\varepsilon)$ as $\varepsilon \rightarrow 0$? Discuss your answer.  
\end{exercise}

\section{General State Spaces} \label{sec:General}

To show that the preceding insights do not require assumption of a finite state space, we now briefly discuss two generalizations of these ideas. In each case, we begin with a probability space $(\Omega, \Sigma, P)$ where $\Omega$ is a set of states endowed with $\sigma$-algebra $\Sigma$, and $P: \Sigma \rightarrow [0,1]$ is a probability measure. 

\paragraph{The first generalization.} Let each information partition $\Pi_i$ be a partition of $\Omega$, where we require that each partition element is $\Sigma$-measurable and has strictly positive measure under $P$ (see e.g., \citet{MondererSamet}). All of the above definitions and proofs generalize as stated.

\paragraph{The second generalization.} Alternatively, we might model each agent $i$'s information as a sub $\sigma$-algebra of $\Sigma$, denoted by $\Pi_i$.\footnote{That is, $\Pi_i$ is a $\sigma$-algebra and $\Pi_i \subseteq \Sigma$.} One foundation for this approach (which we will examine in detail in subsequent chapters) is that each agent $i$ privately observes a random variable $X_i: \Omega \rightarrow \mathbb{R}$ that is measurable with respect to $\Sigma$. In this case, each agent $i$'s $\sigma$-algebra is $\sigma(X_i)$, the $\sigma$-algebra generated by $X_i$, which is indeed coarser than $\Sigma$. 

The definition of knowledge can be extended as follows.

\begin{definition} \label{def:KnowledgeGeneral} Agent $i$ \emph{knows} the event $A \in \Sigma$ to be true at $\omega$ if there exists some $B \in \Pi_i$ such that $\omega \in B \subseteq A$.
\end{definition}

Common knowledge cannot in general be iteratively constructed (\`{a} la Definition \ref{def:CK1}) using this definition of $K_i$, since the set of states at which agent $i$ knows $A$ to be true may not be $\Sigma$-measurable. Nevertheless, similar to Definition \ref{def:CK2}, we can define $\bigwedge_{i \in \mathcal{I}} \Pi_i$ to be the finest common coarsening of the $\sigma$-algebras $\Pi_1, \dots, \Pi_n$, and say that an event $A$ is common knowledge at $\omega$ if there is an element $A$ of $\bigwedge_{i \in \mathcal{I}} \Pi_i$ such that $\omega \in A$. We can also generalize Definition \ref{def:CK3} as follows:

\begin{definition}  The event $A \in \Sigma$ is \emph{evident} if $A \in \Pi_i$ for every $i \in \mathcal{I}$, i.e., $A$ belongs to every agent's $\sigma$-algebra.
\end{definition} 

\begin{definition}  The event $A \in \Sigma$ is common knowledge at state $\omega$ if there is an evident event $E$ such that $\omega \in E$ and $E \subseteq A$.
\end{definition}

Theorem \ref{thm:Aumann} can also be generalized, although the previous proof does not extend (for example, there is no longer guaranteed to be a unique element of $\Pi_1 \meet \Pi_2$ that contains $\omega$).

\begin{proposition} \label{prop:AumannGeneral} Let $X \in \mathcal{L}^1(\Omega,\Sigma,P)$, and define $Y = \mathbb{E}(X \mid \Pi_1)$, $Z=\mathbb{E}(X \mid \Pi_2)$. If it is common knowledge that $Y=y$ and $Z=z$ at a state $\omega$ with strictly positive probability, then it must be that $y=z$.
\end{proposition}

\begin{proof} If it is common knowledge that $Y=y$ and $Z=z$, there must exist an event $E \in \Pi_1 \cap \Pi_2$, where $Y$ takes the constant value $y$ on $E$, and $Z$ takes the constant value $z$ on $E$. Let $\mathbbm{1}_E$ denote the indicator variable that takes  value 1 on  $E$. Then
\begin{align*}
y \cdot P(E) & = \mathbb{E}(Y \mathbbm{1}_E) \\
	& = \mathbb{E}(X \mathbbm{1}_E) \\
	& = \mathbb{E}(Z \mathbbm{1}_E) = z \cdot P(E)
	\end{align*}
	using in the second and third equalities that $Y$ and $Z$ are conditional expectations of $X$. Since $P(E) > 0$ (by assumption that $\omega \in E$ has strictly positive probability), it follows that $y=z$ as desired.
\end{proof}

\medskip

This result is in fact more general than Theorem \ref{thm:Aumann}, nesting the previous result as a special case when we choose $X$ to be an indicator function on some set. 

\begin{exercise}[G$^*$] Generalize Proposition \ref{prop:AumannGeneral} by demonstrating that the conclusion still holds if we assume that there is a measurable set of states $B \subseteq \Omega$ with strictly positive probability, where at every $\omega \in B$ it is common knowledge that $Y=y$ and $Z=z$.
\end{exercise}

\section{Additional Exercises}

\begin{exercise}[G] Two spies in an underground organization are stationed at remote locations. Each spy privately observes whether the coast is clear at their location. The spies share a common prior that the coast is clear at each location independently with probability $1/2$. 

\textbf{Communication protocol.} The spies communicate by email with a third party electronic server at their home base.  If and only if the coast is clear at a spy's location, that spy's computer will automatically send a message to the home base with the information that the coast is clear. 

If the home base electronic server receives information from both spies indicating that the coast is clear, then it will automatically send a message to both spies indicating that it has received both messages. (Otherwise, it will send no messages.) As these are dangerous times, each message has only a $1-\varepsilon$ chance of being received (again independent). If either spy receives a message from the home base, that spy will send a reply to the home base confirming receipt. The reply is lost with probability $\varepsilon$, independently of everything that's happened before. So on and so forth. Everything stated above is common knowledge.

Each spy observes the number of messages he has sent, and chooses an action in $\{A,B\}$. If the coast is clear at both locations, then payoffs are given by the \textbf{right} matrix below, and otherwise payoffs are given by the \textbf{left} matrix below. 
\[\begin{array}{ccc} 
			 & A & B\\
			A& M,M & 0,-L\\
			B&-L,0 & 0,0
		\end{array}
		\hspace{5em}
		\setlength{\tabcolsep}{6pt}\renewcommand{\arraystretch}{1.2}
		\begin{array}{ccc}
			& A & B\\
			A& 0,0 & 0,-L\\
			B&-L,0 & M,M
		\end{array}\]
The payoff parameters satisfy $L > 3M > 0$.

\begin{itemize}
\item[(a)] Prove the following analogue of Rubinstein (1989)'s result: Let $T_1 = \mathbb{Z}_+$ and $T_2=\mathbb{Z}_+$ denote the two players' type spaces. There is a unique pure-strategy equilibrium in which both players choose $A$ when the coast is \textbf{not} clear at their location, i.e. $s_1(0)=s_2(0)=A$. In this equilibrium, players choose $A$ for any number of messages sent, i.e. $s_i(t)=A$ for both players $i$ and all $t\in T_i$.

\item[(b)] Suppose instead that $L=2$ while $M=1$, and demonstrate that the result in Part (a) no longer holds by finding some $\varepsilon>0$ and a pair of strategies $(s_1,s_2)$ that constitute a pure-strategy Bayesian Nash equilibrium, where $s_1(0)=s_2(0)=A$ and $s_i(t)=B$ for some player $i$ and type $t\in T_i$. 
\end{itemize}
\end{exercise}

\begin{exercise}[G$^*$] Let $X \in \mathcal{L}^1(\Omega,\Sigma,P)$, and define $Y = \mathbb{E}(X \mid \Pi_1)$, $Z=\mathbb{E}(X \mid \Pi_2)$. Prove that if it is common knowledge that $Y \in A$ and $Z \in B$ at a state $\omega$ with strictly positive probability, then it must be that $A \cap B \neq \emptyset$.

\end{exercise}

\chapter{Bayesian Updating and Beliefs}

Section \ref{sec:Preliminaries} introduces the canonical Bayesian framework and the definition of a signal. Section \ref{sec:Bayes} reviews Bayes' rule and key properties of Bayesian posteriors. Section \ref{sec:Gaussian} provides closed-form expressions for posterior beliefs in the special case of Bayesian updating to normal signals, with applications.

\section{Preliminaries} \label{sec:Preliminaries}

There is a set of \emph{parameters} $\Theta$ endowed with a $\sigma$-algebra $\Sigma$. An agent has a \emph{prior} $p \in \Delta(\Theta)$, where $\Delta(\Theta)$ denotes the set of $\Sigma$-measurable probability measures on $\Theta$. The prior describes the agent's belief at an ``ex-ante" stage in the absence of any information, where what is ex-ante is understood in the context of a specific model.

The focus of this chapter is the object that we will call an \emph{information structure}, \emph{experiment}, or a \emph{signal}, which can be formalized in either of several ways:
\begin{itemize}
\item[(a)] We can define the signal to be a mapping $\sigma: \Theta \rightarrow \Delta(S)$ from the set of parameters to distributions over  a set of signal realizations $S$. See for example \citet{henrique}. 
\item[(b)] We can define a signal to be an $(S,\mathcal{S})$-valued random variable $X$ on an underlying probability space $(\Omega,\Sigma,P)$, where $\Omega = \Theta \times E$ for some set $E$. For example, we might define the signal to be $X=\theta + \varepsilon$ for an $E$-valued noise term $\varepsilon$ that is independent of $\theta$, as we do in Section \ref{sec:Gaussian}. 

\item[(c)]  We can define a signal $S$ to be a finite partition of  $\Omega = \Theta \times [0,1]$, whose elements are non-empty and measurable with respect to the Lebesgue $\sigma$-algebra on $\Omega$. Conditional on parameter $\theta$, the probability of observing $s \in S$ is the Lebesgue measure of $\{x \in [0,1] \mid (\theta,x) \in s\}$. See for example \citet{FrankelKamenica}. 
\end{itemize}

\begin{remark} It is straightforward to see that the first two formalisms nest one another when all the relevant sets are finite. Suppose we are given a prior $p \in \Delta(\Theta)$ and a signal $\sigma: \Theta \rightarrow \Delta(S)$. Define the expanded state space to be $\Omega = \Theta \times S$ and let $P(\theta,s) = p(\theta) \sigma(s \mid \theta)$. Then the random variable $X: \Omega \rightarrow S$ satisfying $X(\theta,s)=s$ is equivalent to $\sigma$ in the sense that posterior beliefs about $\theta$ are the same whether we condition on the realization of $X$ or the realization of $\sigma(\theta)$. In the other direction, if we start with a random variable $X: \Theta \times E \rightarrow S$ and a distribution $P\in \Delta(\Theta \times E)$, then we can define $\sigma: \Theta \rightarrow \Delta(S)$ to satisfy $\sigma(s \mid \theta) = P(X^{-1}(s) \mid \theta)$. The  formalism in (c)  is a special case of (b), where $E=[0,1]$, the random variable $X: \Omega \rightarrow S$ maps each $\omega$ into the partition element of $S$ to which it belongs, and the probability distribution $P$ is the Lebesgue measure. 
\end{remark}

Example families of signals include:

\begin{example}[\citet{Aumann}'s Partitional Information Structures] For each agent $i$, let $\Pi_i$ be a finite partition of $\Theta$ into measurable elements of strictly positive measure. Index these partition elements to $S=\{1,\dots,n\}$ where $n$ is the size of $\Pi_i$. Then let $\sigma$ map each $\theta$ with probability 1 to the index of the partition element to which $\theta$ belongs.
\end{example}

\begin{example}[Finite Information Structures] Suppose $ \lvert \Theta \rvert, \lvert S \rvert<\infty $. Then we can express $\sigma$ as a $\vert \Theta \vert \times \vert S \vert$ matrix where (1) all entries are nonnegative, and (2) all rows sum to 1. For example, suppose a drug is either good (g) or bad (b). The drug is administered to a patient who is either cured (C) or not (N). The patient is cured with probability $3/4$ if the drug is good and with probability $1/4$ if the drug is bad. Then $ \Theta=\{g,b\} $ and $ S=\{C,N\} $ and the information structure is \\[-4mm]
\[\begin{array}{ccc}
		& C & N \\[-1mm]
		g & 3/4 & 1/4 \\
		b & 1/4 & 3/4\end{array}\]
with each row depicting the probability over the signal realizations in the associated state.
\end{example}

\begin{example}[Gaussian Information] \label{ex:Gaussian} The signal is
			$X = \theta + \eps$,
where $\theta \sim \mathcal{N}(\mu_\theta,\sigma_\theta^2)$, $\eps \sim \mathcal{N}\left(0,\sigma_\eps^2\right)$, and $\theta \indep \eps$. 
\end{example}
		
\section{Posterior Beliefs} \label{sec:Bayes}

\subsection{Bayes' Rule} \label{sec:BayesRule}

The agent updates his prior to the realization of the signal using Bayes' rule.

\begin{definition}[Bayes' Rule, Finite Case]  Suppose $\vert \Theta \vert < \infty$. Fix any distribution $P \in \Delta(\Theta)$ and any events $A, B \subseteq \Theta$ where $P(A),P(B) >0$. Then 
\begin{equation} \label{eq:Bayes}
P(A \mid B)  = \frac{P(B \mid A) P(A)}{P(B)}.
\end{equation}
\end{definition}

\begin{remark} Rather than memorizing this formula, it is easier to remember that by the law of total probability, we can rewrite $P(A \cap B)$ as $P(A \mid B) P(B) $ or as $P(B \mid A) P(A)$, so
\[P(A \mid B) P(B) = P(B \mid A) P(A).\]
Dividing through by $P(B)$ yields (\ref{eq:Bayes}).
\end{remark}

\begin{remark} Applying (\ref{eq:Bayes}) twice for the pairs of events $(A,E)$ and $(B,E)$, we have 
\[\frac{P(A \mid E)}{P(B \mid E)} = \frac{P(E \mid A)}{P(E \mid B)} \cdot \frac{P(A)}{P(B)}\] so the relative conditional probabilities of events $A$ and $B$ is determined by their relative probabilities under the prior, $\frac{P(A)}{P(B)}$, and the \emph{likelihood ratio} of  $E$ under $A$ and $B$, $\frac{P(E \mid A)}{P(E \mid B)}$. \emph{Base-rate neglect} is the tendency to falsely equate $\frac{P(A \mid E)}{P(B \mid E)}$ with $\frac{P(E \mid A)}{P(E \mid B)}$, neglecting the prior distribution. This can lead to compelling but inaccurate statistical conclusions.

For example, suppose $\Omega = \{p,n\}$, where $\omega = p$ indicates that an individual is positive for a medical condition while $\omega=n$ indicates that the individual is negative, with $P(\omega = p) = 0.01$.  Let $X \in \{+,-\}$ be the outcome of a test where $P(X = + \mid \omega = p) = 0.95$ and $P(X=+ \mid \omega = n) = 0.05$. Since the likelihood of observing $X=+$ is much higher when the individual has the condition than when he does not---indeed, the likelihood ratio is $\frac{P(X=+ \mid \omega = p)}{P(X=+ \mid \omega = n)} = 19$---it is tempting to conclude from a positive test result that the individual has the condition. But correctly applying Bayes' rule yields that $\frac{P(\omega = p \mid X=+)}{P(\omega = n \mid X=+)}<1$; that is, even with a positive test it is more likely that the individual is negative for the condition.
\end{remark}

A useful rewriting of Bayes' rule is
\begin{equation} \label{eq:BayesFinite}
P(\theta \mid X=x)  = \frac{P(X=x \mid \theta) P(\theta)}{\sum_{\theta' \in \Theta} P(X=x \mid \theta') P(\theta')} \quad \forall \theta \in \Theta
\end{equation}
where the conditional distribution $P(\cdot \mid X=x)$ is precisely the agent's posterior belief upon observing $X=x$. 

\begin{example} \label{ex:BinaryExample} A drug is either effective ($\theta = A$) or not ($\theta=B$), where the prior probability that the drug is effective is $p \in (0,1)$. The signal is 
\[\begin{array}{ccc}
& a & b\\
A & q & 1-q \\
B & 1-q & q
\end{array}\]
for some $q\in (0,1)$. Then upon observing $a$, the agent assigns to $\theta=A$ a posterior probability of
\[\frac{p q}{p q + (1-p) (1-q) } = \frac{1}{1 + \frac{1-p}{p} \left(\frac{1-q}{q}\right)}\]
which exceeds the prior belief of $p$ if and only if $q \geq \frac{1}{2}$. 
\end{example}

More generally, when $\theta$ and $X$ are (not necessarily finite-valued) random variables with densities $f_\theta$ and $f_X$ and conditional densities $f_{\theta \mid X=x}$ and $f_{X \mid \theta=t}$, then the posterior belief given $X=x$ is 
\begin{equation} \label{eq:BayesContinuous}
f_{\theta \mid X=x} (t)  = \frac{f_{X\mid \theta=t} (x ) f_\theta(t)}{\int_{\theta' \in \Theta} f_{X\mid \theta =t'}(x) f_\theta(t') dt'} \quad \forall t \in \Theta.
\end{equation}
Somewhat more generally, we may suppose that the joint distribution of $(\theta,X)$ is such that for every realization $x$ of $X$, there is a (measurable) function $q_x$ satisfying 
\[q_x(A) = \mathbb{E}(\mathbbm{1}_A \mid X=x) \quad \mbox{ for all events $A \subseteq \Theta$}\]
Then this $q_x$ is the posterior belief.

\subsection{Bayes' Plausibility} \label{sec:BayesPlausibility}

Outside of special cases (such as the one we will cover in Section \ref{sec:Gaussian}), posterior beliefs often cannot be expressed in closed-form. Nevertheless, there are certain properties they must satisfy. One important property is that beliefs are a martingale, i.e., the expected posterior is equal to the prior. Intuitively, if you expect to change your mind given more information, then why haven't you done so already?

\begin{fact}[Beliefs are a martingale.] \label{fact:Martingale} Let $p \in \Delta(\Theta)$ denote the agent's prior belief, and choose any event $A$. Then the posterior probability assigned to this event conditional on the realization of random variable $X$ is $\mathbb{E}(\mathbbm{1}_A \mid X)$. By the law of iterated expectations,
\[\mathbb{E}(\mathbb{E}(\mathbbm{1}_A \mid X)) = \mathbb{E}(\mathbbm{1}_A)\]
so the expected posterior probability of $A$ is equal to the prior probability of $A$. Since the event $A$ was arbitrarily chosen, we can conclude that the expected posterior belief is equal to the prior belief. (In the case of a finite state space $\Theta$, choosing $A=\{\theta\}$ yields $\mathbb{E}(p(\theta \mid X)) = p(\theta)$ for every $\theta$.)
\end{fact}

Since any signal $X$ induces a distribution $\tau \in \Delta(\Delta(\Theta))$ over posterior beliefs, Fact \ref{fact:Martingale} implies that this distribution must average to the prior.

\begin{definition} \label{def:BayesPlausible} Fixing a prior $p\in \Delta(\Theta)$, say that a distribution of posteriors $\tau$ is \emph{Bayes plausible} if
\[\int_{\Delta(\Theta)} q  d\tau(q) = p\]
i.e. the expected posterior is equal to the prior. We'll use
\[\mathcal{T}(p) \equiv \left\{ \tau \in \Delta(\Delta(\Theta)) \, \mid \, \int q  d\tau(q) = p\right\}\]
to denote the set of Bayes plausible posterior distributions given prior $p$.
\end{definition}

\begin{exercise}[U]
The state space is $\Theta = \{\theta_1,\theta_2\}$ and the prior  is $(\mu,1-\mu)$ for some $\mu \in [0,1]$. The signal structure is
\[\begin{array}{cccc}
& s_1 & s_2  \\
\theta_1 & p & 1-p \\
\theta_2 & q & 1-q 
\end{array}\]
where $p,q \in [0,1]$. What is the distribution over posterior beliefs induced by this signal structure? Verify that the expected posterior belief is equal to the prior belief.
\end{exercise}

Not only are we guaranteed that any signal induces a Bayes-plausible distribution over posterior beliefs, but also any Bayes-plausible distribution over posterior beliefs can be induced by some signal.

\begin{definition} For any signal $X \sim P_X$, let $\tau_X \in \Delta(\Delta(\Theta))$ satisfy
$\tau_X(q) = P_X( \{x:q_x =q \})$. Say that $\tau \in \Delta(\Delta(\Theta))$ is \emph{induced by $X$} if $\tau = \tau_X$.
\end{definition}

\begin{proposition} \label{prop:BayesPlausible} Suppose the prior $p$ belongs to the interior of the set $\Delta(\Theta)$. Then every Bayes-plausible distribution $\tau \in \mathcal{T}(p)$ is induced by some signal $X$.
\end{proposition}

The proof (demonstrated in \citet{KamenicaGentzkow} and \citet{ShmayaYariv} among others) proceeds by construction. For any distribution $\tau$, index the distinct posterior beliefs in the support of $\tau$ to be $\{q_x\}_{x\in \mathcal{X}}$, where $\mathcal{X}$ may not be finite. Then define $\sigma: \Theta \rightarrow \Delta(\mathcal{X})$ to satisfy
\begin{equation} \label{eq:ConstructSignal}
\sigma(x \mid \theta) = \frac{q_x(\theta) \tau(q_x)}{p(\theta)}
\end{equation}
We have constructed a signal $\sigma$ whose realizations $x$ are identified with posterior beliefs $q_x$, where the conditional distribution over signal realizations mimics Bayes' rule
$p(x \mid \theta) = \frac{p(\theta \mid x) p(x)}{p(\theta)}$, setting $q_x(\theta) = p(\theta \mid x)$ and $\tau(q_x) = p(x)$ . 
 This is a valid signal structure since
\begin{align*}
 \int_{\mathcal{X}} \sigma(x \mid \theta) dx = \int_{\mathcal{X}} \frac{q_x(\theta) \tau(q_x)}{p(\theta)} dx = 1
 \end{align*}
by (\ref{eq:ConstructSignal}) and the definition of Bayes-plausibility. Moreover,
  \begin{align*}
  \frac{\sigma(x \mid \theta) p(\theta)}{\int_\Theta \sigma(x \mid \theta) p(\theta) d\theta}    & = \frac{\sigma(x \mid \theta) p(\theta)}{\tau(q_x) \int_\Theta q_x(\theta) d\theta}  = \frac{\sigma(x \mid \theta) p(\theta)}{\tau(q_x)} =  q_x(\theta)
  \end{align*}
 so $q_x(\cdot)$ is precisely the posterior belief when updating to the signal $\sigma$.
 
 Thus the probability that the posterior belief is $q_x$ is exactly the probability that the realization of the constructed signal $\sigma$ is $x$, so $\tau$ is induced by $\sigma$ as desired.
 
\begin{exercise}[U] Suppose the prior is over $\Theta = \{\theta_1, \theta_2\}$ is $(1/3,2/3)$. Provide a set $S$ and a signal structure $\sigma: \Theta \rightarrow \Delta(S)$ that induces the belief (0,1) with probability 1/3, and the belief (1/2,1/2) with probability 2/3.
\end{exercise}

Together, Fact \ref{fact:Martingale} and Proposition \ref{prop:BayesPlausible} imply:
\begin{corollary} Fix any prior belief $p \in Int(\Delta(\Theta))$. Then a distribution over posteriors $\tau \in \Delta(\Delta(\Theta))$ is induced by some signal if and only if it is Bayes-plausible, i.e., $\tau \in \mathcal{T}(p)$.
\end{corollary}

\subsection{Application of Bayes' Rule: Incompatibility of Fairness Definitions}

Here we take a detour to demonstrate the power of Bayes' rule. Individuals in a population are each described by a covariate vector  $C \in \mathcal{C}$, a group membership $ G \in \{g_1,g_2\}$, and a type $\theta \in \{0,1\}$. For example, we might interpret $\theta$ as the individual's creditworthiness (whether the individual would pay back a loan if approved), $G$ as a demographic group, and $C$ as the individual's credit history. Across individuals, the random vector $(C,G,\theta)$ is distributed according to $P$, and we use $p_g= P(\theta=1 \mid G=g)$ for the base rate of $\theta=1$ in each group $g$.  A \emph{scoring rule} is any mapping $S: \mathcal{C} \rightarrow \{0,1\}$ that predicts the type given the covariate vector.

\begin{definition}[Equality of False Positives] A scoring rule $S$ has equal false positive rates if
\[P(S=1 \mid \theta=0, G=g_1) = P(S=1 \mid \theta=0, G=g_2)\]
\end{definition}
In words, the probability of being incorrectly assessed to pay back the loan is independent of group membership. Equivalently: $S \indep G \mid \theta=0$, i.e., the score is conditionally independent of group membership given type $\theta=0$.

\begin{definition}[Equality of False Negatives] A scoring rule $S$ has equal false negative rates if
\[P(S=0 \mid \theta=1, G=g_1) = P(S=0 \mid \theta=1, G=g_2)\]
\end{definition}
In words, the probability of being incorrectly assessed to not pay back the loan is independent of group membership. Equivalently: $S \indep G \mid \theta=1$, i.e., the score is conditionally independent of group membership given type $\theta=1$.

\begin{definition}[Calibrated] A score $S$ is calibrated if for each $s\in \{0,1\}$,
\[P(\theta=1 \mid S = s, G=g_1) = P(\theta=1 \mid S=s, G=g_2)\]
\end{definition}
In words, among those assessed to pay back the loan (or, to not pay back the loan), the probability of paying back the loan is independent of group membership. Equivalently: $\theta \indep G \mid S$, i.e.,  type is independent of  group membership conditional on the score.

The following impossibility result demonstrates that (outside of edge cases) these fairness criteria cannot be simultaneously satisfied.

\begin{proposition}[\citet{KMR},\citet{Chouldechova:BD2017}] Suppose $p_{g_1} \neq p_{g_2}$. Then no scoring rule $S$ can simultaneously satisfy calibration, equal false positive rates, and equal false negative rates.
\end{proposition}

\begin{proof} Choose either group $g$ and define $FP_g= P(S=1 \mid \theta=0,G=g)$, $FN_g = P(S=0 \mid \theta=1,G=g)$, and $PPV_g = P(\theta=1 \mid S=1,G=g)$. We'll show that these quantities are related by the following identity:
\begin{equation} \label{eq:Identity}
FP_g = \frac{p_g}{1-p_g} \times \frac{1-PPV_g}{PPV_g} \times  (1-FN_g).
\end{equation}
To simplify notation, let $Q$ denote the joint distribution over $(C,G,S)$ after conditioning on $G=g$. Then, expanding (\ref{eq:Identity}), we have
\[Q(S=1\mid \theta=0) = \frac{Q(\theta=1)}{Q(\theta=0)} \times \frac{Q(\theta=0 \mid S=1)}{Q(\theta=1 \mid S=1)} \times Q(S=1 \mid \theta=1)\]
Multiplying both sides by $Q(\theta = 0)$ and applying Bayes' rule,
\begin{align*}
Q(S=1 , \theta=0) & =  \frac{Q(\theta=0 \mid S=1)}{Q(\theta=1 \mid S=1)} \times Q(S=1, \theta=1)
\end{align*}
Thus, (\ref{eq:Identity}) is equivalent to
\begin{equation} \label{eq:RatioIdentity}
\frac{Q(S=1,\theta=0)}{Q(S=1,\theta=1)} =  \frac{Q(\theta=0 \mid S=1)}{Q(\theta=1 \mid S=1)} 
\end{equation}
Again using Bayes' rule, the RHS can be rewritten
\[\frac{Q(\theta=0 \mid S=1)}{Q(\theta=1 \mid S=1)} = \frac{Q(\theta=0, S=1)/Q(S=1)}{Q(\theta=1, S=1)/Q(S=1)} =  \frac{Q(S=1,\theta=0)}{Q(S=1,\theta=1)}\]
so (\ref{eq:RatioIdentity}) is equivalent to
\[ \frac{Q(S=1,\theta=0)}{Q(S=1,\theta=1)} = \frac{Q(S=1,\theta=0)}{Q(S=1,\theta=1)}\]
and is therefore trivially true.

The identity (\ref{eq:Identity}) holds for both groups $g\in \{g_1,g_2\}$. So if $FP_{g_1}=FP_{g_2}$ (as required by equality of false positive rates), $FN_{g_1}=FN_{g_2}$ (as required by equality of false negative rates), and also $PPV_{g_1}=PPV_{g_2}$ (as required by calibration), it must also hold that $p_{g_1}=p_{g_2}$.
\end{proof}

\section{Gaussian Information} \label{sec:Gaussian}

Gaussian information environments are unusually tractable, since the  posterior belief can be expressed in closed-form. We'll cover the main formulae for Bayesian updating in these environments, and show how these can be used to derive  results in three applications. 

\subsection{Formulae}
We'll start with the simplest case. The state is $\theta \sim \mathcal{N}(\mu ,\sigma_\theta^2)$ and the signal is
$X=\theta + \eps$, where $\eps \sim \mathcal{N}(0,\sigma_\eps^2)$, $\theta \indep \eps$, and $\sigma_\theta^2,\sigma_\eps^2>0$. Then:
 
 \begin{fact} \label{fact:BiVar} The agent's posterior belief about $\theta$ conditional on signal realization $X=x$ is normally distributed with mean 
 \[\mathbb{E}(\theta \mid X=x) = \left(\frac{\sigma_\eps^2}{\sigma_\theta^2+\sigma_\eps^2}\right) \mu + \left(\frac{\sigma_\theta^2}{\sigma_\theta^2+\sigma_\eps^2}\right)x\]
 and variance
 \[Var(\theta \mid X=x) = \frac{\sigma_\theta^2 \sigma_\eps^2}{\sigma_\theta^2+\sigma_\eps^2}.\]
\end{fact}

A key property worth remembering is that the posterior mean is a convex combination of the prior mean $\mu$ and the signal realization $x$, where the weights are proportional to prior precision and signal precision. Additionally, while the posterior mean depends on the signal realization, the posterior variance is a constant.

Fact \ref{fact:BiVar} is also sometimes written as:
\[(\theta \mid X=x) \sim \mathcal{N}\left(\left(\frac{\tau_\theta}{\tau_\theta + \tau_\eps}\right) \mu + \left(\frac{\tau_\eps}{\tau_\theta + \tau_\eps}\right)x\, , \, \frac{1}{\tau_\theta + \tau_\eps} \right)\]
where $\tau_\theta = 1/\sigma_\theta^2$ is the precision of the prior belief and $\tau_\eps = 1/\sigma_\eps^2$ is the precision of the signal. This restatement makes it  apparent that the posterior precision is the sum of the prior precision and signal precision. 

We can use Fact \ref{fact:BiVar} to derive the distribution of the posterior mean.

\begin{exercise}[U]  Let $\theta \sim \mathcal{N}(\mu,1)$ and define two signals
 \begin{align*}
 Y_1 & = \theta + \eps_1 \\
 Y_2 & = \theta + \eps_2
 \end{align*}
where $\theta$, $\eps_1$, and $\eps_2$ are all independent of one another, and $\eps_1 \sim\mathcal{N}(0,1)$ while $\eps_2\sim \mathcal{N}(0,2)$.
\begin{itemize}
\item[(a)] Solve for the conditional distributions $\theta \mid Y_1 = y$ and $\theta \mid Y_2 = y$.
\item[(b)] What are the values of $(\mu,y)$ for which it the case that $\mathbb{E}(\theta \mid Y_1 = y) > \mathbb{E}(\theta \mid Y_2 = y)?$ Provide intuition for the condition you derive.
\end{itemize}
\end{exercise} 

\begin{exercise}[U] Suppose we write the posterior belief as $\mathcal{N}(\hat{\mu}, \hat{\sigma}^2)$, where $\hat{\mu}$ is a random variable that depends on the realization of the signal $X$. Prove that
\[\hat{\mu} \sim \mathcal{N}\left(\mu, \sigma_\theta^2 - \frac{\sigma_\theta^2 \sigma_\eps^2}{\sigma_\theta^2 + \sigma_\eps^2}\right),\]
i.e. the expected posterior mean is the prior mean, and the variance of the posterior mean is equal to the prior variance ($\sigma_\theta^2$), reduced by the posterior variance, $\left(\frac{\sigma_\theta^2 \sigma_\eps^2}{\sigma_\theta^2 + \sigma_\eps^2}\right)$.\end{exercise}

\noindent This characterization implies that the more informative the signal is, the more variable the posterior mean is. 

\begin{remark} More generally (i.e., for $\theta$ and $X$ that are not necessarily normally-distributed), the law of total variance says that
\[\Var(\mathbb{E}[\theta \mid X]) = \Var(\theta) - \mathbb{E}[\Var(\theta \mid X) ]\]
so the variance of the posterior mean is equal to the difference of the prior variance and the expectation of the posterior variance. 
\end{remark}

Similar closed-forms exist for multivariate Gaussian states and signals. Suppose $Z$ is a $1\times K$ vector distributed according to $\mathcal{N}(\mu, \Sigma)$, where $\Sigma$ has full rank. Partition the vector as follows:
\[\left(\begin{array}{c}
Z_1 \\
Z_2 \end{array}\right) \sim \mathcal{N}\left(\left(\begin{array}{c}
\mu_1 \\ \mu_2 
\end{array}\right), \left(\begin{array}{cc}
\Sigma_{11} & \Sigma_{12} \\
\Sigma_{21} & \Sigma_{22}
\end{array}\right)\right)\]

\begin{fact} \label{fact:MultiVar} The conditional distribution of $Z_1$ given $Z_2=z_2$ is $\mathcal{N}(\hat{\mu},\widehat{\Sigma})$ where
\begin{align*}
\hat{\mu} & = \mu_1 + \Sigma_{12} \Sigma_{22}^{-1}(z_2 - \mu_2) \\
\widehat{\Sigma} &= \Sigma_{11} - \Sigma_{12} \Sigma_{22}^{-1} \Sigma_{21}
\end{align*}
\end{fact}
\noindent Again, the posterior mean depends on the signal realization, but the posterior covariance matrix does not.

\begin{example} Let
$\left(\begin{array}{c}
Z_1 \\
Z_2 \end{array}\right) \sim \mathcal{N}\left(\left(\begin{array}{c}
\mu_1 \\ \mu_2 
\end{array}\right), \left(\begin{array}{cc}
\sigma_1^2 & \rho \sigma_1 \sigma_2 \\
\rho \sigma_1 \sigma_2 & \sigma_2^2
\end{array}\right)\right)$.
Then $(Z_1 \mid Z_2 = z_2) \sim \mathcal{N}(\hat{\mu}, \widehat{\Sigma})$ where
\begin{align*}
\hat{\mu} & = \mu_1 + \rho \frac{\sigma_1}{\sigma_2} (z_2 - \mu_2) \\
\widehat{\Sigma} & = \sigma_1^2 (1-\rho^2)
\end{align*}

\end{example}

\begin{exercise}[U] Let $Z_1 = \theta$ and $Z_2 = X$ where $\theta$ and $X$ are as defined at the beginning of this section. Show that Fact \ref{fact:MultiVar} implies Fact \ref{fact:BiVar}.
\end{exercise}

\noindent Sections \ref{sec:CareerConcerns}-\ref{sec:DataSharing} demonstrate three applications of these Bayesian updating formulae.

\subsection{Application 1: Career Concerns} \label{sec:CareerConcerns}

Our first application is solving the two-period version of \citet{Holmstrom:REStud1999} model of career concerns.

There is a single agent and a manager. The agent has a type $\theta \sim \mathcal{N}(\mu, \sigma_\theta^2)$ that is unknown to both the agent and the manager. In period 1, the agent chooses an effort level $a \in \mathbb{R}_+$ at cost $c(a) = \frac12 a^2$. This effort is not observed by the manager. The agent's type and effort jointly determine the realization of a performance signal
\begin{equation} \label{def:PerformanceSignal}
X = \theta + a + \varepsilon
\end{equation}
where $\theta \indep \eps$ and $\varepsilon \sim \mathcal{N}(0,\sigma_\eps^2)$. 
In period 2, the manager observes the realization of $X$ and forms an expectation about the agent's type.
Since the manager does not observe $a$, this expectation is taken with respect to the manager's possibly misspecified perception about the distribution of $X$ (more soon).  The agent receives the manager's expectation of his type.

For arbitrary $a \in \mathbb{R}_+$, write $\mathbb{E}^a(\theta \mid X)$ for the conditional expectation of $\theta$ with respect to $X= \theta+a+\varepsilon$. If the manager expects the agent to choose effort $a^*$ while the agent in fact chooses effort $a$, then the agent's total expected payoff is
\[\mathbb{E}^a[\mathbb{E}^{a^*}(\theta \mid X)] - c(a),\]
where the inner expectation $\mathbb{E}^{a^*}(\theta \mid X)$ is the manager's expectation of the agent's type, and $\mathbb{E}^a[\mathbb{E}^{a^*}(\theta \mid X)] $ is the agent's expectation of the manager's expectation.

\begin{claim} \label{claim:Holmstrom} There is a unique equilibrium in which the agent chooses effort 
$a^* = \frac{\sigma_\theta^2}{\sigma_\theta^2 + \sigma_\eps^2}$.
\end{claim}

\begin{corollary}
Equilibrium effort $a^*$ is decreasing in $\sigma_\eps^2$ (i.e., it is less valuable to manipulate a noisier signal) and is increasing in $\sigma_\theta^2$ (i.e., it is more valuable to manipulate information about a more uncertain unknown).
\end{corollary}

We'll now prove Claim \ref{claim:Holmstrom}. Equilibrium effort $a^*$ must satisfy the first-order condition
\begin{equation} \label{eq:FOC}
\left.\frac{\partial \mathbb{E}^a[\mathbb{E}^{a^*}(\theta \mid X)]}{\partial a}\right|_{a = a^*}  = a^*
\end{equation}
equating the marginal value of increasing effort (over $a^*$) to the marginal cost of increasing effort (over $a^*$). Applying Fact \ref{fact:BiVar}, the manager's expectation of $\theta$ with respect to the de-biased signal $X-a^* = \theta +\eps$ is
\[
\mathbb{E}^{a^*}(\theta \mid X) = \frac{\sigma_\theta^2}{\sigma_\theta^2 + \sigma_\eps^2} (X-a^*) + \frac{\sigma_\eps^2}{\sigma_\theta^2 + \sigma_\eps^2} \mu
\]
The agent's expectation of this expectation (with respect to $X=\theta+a+\eps$) is 
\begin{align*}
\mathbb{E}^a\left[\mathbb{E}^{a^*}(\theta \mid X)\right] 
 = \mu + \frac{\sigma_\theta^2}{\sigma_\theta^2 + \sigma_\eps^2} (a-a^*)
\end{align*}
So (\ref{eq:FOC}) implies that equilibrium effort is $a^* = \frac{\sigma_\theta^2}{\sigma_\theta^2 + \sigma_\eps^2}$. (Uniqueness follows from strict concavity of the agent's payoff function.)

\begin{exercise}[U]
Consider the model described in this section, and set $\sigma_\theta^2 = \sigma_\eps^2 = 1$. 
\begin{itemize}
\item[(a)] Suppose that in addition to the worker's performance signal, the firm (through collection of additional data about the worker's type) is able to separately observe a signal
\[S=\theta + \delta\]
where $\theta$ and $\delta$ are jointly normal and independent, and $\delta \sim \mathcal{N}(0,1)$. The firm's expectation about the worker's type is based both on $S$ as well as on the worker's performance signal $X$ (as defined in (\ref{def:PerformanceSignal})). Solve for the worker's equilibrium action and compare it with the previous solution $a^* = \frac{\sigma_\theta^2}{\sigma_\theta^2 + \sigma_\eps^2}$. Does the worker exert more or less effort in equilibrium? Provide intuition for your result.
\item[(b)] Suppose that the firm instead acquires data that allows it to more accurately monitor the performance shocks that the worker experiences (e.g., whether the worker had a rough day, or had help at work). Formally, the firm observes
\[S=\varepsilon + \delta\]
where $\varepsilon$ and $\delta$ are jointly normal and independent, and $\delta \sim \mathcal{N}(0,1)$.  The firm's expectation about the worker's type is based both on $S$ as well as on the worker's performance signal $X$ (as defined in (\ref{def:PerformanceSignal})).  Solve for the worker's equilibrium action and compare it with the previous solution $a^* = \frac{\sigma_\theta^2}{\sigma_\theta^2 + \sigma_\eps^2}$. Does the worker exert more or less effort in equilibrium? Provide intuition for your result.
\end{itemize}

\end{exercise}

\begin{exercise}[G] Consider a variation on \citet{Holmstrom:REStud1999}'s career concerns model, in which the type $\theta$ and noise term $\eps$ are correlated. Specifically, the type is decomposed as $\theta = \theta_1 + \theta_2$, the signal is $X= \theta + \varepsilon + a$, and we suppose that
 \begin{align*}
 \theta_2 = \alpha \theta_1 + z \\
  \eps = \beta \theta_1 + w
  \end{align*} where
$\alpha,\beta \in \mathbb{R}$ are known constants, and $\theta_1 \sim \mathcal{N}(\mu_\theta, \sigma_\theta^2)$, $z\sim \mathcal{N}(0,\sigma_z^2)$, and $w\sim \mathcal{N}(0,\sigma_w^2)$ are mutually independent and unknown to both the agent and the manager. 
\begin{itemize}
\item[(a)] Solve for equilibrium effort. How does this compare to Claim \ref{claim:Holmstrom} in the special case $\alpha=\beta=0$?
\item[(b)] Suppose $\alpha,\beta>0$. How does equilibrium effort change in the parameters $\alpha$ and $\beta$? Provide intuition.
\end{itemize}
\end{exercise}

\subsection{Application 2: Linear-Quadatic Coordination Games}

Our second application is solving for equilibrium in a two-agent linear-quadratic coordination game \citep{MorrisShin}.

Let $\theta\sim\mathcal{N}(\mu,\sigma_\theta^2)$ be an unknown state. Each agent $i=1,2$ receives a private signal about the state
\[X_i = \theta + \eps_i\]
where $\eps_i \sim \mathcal{N}(0,\sigma_\eps^2)$ is independent of the state and across agents. Each agent chooses an action $a_i \in \mathbb{R}$ given their signal realization $x_i$. Agent $i$'s payoff is
\[U_i(a_1, a_2) = - (1-\beta) (a_i - \theta)^2 - \beta (a_i - a_j)^2\]
where $\beta \in (0,1)$ controls how much the agent cares about matching the state versus matching the other agent's action.

We'll solve for a symmetric linear Bayesian Nash equilibrium $(a^*_1, a^*_2)$ in which each agent's strategy satisfies \begin{equation} \label{eq:Eq}
a_i^*(x_i) = c x_i + \kappa 
\end{equation}
for some constants $c, \kappa \in \mathbb{R}$. Let's first conjecture that such an equilibrium exists. Given agent $j$'s strategy $a_j(x_j) = c x_j + \kappa$, agent $i$'s expected payoff (conditional on $X_i=x_i$) is
\[\mathbb{E}[ -(1-\beta)(a_i - \theta)^2 - \beta(a_i - (cX_j + \kappa))^2 \mid X_i=x_i] \]
Taking a derivative with respect to $a_i$, agent $i$'s best reply is
\[ a_i^*(x_i) = (1-\beta) \mathbb{E}(\theta \mid X_i=x_i) + \beta ( c\mathbb{E}( \theta \mid X_i=x_i) + \kappa ).\]
Plugging in the expression for $\mathbb{E}(\theta \mid X_i = x_i)$ from Fact \ref{fact:BiVar}, and matching coefficients with (\ref{eq:Eq}), we have $
c = \frac{\sigma_\theta^2 (1-\beta)}{\sigma_\eps^2 + \sigma_\theta^2(1-\beta)}$ and $\kappa = \frac{\sigma_\eps^2}{\sigma_\eps^2+\sigma_\theta^2(1-\beta)} \mu$. Thus a symmetric linear equilibrium exists in which each agent $i$ chooses
\begin{equation} \label{eq:EqAction}
a_i^*(x_i) = \frac{\sigma_\theta^2 (1-\beta)}{\sigma_\eps^2 + \sigma_\theta^2(1-\beta)}  x_i +  \frac{\sigma_\eps^2}{\sigma_\eps^2+\sigma_\theta^2(1-\beta)} \mu
\end{equation}
\citet{MorrisShin} further show that this is the unique pure-strategy equilibrium.

Suppose we interpret the common prior $\mathcal{N}(\mu, \sigma_\theta^2)$ as informed by a public signal, where a more informative signal implies a smaller $\sigma_\theta^2$. Then we see from (\ref{eq:EqAction}) that the more informative the public signal is, the less weight agents place on their private signal.

\subsection{Application 3: Data Sharing} \label{sec:DataSharing}
Our final application is an example from \citet{AcemogluMakhdoumiMalekianOzdaglar} regarding why  online platforms don't compensate users for the data that they give up.

There is a single platform and two agents $i=1,2$ with types distributed 
\[\begin{pmatrix} 
\theta_1 \\ \theta_2
\end{pmatrix} \sim \mathcal{N} \left( \begin{pmatrix} 
0 \\ 0
\end{pmatrix}, \begin{pmatrix} 
1 & \rho \\ \rho & 1
\end{pmatrix}\right) \]
Each agent $i$ privately observes the realization of a signal $X_i = \theta_i + \eps_i$, where $\eps_i \sim \mathcal{N}(0,1)$ is independent across agents and independent of both types.

 The platform chooses a payment $p_i$ to offer to each agent $i$ for sharing their data. After receiving these offers, each agent $i$ chooses whether to share $(a_i=1)$ or withhold $(a_i=0)$ their signal realization. Write $X_{\bold{a}}$ for the signals shared under action profile $\bold{a}=(a_1,a_2)$. For example, if $\bold{a}=(1,0)$, then $X_{\bold{a}} = X_1$, while if $\bold{a}=(1,1)$, then $X_{\bold{a}}=(X_1,X_2)$.

Each agent $i$'s payoff is determined by the platform's posterior uncertainty about his type, a privacy parameter $v \in \mathbb{R}_+$, and his payment via
\[u_i(\bold{a},\bold{p})= v \cdot \Var(\theta_i \mid X_{\bold{a}}) + p_i \cdot \mathbbm{1}(a_i = 1). \]
The platform's payoff is $u_P(\bold{a},\bold{p})=-u_1(\bold{a},\bold{p})-u_2(\bold{a},\bold{p})$. So the agents prefer for the platform to be more uncertain about their types, while the platform prefers to be less uncertain.

We'll show that when agent types are sufficiently correlated, i.e., $\rho$ is large, then the platform can induce both agents to share their data at a lower total payment than what is required to induce exactly one agent to share.

Let's first solve for payment vectors $(p_1,p_2)$ given which it is an equilibrium for both agents to share their signals. Suppose agent $j$ chooses to share. Then if agent $i$ does not share, the platform's belief about $\theta_i$ is updated only to $X_j$. Since
\[\begin{pmatrix} 
\theta_i \\ X_j
\end{pmatrix} \sim \mathcal{N} \left( \begin{pmatrix} 
0 \\ 0
\end{pmatrix}, \begin{pmatrix} 
1 & \rho \\ \rho & 2
\end{pmatrix}\right) \]
the platform's posterior variance of $\theta_i$ is $1-\rho^2/2$ (by Fact \ref{fact:MultiVar}). So agent $i$'s payoff is
$
v \cdot \left(1-\rho^2/2\right)
$. 
If agent $i$ does share, then beliefs about $\theta_i$ are further updated to the signal 
$X_i$, and (by Fact \ref{fact:BiVar}) the platform's posterior variance of $\theta_i$ reduces to $\frac{2-\rho^2}{4-\rho^2}$. So agent $i$'s payoff is 
$
v \cdot \left(\frac{2-\rho^2}{4-\rho^2}\right) + p_i.
$ Thus, 
agent $i$'s best reply to $a_j=1$ is to share if and only if
\[p_i \geq v\cdot \left( \frac{(2-\rho^2)^2}{2(4-\rho^2)}\right)\]
and the action profile $(a_1,a_2)=(1,1)$ is an equilibrium if the above display holds for both agents $i$. The minimum total payment is twice the right-hand-side, i.e., $v\cdot \left( \frac{(2-\rho^2)^2}{(4-\rho^2)}\right)$.

Let's now solve for payment vectors $(p_1,p_2)$ given which it is an equilibrium for exactly one agent to share his data. Without loss, fix $a_2=0$. If agent 1 chooses $a_1=0$, then the platform's uncertainty about $\theta_1$ is its prior uncertainty, 1, so agent $1$'s payoff is $v$. If agent $1$  chooses $a_1=1$, then the platform's belief about $\theta_1$ updates to the signal 
$X_1$. Applying Fact \ref{fact:BiVar}, the platform's posterior variance about $\theta_1$ is $1/2$ and so agent $1$'s payoff is $v \cdot (1/2) + p_1.$ Thus, $a_1=1$ is a best reply to $a_2=0$ if and only if $p_1 \geq v/2.$
So the platform can induce (exactly) one agent to share if it offers one agent a payment of at least $v/2$ (which is accepted) and another a payment of no more than $v/2$ (which is rejected), at a total payment of $v/2$.

When $\rho^2 \geq \frac{7-\sqrt{17}}{4} \approx 0.71$, then $v \cdot \left( \frac{(2-\rho^2)^2}{(4-\rho^2)}\right) < v/2$, so the platform pays less to induce two users share than one. Intuitively, each agent's choice to share their data exerts a negative externality on other agent: When both users share, each of their signals is less valuable in view of the signal revealed by the other. Agents paid their marginal value thus receive lower compensation, and in a limiting version of this model with a growing number of agents, the amount of compensation needed to induce all agents to share vanishes to zero.

\section{Additional Exercises}

\begin{exercise}[U] Suppose $\theta \sim \mathcal{N}(0, \sigma_\theta^2)$ and 
\begin{align*}
Y_1 & = \theta + b + \eps_1 \\
Y_2 &  = b + \eps_2
\end{align*}
where $\theta$, $b$, $\eps_1$, and $\eps_2$ are all independent of one another, $b \sim \mathcal{N}(0, \sigma_b^2)$, $\eps_1 \sim \mathcal{N}(0,\sigma_1^2)$, and $\eps_2 \sim \mathcal{N}(0, \sigma_2^2)$. We can interpret $Y_1$ as a biased signal about $\theta$ and $Y_2$ as a signal about the size of the bias.

Your friend says: ``The only value of $Y_1$ and $Y_2$ for learning about $\theta$ is to provide information about the size of $b$. Since $Y_1 - Y_2$ is an unbiased signal about $b$, it is equally valuable to learn the outcome of $Y_1-Y_2$ as it is to learn the pair of signals $(Y_1,Y_2)$."

Show that your friend is wrong: The distribution of $\theta \mid Y_1, Y_2$ is different from the distribution of
$\theta \mid Y_1 - Y_2$. Also provide an intuition explaining to your friend the error in their reasoning.

\end{exercise}

\begin{exercise} Consider the two-player game described in Section \ref{sec:DataSharing}, but suppose that the types are distributed
\[\left(\begin{array}{c}
X_1 \\
X_2 \end{array}\right) \sim \mathcal{N} \left(\left( \begin{array}{c}
0 \\
0
\end{array} \right), \left(\begin{array}{cc}
2 & 1 \\
1 & 1 
\end{array}\right)\right)\]
As in Section \ref{sec:DataSharing}, let $a_1$ and $a_2$ denote the actions of players 1 and 2, where an action of `0' means that the player does not share their data, while `1' means that they do. \begin{itemize}
\item[(a)] Suppose player 2 chooses $a_2 = 0$. What is player 1's payoff from choosing $a_1=0$ and what is player 1's payoff from choosing $a_1 = 1$? Provide a condition that characterizes when it is the case that player 1's best reply is to share data (i.e., $a_1=1$). 
\item[(b)] Suppose player 2 chooses $a_2 = 1$. What is player 1's payoff from choosing $a_1=0$ and what is player 1's payoff from choosing $a_1 = 1$? Provide a condition that characterizes when it is the case that player 1's best reply is to share data (i.e., $a_1=1$).
\item[(c)] Suppose player 1 chooses $a_1 = 0$. What is player 2's payoff from choosing $a_2=0$ and what is player 2's payoff from choosing $a_2 = 1$? Provide a condition that characterizes when it is the case that player 2's best reply is to share data (i.e., $a_2=1$). 
\item[(d)] Suppose player 1 chooses $a_1 = 1$. What is player 1's payoff from choosing $a_2=0$ and what is player 2's payoff from choosing $a_2 = 1$? Provide a condition that characterizes when it is the case that player 2's best reply is to share data (i.e., $a_2=1$). 
\item[(e)] Suppose $v=1$. For what set of values of $(p_1,p_2)$ is it the case that:
\begin{itemize}
\item[(i)] $(a_1,a_2)=(1,1)$ is a Nash equilibrium?
\item[(ii)] $(a_1,a_2)=(1,0)$ is a Nash equilibrium?
\item[(iii)] $(a_1,a_2)=(0,1)$ is a Nash equilibrium?
\item[(iv)] $(a_1,a_2)=(0,0)$ is a Nash equilibrium?
\end{itemize}
\item[(f)] Again let $v=1$. What is the smallest total payment the firm must make to induce an equilibrium where both players share their data? (That is, what is the smallest sum $p_1 + p_2$ such that $(a_1,a_2)=(1,1)$ is an equilibrium given the payment profile $(p_1,p_2)$?) Comment on whether it is the case that one player receives the higher payment, and why this answer makes sense. 
\item[(g)] Again let $v=1$. Suppose there is a firm 1 and firm 2, where firm 1 only interacts with player 1, and firm 2 only interacts with player 2. What is the smallest amount $p_1$ that firm 1 must pay to induce player 1 to choose $a_1=1$? What is the smallest amount $p_2$ that firm 2 must pay to induce player 2 to choose $a_2=2$? Compare the sum of these values $p_1 + p_2$ to your answer from part (f).
\end{itemize}

\end{exercise}

\begin{exercise}[G] \label{ex:Average} Suppose $\theta$ is normally distributed. For each $i=1, \dots, n$, let $X_i = \theta + \eps_i$ where $\eps_i$ is independent of $\theta$, the vector $(\eps_1, \dots, \eps_n)$ is jointly normal, and the signals $X_1, \dots, X_n$ are exchangeable. Define $\overline{X} = \frac1n (X_1 + \dots + X_n)$. Prove that $\theta \mid X_1, \dots, X_n$ is identical in distribution to $\theta \mid \overline{X}$. 

\begin{hint*} Recall that $\mathbb{E}(\theta \mid X)$ minimizes    $\mathbb{E}[ (\hat{\theta} - \theta)^2]$ among all $\sigma(X)$-measurable random variables $\hat{\theta}$.

\end{hint*}

\end{exercise}

\begin{exercise}[G] Consider two processes of social learning about an unknown state $\theta \sim \mathcal{N}(0, 1)$. \\

\textbf{Scenario 1:} At $t=0$, a single agent privately observes the signal 
\[Y = \theta + \delta, \quad \delta \sim \mathcal{N}(0,1/\tau)\]
where $\theta$ and $\delta$ are independent of one another, and the precision $\tau \in \mathbb{R}_+$ is a known constant. The agent chooses an action $y$ and receives the payoff $-\mathbb{E}[(y-\theta)^2]$. At $t=1$, each of $n$ agents, indexed by $i$, privately observes a signal 
\[X_i=\theta + \varepsilon_i, \quad \varepsilon_i \sim \mathcal{N}(0,1)\]
as well as the action $y$ of the first agent. The error terms $\varepsilon_i$ are independent across agents, and independent of $\theta$ and $\delta$. Each agent $i$ from this generation then takes an action $a_i$ to maximize the payoff $-\mathbb{E}[(a_i - \theta)^2]$. At $t=2$, you arrive, observe the actions $(a_1, \dots, a_n)$ of the preceding generation (but \emph{not} the action of the first agent), and choose an action $a^*$ with payoff $-\mathbb{E}[(a^* - \theta)^2]$. \\

\textbf{Scenario 2:} At $t=1$, each of $m$ agents observes a private signal 
\[Z_i=\theta + \eta_i, \quad \eta_i \sim \mathcal{N}(0,1)\]
where the error terms $\eta_i$ are independent across agents and of $\theta$. Each agent $i$ takes an action $b_i$ with payoff $-\mathbb{E}[(b_i - \theta)^2]$. At $t=1$, you arrive, observe the actions $(b_1, \dots, b_m)$ of the preceding generation, and choose an action $a^*$ with payoff $-\mathbb{E}[(a^* - \theta)^2]$. \\

Characterize the function $h(n)$ such that your expected payoff is higher in scenario 1 if and only if $m<h(n)$. As clearly as you can, write out an intuition for this result.

\begin{hint*} Use the fact given in Exercise \ref{ex:Average}.
\end{hint*}

\end{exercise}

\chapter{Properties of Information}

Many economic settings involve an unknown type or quality $\theta$ and a signal $X$ about $\theta$, where both $\theta$ and $X$ are ordered (i.e., there are ``better" qualities $\theta$ and ``higher" signal realizations $X$). In these settings, we might think that higher realizations of $X$  are good news about $\theta$---for example, that higher test scores suggest higher ability or that better reviews for a product suggest higher quality. These positive inferences are not in general justified, requiring assumptions on the joint distribution of $(\theta,X)$.

Section \ref{sec:DefinePD} presents three useful definitions of positive dependence between random variables, which are applied to our motivating problem (inference about $\theta$ from observation of a signal $X$) in Section \ref{sec:Relationship}. Section \ref{sec:LL} presents an example of the kind of counterintuitive result that can obtain when these properties are not imposed on the informational environment.
 
\section{Definitions} \label{sec:DefinePD}

\subsection{Monotone Likelihood Ratio Property}

Consider two random variables $Z$ and $\widetilde{Z}$ with distributions $F$ and $\widetilde{F}$ that admit densities $f$ and $\tilde{f}$. To simplify exposition, all densities mentioned in this chapter are assumed to be everywhere positive.

\begin{definition} \label{def:LRDominance} The distribution $F$ \emph{likelihood-ratio dominates} the distribution $\widetilde{F}$ if\footnote{The assumption that densities are everywhere strictly positive allows us to define the monotone likelihood ratio property in terms of likelihood ratios. More generally, we can consider a distribution $F$ to likelihood-ratio dominate another distribution $\widetilde{F}$ if $f(z)\tilde{f}(z') \geq f(z')\tilde{f}(z)$ for all $z > z'$.} 
\[\frac{f(z)}{f(z')} \geq \frac{\widetilde{f}(z)}{\widetilde{f}(z')} \quad \quad \mbox{for all $z > z'$}\]
\end{definition}
\noindent Intuitively, moving up in the likelihood-ratio dominance order renders higher realizations of $z$ more likely relative to lower realizations.

This definition is often specialized to conditional densities in the following way. Suppose $\theta$ and $X$ are real-valued random vectors defined on the same probability space with densities $f_\theta$ and $f_X$ and conditional densities $f_{\theta \mid X}$ and $f_{X \mid \theta}$.

\begin{definition} \label{def:MLRP} The family of conditional densities $\{f_{X \mid \theta}(\cdot \mid \theta)\}_{\theta \in \Theta}$ have the \emph{monotone likelihood ratio property} (MLRP) if for every $x>x'$ and $\theta > \theta'$, \begin{equation} \label{eq:MLRP}
\frac{f_{X\mid \theta} (x\mid \theta)}{f_{X\mid \theta}(x' \mid \theta)}  \geq \frac{f_{X\mid \theta}(x \mid \theta')}{f_{X\mid \theta}(x' \mid \theta')}.
\end{equation}
If the inequality above holds strictly at every $x>x'$, then we say that $\{f_{X \mid \theta}(\cdot \mid \theta)\}$ have the \emph{strict} monotone likelihood ratio property. 
\end{definition}

\begin{remark} If $\{f_{X \mid \theta}(\cdot \mid \theta)\}$ satisfy MLRP, then $\{f_{\theta \mid X}(\cdot \mid X)\}$ also satisfy MLRP. To see this, observe that by Bayes' rule, (\ref{eq:MLRP}) can be rewritten
\[ \frac{f_{\theta \mid X} (\theta\mid x) f_X(x)}{f_\theta(\theta)} \frac{f_\theta(\theta)}{f_{\theta \mid X}(\theta\mid x') f_X(x')}  \geq \frac{f_{\theta \mid X}(\theta'\mid x)f_X(x)}{f_\theta(\theta')} \frac{f_\theta(\theta')}{f_{\theta \mid X}(\theta'\mid x')f(x')}\]
which simplifies to the condition that $\{f_{\theta \mid X}(\cdot \mid X)\}$ have the monotone likelihood ratio property.
\end{remark}

\medskip

In the special case of an additive signal $X = \theta + \eps$, where $\eps$ is independent of $\theta$ and has density $f_\eps$,
\[\frac{f_{\theta\mid X} ( \theta \mid x)}{ f_{\theta \mid X}( \theta' \mid x)} = \frac{f_\eps(x-\theta)}{f_\eps(x-\theta')}\]
so the MLRP condition in (\ref{eq:MLRP}) becomes
\[\frac{f_\eps(x-\theta)}{f_\eps(x'-\theta)} \geq \frac{f_\eps(x-\theta')}{f_\eps(x'-\theta')} \quad \quad \mbox{for every $x>x'$ and $\theta > \theta'$},\]
i.e., for every $\theta > \theta'$, the function $\frac{f_\eps(x-\theta)}{f_\eps(x-\theta')}$ is nondecreasing in $x$. It turns out that this is precisely the condition that $f_\eps$ is log concave.

\begin{definition} A function $f$ that maps a convex set into the positive reals is \emph{log-concave} if the function $\ln f$ is concave. 
\end{definition}

\begin{proposition}[\citet{SaumardWellner}] A density function $f$ on $\mathbb{R}$ is log-concave if and only if for every $\theta > \theta'$, the ratio
$\frac{f(x-\theta)}{f(x-\theta')}$
is a non-decreasing function of $x$.
\end{proposition}

\noindent Thus, in any model where (1) $X= \theta + \eps$,  (2) the noise term $\eps$ is independent of $\theta$, and (3) $\eps$ has a log-concave density, we can be guaranteed that $\{f_{\theta \mid X}(\cdot \mid x)\}$ has the monotone likelihood ratio property (no matter the distribution of $\theta$).

Many distributions have log-concave densities---for example, normal distributions, exponential distributions, the uniform distribution over any convex set, the logistic distribution, and the extreme value distribution. But others do not---for example, the Pareto distribution and Cauchy distribution. See \citet{SaumardWellner} or \citet{BagnoliBergstrom} for other examples and properties of log-concave distributions.

\subsection{Affiliation} 

Let $Z_1, \dots, Z_n$ be real-valued random variables taking values in $\mathbb{R}^n$ and admitting joint density $f$, which again we'll assume to be everywhere strictly positive. For any $z,z' \in \mathbb{R}^n$, let $z \wedge z'$ (``z meet z') denote the component-wise minimum of $z$ and $z'$, and $z \vee z'$ (``z join z') denote the component-wise maximum, i.e.,
\begin{align*}
z \vee z' & = (\max(z_1,z'_1), \dots, \max(z_n,z'_n)) \\
z \wedge z' & = (\min(z_1,z'_1), \dots, \min(z_n,z'_n))
\end{align*}

\begin{definition} \label{def:Affiliation} The variables $Z_1, \dots, Z_n$ are \emph{affiliated} if 
\begin{equation} \label{eq:Affiliated}
f(z \vee z')f(z \wedge z') \geq f(z)f(z')
\end{equation}
 for all $z,z' \in \mathbb{R}^n.$
\end{definition}

\noindent This condition loosely says that larger realizations of any one variable make larger realizations of the other variables more likely. Figure \ref{fig:Affiliation} depicts this relationship for two binary variables.

\begin{figure}[H]
\begin{center}
\includegraphics[scale=0.4]{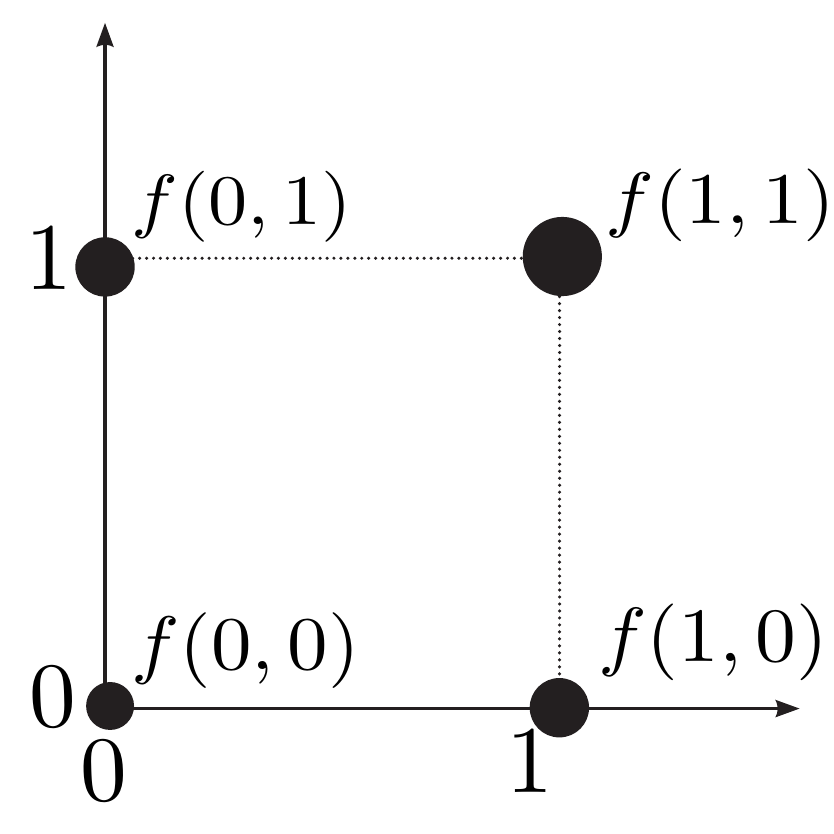}
\end{center}
\caption{Two binary variables with joint density $f$ are affiliated if $f(1,1)f(0,0)\geq f(1,0)f(0,1)$.} \label{fig:Affiliation}
\end{figure}

\begin{remark} If $Z_1, \dots, Z_n$ are mutually independent, then they are affiliated.
\end{remark}

Besides Definition \ref{def:Affiliation}, there are several equivalent ways to characterize affiliation. The first follows by taking logs of both sides of (\ref{eq:Affiliated}).

\begin{proposition} $Z_1, Z_2, \dots, Z_n$ are affiliated if and only if $f$ is log-supermodular, i.e.
\[\log f(z \vee z') + \log f(z \wedge z') \geq \log f(z) + \log f(z')\]
for all $z,z'$.
\end{proposition}

\begin{proposition} \label{prop:SecondDerivative} Suppose the joint density $f$ is twice-differentiable. Then $Z_1, Z_2, \dots, Z_n$ are affiliated if and only if $\frac{\partial^2 \log f}{\partial z_i z_j} \geq 0$.
\end{proposition} 

We show the only if direction below, leaving the if direction for an exercise.

\begin{proof}
Without loss let $i=1$ and $j=2$. Choose any $z_1,z'_1, z_2,z_2' \in \mathbb{R}$ where $z_1 > z_1'$ and $z_2 > z_2'$. Suppose $Z_1,Z_2, \dots, Z_n$ are affiliated. Then by definition
\[\log f(z_1, z_2, z_{-12}) - \log f(z_1', z_2, z_{-12}) \geq \log f(z_1, z_2', z_{-12}) - \log f(z_1', z_2', z_{-12})\]
where $z_{-12}$ is shorthand for $(z_3, \dots, z_n)$. Rewrite $z_1$ as $z_1' +\eps$ and divide both sides by $\eps$. Taking the limit as $\eps \rightarrow 0$, we have
\begin{align*}
\lim_{\eps \rightarrow 0} & \left(\frac{\log f(z_1' +\eps, z_2, z_{-12}) - \log f(z_1', z_2, z_{-12})}{\eps}\right) \\
& \quad \quad \quad \quad \quad \geq \lim_{\eps \rightarrow 0} \left(\frac{\log f(z_1'+\eps, z_2', z_{-12}) - \log f(z_1', z_2', z_{-12})}{\eps}\right)
\end{align*}
so $\frac{\partial \log f}{\partial z_1}$ is increasing in $z_2$, as desired. 
\end{proof}

\begin{exercise}[G] Prove the `if' direction of Proposition \ref{prop:SecondDerivative}: If the joint density $f$ is twice-differentiable and satisfies $\frac{\partial^2 \log f}{\partial z_i z_j} \geq 0$, then $Z_1, Z_2, \dots, Z_n$ are affiliated.
\end{exercise}

The next characterization simplifies (\ref{eq:Affiliated}) to a pairwise condition.  Specifically, for any $(Z_i,Z_j)$ and any realization of the remaining variables $Z_{-ij}$, higher realizations of $Z_i$ must imply higher realizations of $Z_j$.

\begin{proposition} \label{prop:Pairwise}
$Z_1, \dots, Z_n$ are affiliated if and only if 
\begin{equation} \label{eq:Pairwise}
f(z_i,z_j,z_{-ij}) f(z_i',z_j',z_{-ij}) \geq f(z'_i,z_j,z_{-ij}) f(z_i,z_j',z_{-ij})
\end{equation}
for every pair of distinct indices $i$, $j$, and every  $z_i > z_i'$, $z_j > z_j'$, and $z_{-ij} \in \mathbb{R}^{n-2}$.
\end{proposition}

\begin{exercise}[G] Prove Proposition \ref{prop:Pairwise}.
\end{exercise}

This pairwise characterization immediately implies the following characterization, which says that $(Z_1, \dots, Z_n)$ are affiliated if and only if for every pair of variables $i,j$, and every realization of $z_{-ij}$, the family of conditional densities $\{f(\cdot \mid z_j, z_{-ij})\}_{z_j \in \mathbb{R}}$ has the monotone-likelihood ratio property.

\begin{proposition} \label{prop:A_MLRP}$Z_1, \dots, Z_n$ are affiliated if and only if
\begin{equation} \label{eq:A_MLRP}
f(z_i \mid z_j, z_{-ij})f(z'_i \mid z'_j, z_{-ij}) \geq f(z_i \mid z'_j, z_{-ij})f(z'_i \mid z_j, z_{-ij})
\end{equation}
 for every pair of distinct indices $i$, $j$, and every $z_i > z_i'$, $z_j > z_j'$, and $z_{-ij} \in \mathbb{R}^{n-2}$.
\end{proposition}

\begin{proof} The displays in (\ref{eq:Pairwise}) and (\ref{eq:A_MLRP}) are equivalent to one another by Bayes' rule, so Proposition \ref{prop:Pairwise} implies Proposition \ref{prop:A_MLRP}. 
\end{proof}

\medskip

Operations that preserve affiliation include:

\begin{proposition}[Monotone Functions] Suppose $Z_1, \dots, Z_n$ are affiliated, and the functions $g_i: \mathbb{R} \rightarrow \mathbb{R}$, $1\leq i \leq n$, are either all nondecreasing or all nonincreasing. Then the variables $g_1(Z_1), \dots, g_n(Z_n)$ are affiliated.
\end{proposition}

\begin{proposition}[Subsets] Suppose $Z_1, \dots, Z_n$ are affiliated and let $A \subseteq \{1, \dots, n\}$ be any subset of these variables. Then the variables $(Z_i)_{i \in A}$ are affiliated.
\end{proposition}

\begin{proposition}[Order Statistics] For each $1\leq i \leq n$, let $z^{(i)}$ denote the $i$-th largest realization among $(z_1, \dots, z_n)$. Then the variables $(Z^{(1)}, \dots, Z^{(n)})$ are affiliated.
\end{proposition}

\begin{exercise}[G] Show that affiliation is not preserved under arbitrary linear combinations of affiliated variables by constructing an example of random variables $Z_1,Z_2,Z_3$ where $(Z_1,Z_2,Z_3)$ are affiliated but $(Z_1+Z_2,Z_3)$ are not.
\end{exercise}

\subsection{First-Order Stochastic Dominance}

Again consider two real-valued random variables, a parameter $\theta$ and a signal $X$, defined on the same probability space with joint distribution $F$. In many applications we may expect a higher signal realization to lead to a higher inference about the unknown parameter. We now formalize `higher inference' as a first-order stochastic dominance shift in the posterior belief. 

\begin{definition} A distribution $F$ \emph{first-order stochastically dominates} $\widetilde{F}$, which we denote by $F \geq_{FOSD} \widetilde{F}$, if 
\[\int u(\theta) dF(\theta) \geq \int u(\theta) d\widetilde{F}(\theta)\]
for every nondecreasing function $u: \mathbb{R} \rightarrow \mathbb{R}$. Equivalently, $F(\theta) \leq \tilde{F}(\theta)$ at every $\theta \in \Theta$.
\end{definition}
\noindent If $u$ is interpreted as a utility function over money, then a monetary gamble distributed according to $F$ is preferred over one distributed according to $\tilde{F}$ by every agent that prefers more money over less, regardless of the specific shape of the agent's utility function. We can use this definition to compare conditional beliefs about $\theta$.

\begin{definition} \label{def:FOSDProperty} Say that $F$ has the \emph{FOSD property} if $F_{\theta \mid X}(\cdot \mid X=x) \geq_{FOSD} F_{\theta \mid X}(\cdot \mid X=x')$ for all $x > x'$.
\end{definition} 

\citet{Milgrom} proposed a closely related property, which is imposed on conditional distributions $F_{X \mid \theta}$ rather than joint distributions $F$. (This is analogous to considering a signal $\sigma: \Theta \rightarrow \Delta(S)$ without fixing a prior on $\Theta$.)

\begin{definition} Say that a signal realization $x$ is \emph{more favorable than}  signal realization $x'$ if for every  prior distribution $F_\theta \in \Delta(\Theta)$, the posterior distribution $F_{\theta \mid X}(\cdot \mid x)$ first-order stochastically dominates the posterior distribution $F_{\theta \mid X}(\cdot \mid x')$.
\end{definition}

That is, $x$ is more favorable than $x'$ if observing the realization $x$ leads to a FOSD-higher posterior belief about $\theta$ (compared to observing $x'$). If $x$ is more favorable than $x'$ for all $x>x'$, then we have a stronger version of the FOSD property (given in (\ref{def:FOSDProperty})) that holds not only for the specific joint distribution $F$, but for all joint distributions $F$ that are generated by $F_{X \mid \theta}$ and some choice of prior $F_\theta$.

\begin{example} Recall that in the normal-updating setting with $\theta \sim \mathcal{N}(\mu, \sigma_\theta^2)$, $X= \theta + \eps$, $\eps \sim \mathcal{N}(0, \sigma_\eps^2)$, and $\theta \indep \eps$, the agent's posterior belief about $\theta$ conditional on $X$ is 
\[\mathcal{N}\left(\frac{\sigma_\theta^2}{\sigma_\eps^2 + \sigma_\theta^2} X + \frac{\sigma_\eps^2}{\sigma_\eps^2 + \sigma_\theta^2} \mu , \frac{\sigma_\eps^2 \sigma_\theta^2}{\sigma_\eps^2 + \sigma_\theta^2}\right).\]
This distribution is increasing (in the FOSD order) in the realization of $X$ for all parameters $\mu$ and $\sigma_\theta^2$. So $x$ is more favorable than $x'$ for every pair $x>x'$. 
\end{example}

\section{How They are Related} \label{sec:Relationship}

Let $\theta$ and $X$ be real-valued random vectors defined on the same probability space. We'll use $F$ to denote their joint distribution, and assume throughout that the densities $f_\theta$ and $f_X$ and conditional densities $f_{\theta \mid X}$ and $f_{X \mid \theta}$ exist. In this setting, our main properties from above are:
\begin{itemize}
\item[] \quad \textbf{A:} $(X,\theta)$ are affiliated.
\item[] \quad \textbf{MLRP:} $\{f_{X \mid \theta}(\cdot \mid \theta)\}$ satisfies MLRP.
\item[] \quad \textbf{FOSD:} For all $x > x'$, $F(\cdot \mid X=x) \geq_{FOSD} F(\cdot \mid X=x')$
\item[] \quad \textbf{MF:} For all $x > x'$, $x$ is more favorable than $x'$
\end{itemize}

\noindent These properties are related in the following way:

\[\mbox{\textbf{(A)}} \quad \Longleftrightarrow \quad \mbox{\textbf{(MLRP)}} \quad  \Longleftrightarrow \quad \mbox{\textbf{(MF)}}\quad  \Longrightarrow \quad \mbox{\textbf{(FOSD)}}\]
where the one-directional implication from (MF) to (FOSD) is strict. See \citet{Castro} for an example of a distribution satisfying (FOSD) but not (MLRP).

\begin{remark} (MLRP) is equivalent to (MF) but strictly stronger than (FOSD). Thus if a  joint distribution $F$ satisfies (MLRP) then it must satisfy (FOSD), but $F$ can satisfy (FOSD) and fail (MLRP). On the other hand, a conditional distribution $F_{X \mid \theta}$ that satisfies (FOSD) for every completion to a joint distribution $F$ (i.e., for every choice of prior $F_\theta$) must also satisfy (MLRP). So "FOSD for every prior" is equivalent to MLRP, while ``FOSD for some prior" is weaker.
\end{remark}

 We've already established the equivalence between (A) and (MLRP) in Proposition \ref{prop:A_MLRP}. Since (FOSD) is necessary for (MF), clearly (MF) implies (FOSD).  The following result proves equivalence of (MLRP) and (MF).

\begin{proposition}[\citet{Milgrom}] $x$ is more favorable than $x'$ if and only if for every $\theta > \theta'$,
\begin{equation} \label{eq:Milgrom_MLRP}
\frac{f_{X\mid \theta} (x\mid \theta)}{f_{X\mid \theta}(x' \mid \theta)}  \geq \frac{f_{X\mid \theta'}(x \mid \theta')}{f_{X\mid \theta'}(x' \mid \theta')}
\end{equation}
\end{proposition}

\begin{proof} 
We will first show that if (\ref{eq:Milgrom_MLRP}) is satisfied at every $\theta > \theta'$, then $x$ must be more favorable than $x'$. Fix any prior $F_\theta$ and parameter $\theta^* \in \Theta$. If $F_\theta(\theta^*) \in \{0,1\}$ then the conclusion is trivially reached. So suppose $F_\theta(\theta^*) \in (0,1)$.

For any $\theta \leq \theta^*$ and $\tilde{\theta}>\theta^*$, (\ref{eq:Milgrom_MLRP}) implies
\[\frac{f (x\mid \tilde{\theta})}{f(x \mid \theta)}  \geq \frac{f(x' \mid \tilde{\theta})}{f(x' \mid \theta)}\]
where we omit subscripts on the densities here and elsewhere in the proof to ease notation.
Integrating over all $\tilde{\theta}$ such that $\tilde{\theta}>\theta^*$ (with respect to the prior distribution $F_\theta$), we obtain
\[\frac{\int_{\tilde{\theta} > \theta^*} f (x\mid \tilde{\theta}) dF_\theta(\tilde{\theta})}{f(x \mid \theta)}  \geq \frac{\int_{\tilde{\theta} > \theta^*}  f(x' \mid \tilde{\theta}) dF_\theta(\tilde{\theta})}{f(x' \mid \theta)}\]
or equivalently
\[\frac{f(x \mid \theta)}{\int_{\tilde{\theta} > \theta^*} f (x\mid \tilde{\theta}) dF_\theta(\tilde{\theta})}  \leq \frac{f(x' \mid \theta)}{\int_{\tilde{\theta} > \theta^*}  f(x' \mid \tilde{\theta}) dF_\theta(\tilde{\theta})}.\]
Integrating over all $\theta$ such that $\theta\leq\theta^*$, we obtain
\[\frac{\int_{\theta \leq \theta^*} f(x \mid \theta) dF_\theta(\theta)}{\int_{\tilde{\theta} > \theta^*} f (x\mid \tilde{\theta}) dF_\theta(\tilde{\theta})}  \leq \frac{\int_{\theta\leq\theta^*} f(x' \mid \theta) dF_\theta(\theta)}{\int_{\tilde{\theta} > \theta^*}  f(x' \mid \tilde{\theta}) dF_\theta(\tilde{\theta})}.\]
Recall that $f(x \mid \theta) f(\theta) = f(\theta \mid x)f(x)$, so the above display implies
\[\frac{\int_{\theta \leq \theta^*} f(\theta \mid x) d\theta}{\int_{\tilde{\theta} > \theta^*} f(\tilde{\theta} \mid x) d\tilde{\theta}}  \leq \frac{\int_{\theta\leq\theta^*} f(\theta \mid x') d\theta}{\int_{\tilde{\theta} > \theta^*}  f(\tilde{\theta}\mid x') d\tilde{\theta}}\]
or more simply
\[\frac{F(\theta^* \mid x)}{1-F(\theta^* \mid x)}  \leq \frac{F(\theta^* \mid x')}{1-F(\theta^* \mid x')} \]
Since $\frac{y}{1-y}$ is a strictly increasing function in $y$, we have
$F(\theta^* \mid x) \leq F(\theta^* \mid x')$
as desired. 
\medskip

In the other direction, we will show that if $x$ is more favorable than $x'$, then (\ref{eq:Milgrom_MLRP}) holds everywere. Consider any two parameter values $\theta > \theta'$, and let $F_\theta$ be a prior distribution supported on these two points with equal probability on each. 

Since by assumption  $x$ is more favorable than $x'$, then $F(\theta' \mid x) \leq F(\theta' \mid x')$, implying
\[\frac{F(\theta'\mid x)}{ 1- F(\theta' \mid x)} \leq  \frac{F(\theta'\mid x')}{ 1- F(\theta'\mid x')}\]
or equivalently
\[\frac{f(\theta'\mid x')}{ f(\theta \mid x')} \geq  \frac{f(\theta'\mid x)}{f(\theta\mid x)}.\]
Applying Bayes' rule again, we can rewrite the above as
$\frac{f(x \mid \theta)}{f(x' \mid \theta)} \geq \frac{f(x \mid \theta')}{f(x' \mid \theta')}$, which is the desired conclusion. \end{proof}

\begin{remark} \citet{Milgrom}'s result is not precisely the proposition above, but instead the equivalence between strict MLRP (as defined in Definition \ref{def:MLRP}) and a definition of ``more favorable" that replaces FOSD with strict FOSD.

Specifically, say that $F$ strictly first-order stochastically dominates $\widetilde{F}$ if $F(\theta) \leq \widetilde{F}(\theta)$ everywhere with strict inequality at some $\theta$. (Equivalently, 
$\int u(\theta) dF(\theta) > \int u(\theta) d\widetilde{F}(\theta)$ for every strictly increasing function $u: \mathbb{R} \rightarrow \mathbb{R}$.)  Say that $x$ is strictly more favorable than $x'$ if for every prior distribution $F_\theta$, the posterior distribution $F_{\theta \mid X}(\cdot \mid x)$ strictly first-order stochastically dominates $F_{\theta \mid X}(\cdot \mid x')$. Then, by substituting strict inequalities in place of weak inequalities in the proof above where appropriate, we can conclude that $\{f_{X\mid \theta}(\cdot \mid \theta)\}$ satisfies strict MLRP if and only if $x$ is strictly more favorable than $x'$.\footnote{Indeed, the same proof demonstrates a stronger (if slightly more cumbersome to state) result: If and only if $\{f_{X\mid \theta}(\cdot \mid \theta)\}$ satisfies strict MLRP, then $F_{\theta \mid X}(\theta \mid x) < F_{\theta \mid X}(\theta \mid x')$ at every $\theta$ such that $0<F(\theta)<1$.}
\end{remark}

We conclude by briefly summarizing other notions of positive dependence and placing the above properties relative to these.

\begin{itemize}
\item[]  \textbf{Positive covariance (C):} $Cov(X,\theta) \geq 0$
\item[] \textbf{Positive quadrant dependence (QD):} $Cov(g(X),h(\theta)) \geq 0$ for all non-decreasing functions $g$ and $h$
\item[]  \textbf{Association (As):} $Cov(g(X,\theta),h(X,\theta))\geq 0$ for all non-decreasing functions $g$ and $h$
\item[] \textbf{Left-Tail Decreasing (LT):} For all $x$, $F_{X \mid \theta}(X \leq x \mid \theta \leq t)$ is non-increasing in $t$, and for all $t$, $F_{\theta \mid X}(\theta \leq t \mid X \leq x)$ is non-increasing in $x$.
\item[]  \textbf{Inverse Hazard Rate Decreasing (IHR):} For all $x$, $F_{X \mid \theta} (x\mid t)/f_{X \mid \theta}(x\mid t)$ is non-increasing in $t$, and for all $t$, $F_{\theta \mid X}(t \mid x)/f_{\theta \mid X}(t\mid x)$ is non-increasing in $x$.
\end{itemize}

\noindent These properties are extensively studied in, for example, \citet{Lehmann}, \citet{Esary}, \citet{Castro}, and Chapter 3 of \citet{Balakrishna2009}. The following chain of implications is summarized in \citet{Castro}:

\begin{theorem} \emph{(A)} $\Longleftrightarrow$ \emph{(MLRP)} $\Longrightarrow$ \emph{(IHR)}
$\Longrightarrow$ \emph{(FOSD)} $\Longrightarrow$ \emph{(LT)} $\Longrightarrow$ \emph{(As)}  $\Longrightarrow$ \emph{(QD)}  $\Longrightarrow$ \emph{(C)}

\end{theorem}

Thus the standard properties of affiliation and MLRP are in fact strong, implying all of the other properties but not in general implied by them. These properties are equivalent to one another in the special case in which the two variables are jointly normal. 

\begin{exercise}[G] Suppose $(X_1, \dots, X_n)$ are jointly normal and exchangeable, where $\sigma^2=\Var(X_i)$ for each $i$, and $\rho=Cov(X_i,X_j)$ for each pair of indices $i,j$. Prove that these variables are  affiliated if and only if $\rho \geq 0$.

\begin{hint*} Use the fact given in Exercise \ref{ex:Average}.
\end{hint*}

\end{exercise}

\section{When These Conditions Fail} \label{sec:LL}

An example from \citet{LagzielLehrer} demonstrates the kind of counterintuitive result that can hold in settings where (A) and (MLRP) fail. 

An editor chooses which papers to publish. Papers have unknown quality graded on a 9-point scale (A+, A, A-, B+, B, B-, C+,C,C-), whose prior distribution is given in Figure \ref{fig:ImpactDistr}.
\begin{figure}[H] 
				\centering
				\includegraphics[scale=.3]{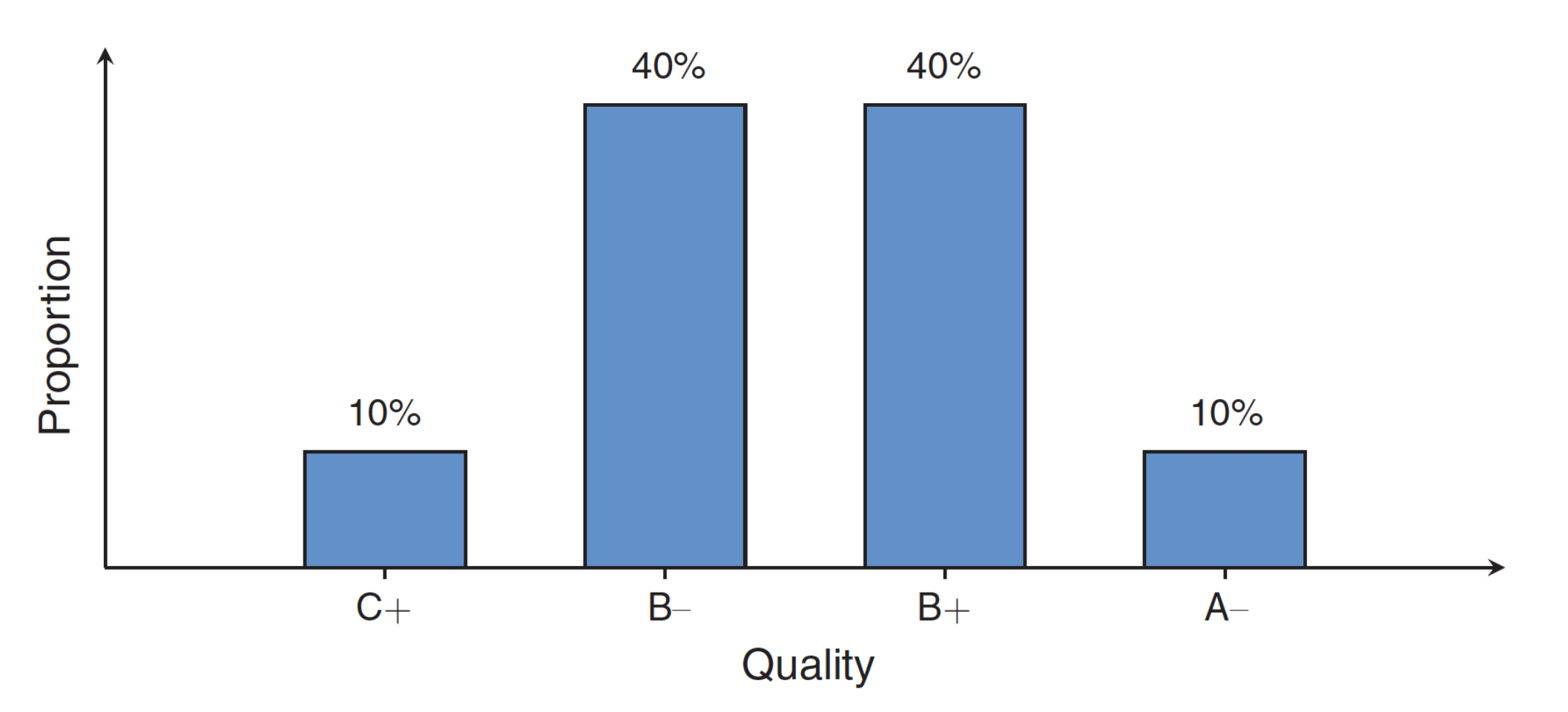}
				\caption{Distribution of Papers' Quality} \label{fig:ImpactDistr}
\end{figure}

The editor learns about quality via a noisy refereeing process, which generates an unbiased signal $X$ about the paper. The realization of $X$ is equal to the true quality with probability 0.8, and otherwise exactly two levels higher or lower than the true quality (each with probability 0.1).
The distribution of $X$ is reported in Figure \ref{fig:XDistr}:
\begin{figure}[H]
				\centering
				\includegraphics[scale=.35]{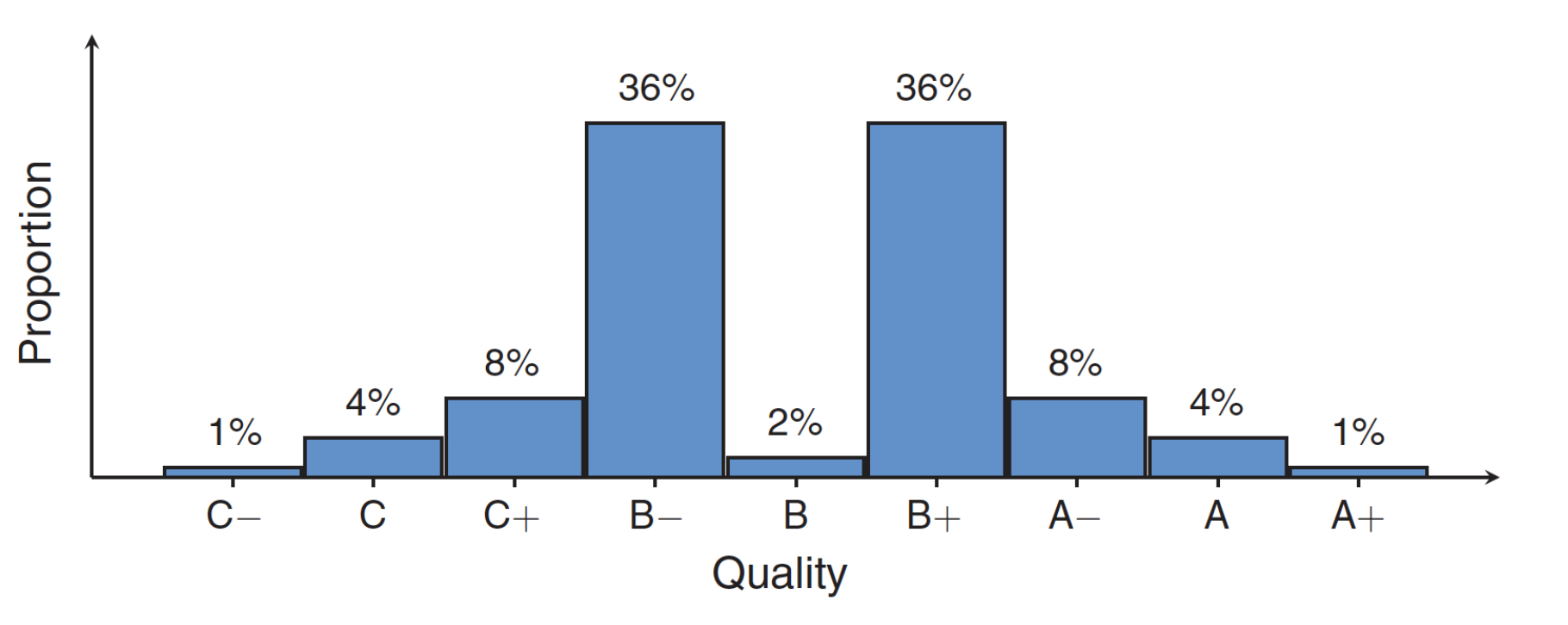}
				\caption{Distribution of Referee Signal} \label{fig:XDistr}
\end{figure}
The editor chooses a threshold and accepts all papers whose expected quality (given the referee's report) exceeds this threshold.  Intuitively, we may expect that the editor faces a tradeoff between publishing more papers versus publishing higher quality papers, where a higher threshold corresponds to publishing fewer but higher quality papers. 

But observe that if the editor chooses to publish only papers with an expected quality that (weakly) exceeds $A$ (i.e., the top-rated 5\% of papers), then the expected value of the published work is close to $B+$. If the editor lowers the bar to $A-$ (i.e., the top-rated 13\%), then the expected value of the published work \emph{increases} to $A-$. Not only are more papers published, but their expected quality is higher. 

In this example, we have $\mathbb{E}(\theta \mid X=x) <\mathbb{E}(\theta \mid X=x')$ even while $x>x'$, so clearly the posterior belief at $x'$ does not first-order stochastically dominate the posterior belief at $x$. \citet{ChambersHealy} demonstrate an even stronger reversal by constructing signals such that the posterior belief at the lower signal realization first-order stochastically dominates the posterior belief at the higher signal realization. Notably, their result relies on natural-seeming signals that satisfy various reasonable properties. 

\begin{theorem}
For every non-degenerate, bounded $\theta$ there exists a signal structure $X$ and two signal realizations $x'>x$ such that $f(\theta \mid X=x')$ is strictly first-order stochastically dominated by $f(\theta \mid X=x)$. Furthermore, $X$ can be chosen to have the following properties: i) $X$ is an additive signal structure, and ii) $e:=X-\theta$ is mean-zero, symmetric, quasiconcave, and has bounded support.
\end{theorem}

\noindent See \citet{Heinsalu} for a strengthening of \citet{LagzielLehrer}'s example using this result, in which lowering the threshold not only increases the expected quality, but results in a quality distribution for published papers that first-order stochastically dominates the one that would obtain at the higher threshold.

\section{Additional Exercises}

\begin{exercise}[G] Let $Z_1, \dots, Z_n$ be affiliated and let $h: \mathbb{R}^n \rightarrow \mathbb{R}$ be any function that is nondecreasing in each of its coordinates. Prove that the function
\[\mathbb{E}(h(Z_1, \dots, Z_n) \mid Z_1 = z_1)\]
is nondecreasing in $z_1$.
\end{exercise}

\begin{exercise}[G] Let $X$ be any real-valued random variable  and let $f: \mathbb{R} \rightarrow \mathbb{R}$ and $g: \mathbb{R} \rightarrow \mathbb{R}$ be bounded nondecreasing functions. Prove that $Cov(f(X),g(X)) \geq 0$. (Do not apply the FKG inequality.)
\begin{hint*} There are at least two short proofs, one that uses Fubini's theorem and the fact that $\mathbb{E}(XY) = \mathbb{E}(X)\mathbb{E}(Y)$ for any independent random variables $X$ and $Y$, and another which relies entirely on elementary (if not obvious) arguments.
\end{hint*}
\end{exercise}

\begin{exercise}[G]
Ann and Bob share the same prior $p$ over an unknown real-valued state $\theta$, and observe a common realization of the signal $X$, but disagree about the distribution of $S$. Ann believes that $X = \theta + \eps$, where $\theta \indep \eps$ and $\eps $ is a real-valued noise term with density $f_\eps$. Bob believes that $X = \theta + \eps + \Delta$ for some $\Delta > 0$. That is, Ann perceives Bob as adding $\Delta$ to the realization of the signal, while Bob perceives Ann as subtracting $\Delta$ from the realization of the signal. 
 
Let $f_A$ denote the joint density of $(\theta, X)$ according to Ann's model and $f_B$ denote the joint density according to Bob's model, with $\mathbb{E}^A$ and $\mathbb{E}^B$ denoting their respective expectation operators. Impose the monotone likelihood-ratio property on $\{f_A(\cdot \mid \omega)\}$, that is,
\[\frac{f_A(x' \mid \theta')}{f_A(x \mid \theta')} \geq \frac{f_A(x' \mid \theta)}{f_A(x \mid \theta)} \quad \forall x'>x, \theta'>\theta\]
\begin{itemize}
\item[(a)] Prove that $\{f_B(\cdot \mid \theta)\}$ also satisfies MLRP. 
\item[(b)] Prove that $\mathbb{E}^A[\mathbb{E}^B[\theta \mid X]]$ is decreasing in $\Delta$, and interpret this result. 
\item[(c)] Suppose that Ann and Bob now additionally observe a common vector of iid signals $(Y_1, Y_2, \dots, Y_N)$ where each $Y_i = \theta + \delta_i$ with $\theta \indep \delta_i$ and $\delta_i$ are iid across signals. Prove that
\[\mathbb{E}^A[\mathbb{E}^B[\theta \mid X, Y_1, \dots, Y_N]] \leq \mathbb{E}^A[\mathbb{E}^B[\theta \mid X, Y_1, \dots, Y_N, Y_{N+1}]]\]
for every $N\geq 1$. Again, interpret the result.
\end{itemize}
\end{exercise}

\chapter{Comparing Information I: The Blackwell Order}

When an agent has access to a choice between multiple signals, we may desire to order these signals based on how informative they are. Intuition can guide us on how to define such an ordering in specific cases, for example:
\begin{itemize}
\item Adding noise to a signal decreases its informativeness.
\item Observing the realization of $(X,Y)$ is more informative than observing the realization of $X$ alone.
\end{itemize}
Any informativeness ordering should satisfy these properties, but there are different ways to generalize from here. One approach is to fix a decision problem and characterize the instrumental value of the signal for that decision problem. Alternatively, we could look for a universal informativeness ordering over signals that holds for all decision problems (as will be the focus of this chapter). Yet another approach is to quantify the ``signal content" contained within the signal based on the physical difficulty of producing or processing that information (see Chapter \ref{sec:CostofInformation}). 

In this chapter we introduce the Blackwell partial order on signals, which considers one signal to be more informative than another if it is more useful for all decision problems. If $\sigma$ dominates $\sigma'$ in this Blackwell order, we will say that  \emph{$\sigma$ is more informative than $\sigma'$} or that \emph{$\sigma$ Blackwell-dominates $\sigma'$}. The following sections demonstrate five perspectives on this order, culminating in Blackwell's theorem (establishing their equivalence) and the proof of this theorem.

\section{Garblings}

We may consider a signal to be more informative than another if the latter is a noised-up version of the former.

\begin{definition}[Markov matrix] A matrix $M$ is a \emph{Markov matrix} if its entries are nonnegative and its rows sum to 1. 
\end{definition}
Recall that when the set of states and the set of signal realizations are finite, we can represent any signal as a Markov matrix.
    \begin{definition}[Garblings, Finite Version] \label{def:Garbling} Markov matrix $ P $ is a \emph{garbling} of  Markov matrix $ Q $  if there exists a Markov matrix $ M $ s.t. $ QM=P $
\end{definition}

\begin{example} \label{ex:Garbling} 
Let $\Theta = \{\theta_1,\theta_2\}$ and consider the signals
			\[P = \left(\begin{array}{cc}
		 3/4 & 1/4 \\
		 1/4 & 3/4\end{array}\right)
		\quad \quad \quad Q = \left(\begin{array}{cccc}
		 9/16 & 3/16 & 3/16 & 1/16 \\
		 1/16 & 3/16 & 3/16 & 9/16 \end{array}\right)
		\]
where as usual the rows are indexed to states and the columns are indexed to signal realizations. Then since
				\[ \underbrace{\left(\begin{array}{cccc}
		 9/16 & 3/16 & 3/16 & 1/16 \\
		 1/16 & 3/16 & 3/16 & 9/16 \end{array}\right)}_{Q}
		 \underbrace{\left(\begin{array}{cc}
			1 & 0 \\  
			1 & 0 \\
			0 & 1 \\
			0 & 1 \\
			\end{array}\right)}_{M} = \underbrace{\left(\begin{array}{cc}
		 3/4 & 1/4 \\
		 1/4 & 3/4\end{array}\right)}_{P}\]
		 where $M$ is a Markov matrix, we can conclude that $P$ is a garbling of $Q$. 
		 
		 This example has a particularly nice intuition. Label the possible realizations of the first information structure $P$ as $s_1$ and $s_2$, and consider the signal which is two independent realizations of $P$.  The set of possible realizations of this new signal is then $\{ s_1s_1, s_1s_2,s_2s_1,s_2s_2\}$ with the conditional distributions over these realizations given precisely by $Q$. So observing $P$ is statistically equivalent to observing $Q$ and forgetting the second realization. Clearly then $P$ is less informative than $Q$.
		 		 \end{example}

\begin{exercise}[U] The state space is $\Theta = \{\theta_1,\theta_2\}$ and two signals are described by the following signal structures
\[P= \left(\begin{array}{cccc}
& s_1 & s_2 & s_3 \\
\theta_1 & 2/3 & 1/3 & 0\\
\theta_2 & 0 & 1/3 & 2/3
\end{array}\right)\]

\[Q = \left(\begin{array}{ccc}
& \tilde{s}_1 & \tilde{s}_2 \\
\theta_1 & 1/6 & 5/6 \\
\theta_2 & 5/6 & 1/6
\end{array}\right)\]
Show that $Q$ is a garbling of $P$. (How does this exercise relate to Example \ref{ex:Garbling}?)
\end{exercise}

\begin{exercise}[G]
For any $q \in [1/2,1]$, define the matrix
\[S_q = \left(\begin{array}{cc}
q & 1-q \\
1-q & q \end{array}\right)\]
Suppose $q,q' \in [1/2,1]$ with $q>q'$. Prove that $S_{q'}$ is a garbling of $S_{q'}$.
\end{exercise}

More generally, we can replace the Markov matrix $M$ in Definition \ref{def:Garbling} with a Markov kernel.

\begin{definition}[Garblings, General Version]
	The signal $ \sigma':\Theta \to\Delta(S') $ is a \emph{garbling} of the signal $ \sigma:\Theta\to\Delta(S) $ if there exists a Markov kernel $ \gamma:S\to \Delta (S') $ such that
	$$ \sigma'(s'\mid \theta) = \int_{s\in S} \gamma (s'\mid s) \sigma(s\mid \theta) ds	$$
	\end{definition}
			 
		 \begin{example} Let $\theta$, $\eps$, and $\delta$ be independent real-valued random variables with densities $f_\theta$, $f_\eps$, and $f_\delta$. Then the signal 
	$X = \theta + \eps + \delta $
	is a garbling of
	$Y = \theta+\eps$, 
since
	\[f_{X \mid \theta}(x \mid t) = \int_{y \in \mathbb{R}} f_\delta (x - y) f_{Y \mid \theta} (y \mid t )  dy\]
where $f_\delta$ is a Markov kernel.
\end{example}

\begin{example} Consider an arbitrary finite set $\Theta$ and let $I$ be the $\vert \Theta \vert \times \vert \Theta \vert$ identity matrix. Then for any set of signal realizations $S$ and any $\vert \Theta \vert \times \vert S \vert$ Markov matrix $Q$, we have $IQ=Q$, so $Q$ is a garbling of $I$.
\end{example}

\begin{exercise}[U] Is it possible for $P$ and $Q$ to both be garblings of one another if $P \neq Q$? Provide an example if so, and otherwise prove that it is not possible.
\end{exercise}

\begin{remark} \label{remark:Garbling} Let $X$ and $X'$ respectively denote the random realizations of the signals $\sigma$ and $\sigma'$. Then $\theta$, $X$, and $X'$ are random variables which can be defined on a common probability space. The property that $\sigma'$ is a garbling of $\sigma$ does not however pin down the joint distribution of $(\theta,X,X')$. What it guarantees is that there is a way of generating these variables such that $\theta$ is independent of $X'$ conditional on $X$, in which case $\theta \mid X$ is identical in distribution to $\theta \mid X, X'$.\footnote{First draw the state $\theta$, then draw $X$ according to its conditional distribution, and finally draw $X'$ according to the garbling kernel $\gamma$, independent of $\theta$.} Other ways of generating these variables---still consistent with $\sigma'$ being a garbling of $\sigma$---can yield different relationships.

For example, suppose $\theta \sim \mathcal{N}(0,1)$ while
\begin{align*}
X&=\theta + \eps_1 \\
X'&=\theta + \eps_2
\end{align*}
where $\eps_1 \sim \mathcal{N}(0,1)$ and $\eps_2 \sim \mathcal{N}(0,2)$ are both independent of $\theta$. Then clearly the latter signal is a garbling of the former. If we further assume that $\eps_2 = \eps_1 + \delta$ where $\delta \sim \mathcal{N}(0,1)$ is an independent noise term, then the following statements are true:
\begin{itemize}
\item $X'$ is independent of $\theta$ conditional on $X$.
\item $X'$ is not independent of $X$ conditional on $\theta$ (since they are further related through the common component $\eps_1$).
\end{itemize}
On the other hand, if we assume that $\eps_1$ and $\eps_2$ are independent, then the statements above are reversed:
\begin{itemize}
\item $X'$ is not independent of $\theta$ conditional on $X$ (since $X'$ provides additional information about $\theta$ beyond what is revealed by $X$).
\item $X'$ is independent of $X$ conditional on $\theta$.
\end{itemize}
Thus in general, the assumption that two signals are related by a garbling does not imply either conditional independence statement given above.
\end{remark}

\section{Decision Problems} \label{sec:DecisionProblem}

Our next two definitions are based on the instrumental value of the signal for decision problems.

\begin{definition} A decision problem is any pair $\bold{D} = (A,u)$ where $A$ is an action set and $u: A \times \Theta \rightarrow \mathbb{R}$ is a payoff function.
\end{definition}

\noindent The full decision problem is described as follows. Fix a prior $p \in \Delta(\Theta)$ and a signal $\sigma: \Theta \rightarrow \Delta(S)$.
\begin{enumerate}
\item The agent chooses a strategy $\alpha: S \rightarrow A$.
\item The state $\theta \sim p$ and signal realization $s \sim \sigma(\cdot \mid \theta) $ are realized, and the agent takes action $\alpha(s)$. The agent's payoff is $u(\alpha(s),\theta)$.
\end{enumerate}

Without the benefit of further information, the best expected payoff the agent can achieve is
\begin{equation} \label{payoff:NoInfo}
\sup_{a \in A} \mathbb{E}\left[u(a,\theta)\right]
\end{equation}
With the benefit of the signal, the agent can achieve an expected payoff of
\begin{equation} \label{payoff:Info}
\sup_{\alpha: S \rightarrow A} \mathbb{E}\left[u(\alpha(s),\theta)\right] = \mathbb{E} \left[ \sup_{a \in A} \mathbb{E}\left[u(a,\theta) \mid s \right]\right]
\end{equation}
where we abuse notation on the LHS by using $s$ to denote the random variable which is the realization of the signal. On the RHS, the inner expectation is with respect to uncertainty about $\theta$ (conditional on the realization of $s$) and the outer expectation is with respect to uncertainty about $s$. One measure of the value of the signal is the difference in these expected payoffs, i.e., 
\begin{align*}
V_{\bold{D},p}(\sigma) \equiv \mathbb{E} \left[ \sup_{a \in A} \mathbb{E}\left[u(a,\theta) \mid s \right]\right]  - \sup_{a \in A} \mathbb{E}\left[u(a,\theta)\right]
\end{align*}
where $\bold{D}=(A,u)$ is the decision problem and $p \in \Delta(\Theta)$ is the agent's prior. 

\begin{remark} It is without loss to assume the use of pure strategies above, but in the subsequent development of the Blackwell order it will be useful to replace $a$ with a mixed strategy $\alpha \in \Delta(A)$ in (\ref{payoff:NoInfo}) and $\alpha$ with a stochastic map $\alpha: S \rightarrow \Delta(A)$ in (\ref{payoff:Info}).
\end{remark}

\begin{example} Suppose $\Theta = \{\theta_1, \theta_2\}$ with a uniform prior $p$. The decision problem is $(A,u)$ where $A = \{a_1,a_2\}$ and the utility function $u: A \times \Theta \rightarrow \mathbb{R}$ assigns a payoff of 1 when the action matches the state, and zero otherwise. The signal $\sigma$ is
\[\begin{array}{ccc}
& s_1 & s_2 \\
\theta_1 & q & 1-q \\
\theta_2 & 1-q & q
\end{array}\]
where $q > 1/2$. Then the agent's ex-ante payoff is maximized by choosing the strategy $\alpha$ that maps $s_1$ to action $a_1$ and $s_2$ to action $a_2$, with an expected payoff of $q$. In the absence of information the agent's best payoff is $1/2$, so $V_{\mathbf{D},p}(\sigma) = q-1/2$.
\end{example}

\begin{example} Suppose $\Theta = \mathbb{R}$ with a prior $\theta \sim \mathcal{N}(0,\sigma_\theta^2)$. The decision problem is $(A,u)$ where $A = \mathbb{R}$ and $u(a,\theta) = -(a-\theta)^2$. The signal is $X = \theta+\eps$ where $\eps \sim \mathcal{N}(0, \sigma_\eps^2)$ is independent of $\theta$. Then the agent's ex-ante payoff is maximized by choosing the strategy $\alpha(x) = \mathbb{E}(\theta \mid X=x)$, with an expected payoff of
\[\mathbb{E}_X\left[ - (\mathbb{E}(\theta \mid X) - \theta)^2\right] = \mathbb{E}_X\left[-\Var(\theta \mid X)\right] = -\frac{\sigma_\theta^2 \sigma_\eps^2}{\sigma_\theta^2 + \sigma_\eps^2}\]
using Fact \ref{fact:BiVar} in the final equality (and in particular, the property that posterior variance is independent of the signal realization). In the absence of information the agent's best payoff is $-\Var(\theta) = -\sigma_\theta^2$, so $V_{\mathbf{D},p}(X) =   \frac{\sigma_\theta^2 \sigma_\eps^2}{\sigma_\theta^2 + \sigma_\eps^2} - \sigma_\theta^2 = \frac{\sigma_\theta^4}{\sigma_\theta^2 + \sigma_\eps^2}$.
\end{example}

In any specific decision problem, a signal that is informative (in the sense of moving the agent's beliefs about $\theta$) may nevertheless have no instrumental value, as demonstrated in the following exercise.

\begin{exercise}[U] \label{ex:Meyer} Suppose $\Theta = \{1,2\}$ and let the prior $p$ assign equal probability to either state. Consider the decision problem $(A,u)$ with $A= \{1,2\}$ and $u(a,\theta)=\mathbbm{1}(a = \theta)$. Let $\sigma_P$ and $\sigma_Q$ respectively be the two signals described by $P$ and $Q$ in Example \ref{ex:Garbling}. Show that $V_{\bold{D},p}(\sigma_P) = V_{\bold{D},p}(\sigma_Q)$. That is, the second independent observation of signal $P$ has no value to the agent over the first.
\end{exercise}

\subsection{Uniformly Better}

We'll say that a signal is more informative than another if it is more useful in every decision problem and for every prior belief.

\begin{definition} \label{def:MoreInformative}
	The signal $ \sigma $ is more informative than $ \sigma' $ if $V_{\bold{D},p}(\sigma) \geq V_{\bold{D},p}(\sigma')$ for every decision problem $\bold{D}$ and every prior $p$. 
\end{definition}

This is a strong condition, and we generally won't be able to order signals in this way. 

\begin{exercise}[U] Let $ \Theta = \{\theta_1,\theta_2,\theta_3\}$ with a uniform prior $p$. Let $A=\{a_1,a_2\} $ and consider two utility functions: Let $u: A \times \Theta \rightarrow \mathbb{R}$ take value 1 if $ (a,\theta)\in \{ (a_1,\theta_1), (a_2, \theta_2), (a_2, \theta_3)\}$, and value 0 otherwise.  Let $u': A \times \theta \rightarrow \mathbb{R}$ take value 1 if $(a,\theta)\in \{ (a_1,\theta_1), (a_2, \theta_2), (a_1, \theta_3)\}$, and value 0 otherwise. Consider the two information structures

\begin{center}

	$\sigma$: \begin{tabular}{ccc}
		& $ s_1 $ & $ s_2 $\\
		$ \theta_1 $& $ 1 $ 	& $ 0 $\\
		$ \theta_2 $& $ 0 $ 	& $ 1 $\\
		$ \theta_3 $& $ 0 $ 	& $ 1 $\\
	\end{tabular} \quad   \quad \quad \quad
	$\sigma'$: \begin{tabular}{ccc}
		& $ s_1 $ & $ s_2 $\\
		$ \theta_1 $& $ 1 $ 	& $ 0 $\\
		$ \theta_2 $& $ 0 $ 	& $ 1 $\\
		$ \theta_3 $& $ 1 $ 	& $ 0 $\\
	\end{tabular}
	\end{center}
Show that $V_{\bold{D},p}(\sigma) > V_{\bold{D},p}(\sigma')$ where $\bold{D}=(A,u)$, but  $V_{\bold{D}',p}(\sigma) > V_{\bold{D}',p}(\sigma')$  where $\bold{D}' = (A,u')$, i.e. the agent prefers the first information given payoffs $u$ and the second given payoffs $u'$.
\end{exercise}

The definition of uniformly better varies both the decision problem and also the prior, but the additional flexibility due to arbitrary priors is not substantial:

\begin{exercise}[G] Prove that if there is a full-support prior $p_0 \in \Delta(\Theta)$ such that
\[V_{\bold{D},p_0}(\sigma) \geq V_{\bold{D},p_0}(\sigma') \quad \mbox{for every decision problem $\bold{D}$}\]
then $\sigma$ is more informative than $\sigma'$.
\end{exercise}

\subsection{Feasible Actions}

Our third definition says that a signal is more informative if observing the realization of the signal allows the agent to more effectively tailor his action to the state. 

\begin{definition} \label{def:Feasible}
	Fix any action set $A$. A conditional distribution over actions $ d:\Theta \to \Delta(A) $ is \emph{feasible under $ \sigma: \Theta \rightarrow \Delta(S)$} if there exists a mapping $\alpha :S\to \Delta(A) $ such that
	$$
	d(a\mid \theta) = \int_{s\in S} \alpha(a\mid s) \sigma(s\mid \theta) ds$$
 We'll use $ \Lambda_\sigma (A)$ to denote the set of all feasible distributions under $\sigma$ given action set $A$.
\end{definition}

When $\sigma$ is a fully revealing signal (e.g., $\sigma: \Theta \rightarrow \Delta(\Theta)$ satisfying $\sigma(\theta \mid \theta) = 1$ for every $\theta$), then every mapping $d: \Theta \rightarrow \Delta(A)$ is feasible under $\sigma$. (Simply set $\alpha=d$.) When $\sigma$ is uninformative---for example, a constant---then $\Lambda_\theta(A)$ consists of all mappings $d: \Theta \rightarrow \Delta(A)$ that take each state into the same distribution over actions. Larger sets $\Lambda_\sigma(A)$ allow the agent more flexibility in tailoring his action to the state, and in this sense are more valuable.

\begin{remark} Observe that $\alpha$ is itself a Markov kernel, so $d$ can be interpreted as a garbling of $\sigma$ where $A$ is the set of signal realizations.
\end{remark}

\section{Dispersion of Posterior Beliefs}

Our final perspective adopts the view on a signal introduced in Section \ref{sec:Bayes}, where a signal is identified with the distribution over posterior beliefs that it induces. We consider the dispersion of these posterior beliefs. Given an uninformative signal, the agent's posterior is deterministically equal to the agent's prior, so there is no dispersion. And if the signal reveals the state directly, then the posterior belief is a point mass on the true state, which ``maximally varies" depending on the realization of the signal. 

We may expect more informative signals to be associated with more dispersed beliefs, but the measure of dispersion is important. For example, using variance to measure dispersion yields a complete order on signals, which cannot possibly be equivalent to the (strict) partial order described in the previous definitions. Below we define two alternative measures for dispersion---mean-preserving spreads and dominance in the convex order---which will turn out to again characterize the previous partial order on signals.

\subsection{Mean-Preserving Spreads}

\begin{definition} A distribution of posterior beliefs $F \in \Delta(\Delta(\Theta))$ is a \emph{mean-preserving spread} of another distribution $\widetilde{F}$ if there exist $\Delta(\Theta)$-valued random variables $Z, \widetilde{Z}$ satisfying the following conditions:
\begin{enumerate}
\item $Z \sim F, \widetilde{Z} \sim \widetilde{F}$
\item $\mathbb{E}(Z \mid \widetilde{Z}) = \widetilde{Z}$ (thus in particular $\mathbb{E}(Z) = \mathbb{E}(\widetilde{Z})$)
\end{enumerate}
\end{definition}

\noindent The name ``mean-preserving spread" reflects that each realization of $\widetilde{Z}$ is spread out into a random $Z$ with the same mean. When $Z$ and $\widetilde{Z}$ are both real-valued, then the second condition can also be stated as
$Z = \widetilde{Z} + \eps$ for some random variable $\eps$ satisfying $\mathbb{E}(\eps \mid \widetilde{Z}) = 0$. 

\begin{example} Consider the two signals
			\[P = \left(\begin{array}{cc}
		 3/4 & 1/4 \\
		 1/4 & 3/4\end{array}\right)
		\quad \quad \quad Q = \left(\begin{array}{cccc}
		 9/16 & 3/16 & 3/16 & 1/16 \\
		 1/16 & 3/16 & 3/16 & 9/16 \end{array}\right)
		\]
		from Example \ref{ex:Garbling}, where the set of states is $\Theta = \{\theta_1,\theta_2\}$. Let the agent's prior be uniform over these states. Then the agent has two possible posterior beliefs after observing $P$, $(3/4,1/4)$ and $(1/4,3/4)$, which are equally likely. We will write the distribution of posterior beliefs as
\[F_P = 1/2 \cdot (3/4,1/4) + 1/2 \cdot (1/4,3/4).\]
Under $Q$, the distribution of posterior beliefs is instead
\[F_Q = 5/16 \cdot (9/10,1/10) + 3/8 \cdot (1/2,1/2) + 5/16 \cdot (1/10,9/10).\]
We will now show that $F_Q$ is a mean-preserving spread of $F_P$. Let $\widetilde{Z}$ be a random variable satisfying $\widetilde{Z} \sim F_P$ and construct the random variable $Z$ given $\widetilde{Z}$ as follows:
\begin{itemize}
\item If $\widetilde{Z} = (1/4,3/4)$ then $Z=(1/10,9/10)$ with probability $5/8$ and $Z=(1/2,1/2)$ with probability $3/8$.
\item If $\widetilde{Z} = (3/4,1/4)$ then $Z = (9/10,1/10)$ with probability $5/8$ and $Z= (1/2,1/2)$ with probability $3/8$.
\end{itemize}
Then $\mathbb{E}(Z \mid \widetilde{Z}) = \widetilde{Z}$ and also $Z \sim F_Q$, so  $F_Q$ is a mean-preserving spread of $F_P$ as desired. This construction is depicted in Figure \ref{fig:MPS}. 
\begin{figure}[H]
\begin{center}
\includegraphics[scale=0.6]{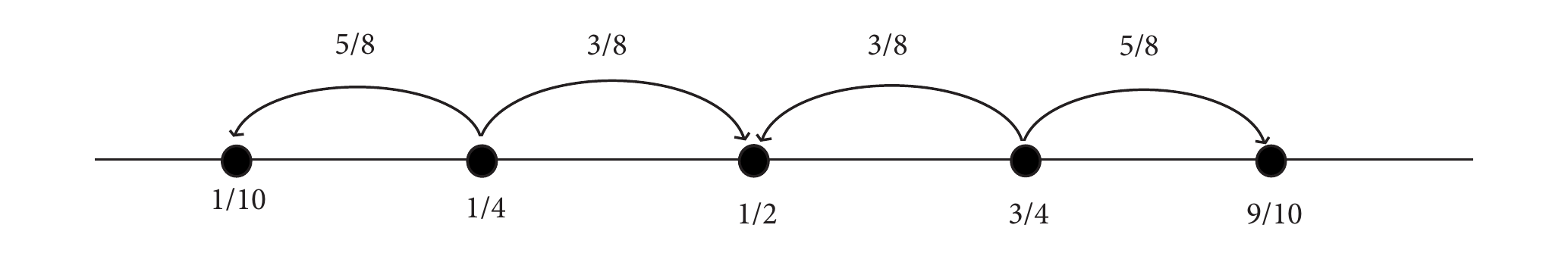}
\end{center}
\caption{Depiction of the mean-preserving spread, where the numbers represent the probability of state $\theta_1$.} \label{fig:MPS}
\end{figure}
\end{example}

\begin{exercise}[G] Let $\Theta = \{\theta_1, \theta_2\}$ and consider the two signals
			\[P = \left(\begin{array}{cc}
		 2/3 & 1/3 \\
		 1/4 & 3/4\end{array}\right)
		\quad \quad \quad Q = \left(\begin{array}{cccc}
		 1/3 & 1/2 & 1/6 \\
		 1/8 & 1/2 & 3/8 \end{array}\right)
		\] Define $F_P$ to be the distribution of posterior beliefs induced by $P$ and $F_Q$ to be the distribution of posterior beliefs induced by $Q$. Prove that $F_P$ is a mean-preserving spread of $F_Q$.
\end{exercise}

\begin{exercise}[G] 
Suppose $Y_1, Y_2, \dots, Y_n$ are independent and identically distributed random variables, and define $\overline{Y}_n = \frac1n \sum_{i=1}^n Y_i$ to be their sample average. Let $n' < n$ and define $\overline{Y}_{n'} = \frac{1}{n'} \sum_{i=1}^{n'} Y_i$. Prove that the distribution of $\overline{Y}_{n'}$ is a mean preserving spread of the distribution of $\overline{Y}_n$.
\end{exercise}

\subsection{Convex Order} \label{sec:ConvexOrder}
\noindent Another partial order of dispersion is the following:

\begin{definition} A distribution of posterior beliefs $ F \in \Delta(\Delta(\Theta))$  \emph{dominates} another distribution $G$ \emph{in the convex order} if for every continuous convex function $h: \Delta (\Theta ) \rightarrow \mathbb{R}$,
\[\int_{\Delta(\Theta)} h(p) dF(p) \geq \int_{\Delta(\Theta)} h(p) dG(p)\]
\end{definition}
\noindent This implies that $F$ and $G$ have the same mean (choosing $h(p)=p$) and that $F$ has the larger variance (choosing $h(p) = \| p \|^2$). 

You may recall the concept of \emph{second order stochastic dominance}:

\begin{definition} For any lotteries $F$ and $G$, $F$ \emph{second-order stochastically dominates} $G$ if and only if 
\[\int_{\Delta(\Theta)} u(p) dF(p) \geq \int_{\Delta(\Theta)} u(p) dG(p)\]
for every nondecreasing and concave utility function $u$.
\end{definition}

Dominance in the convex order is stronger than SOSD. 

\begin{exercise}[G] Prove that if $F$ dominates $G$ in the convex order, then $G$ second order stochastically dominates $F$.
\end{exercise}

The converse is not in general true. 

\begin{example} 
Let $G$ be a distribution uniform on $[1,2]$ and let $F$ be a point mass at zero. Then $G$ second order stochastically dominates $F$ but $F$ does not dominate $G$ in the convex order. 
\end{example}

Intuitively, second-order stochastic dominance confounds changes in the dispersion of the distribution with shifts in the distribution, while dominance in the convex order isolates the former comparison.

\section{Blackwell's Theorem and Proof} \label{sec:BlackwellProof}

We now state and prove \citet{Blackwell}'s theorem, which demonstrates equivalence of these five definitions. For the proof we will work with finite sets (in particular assuming finite $\Theta$) but several parts of the proof extend more generally. 

\begin{theorem}
The following are equivalent:
	\begin{enumerate} 
		\item $ \sigma' $ is a garbling of $ \sigma $. 

				\item $ \sigma $ is more informative than $ \sigma' $.
				
		\item $ \Lambda_\sigma (A) \supseteq \Lambda_{\sigma'} (A) $ for every finite action set $A$.

		\item For any prior on $\Theta$, if we define $F$ and $F'$ to be the distributions of posterior beliefs induced by $\sigma$ and $\sigma'$ (under this prior), then $F$ is a mean-preserving spread of $F'$.
\item For any prior on $\Theta$, if we define $F$ and $F'$ to be the distributions of posterior beliefs induced by $\sigma$ and $\sigma'$ (under this prior), then $F$ dominates $F'$ in the convex order.

	\end{enumerate}

\end{theorem} 

Several proofs exist for different parts of this result (see e.g., \citet{Blackwell} and \citet{LeshnoSpector}). Our proof of the equivalence of (1)-(3) below is based on \citet{henrique}, which presents a particularly simple and elegant argument. 
			
			\bigskip
			
			\begin{proof} Throughout, given  stochastic mappings $ \alpha:X\to \Delta(Y) $ and $ \beta:Y\to \Delta(Z) $, let	\[
		\beta \circ \alpha (z\mid x) \equiv \sum_{y\in Y} \beta(z\mid y)\alpha (y\mid x) \quad \quad \forall x \in X, z \in Z.
		\]

($ 1\Rightarrow 3$) 1 implies existence of a mapping $\gamma: S \rightarrow \Delta(S')$ such that
		$\gamma \circ \sigma = \sigma'$,
		as illustrated below:
		\begin{center}
	
		\begin{Tabular}{l}
			\begin{tikzpicture}
			\matrix (m) [matrix of math nodes,row sep=3em,column sep=4em,minimum width=2em]
			{
				\Theta & S \\
				S' & \\};
			\path[-stealth]
			(m-1-1) edge node [left] {$\sigma'$} (m-2-1)
			edge node [above] {$\sigma$} (m-1-2)

			(m-1-2) edge node [below] {$ \gamma $} (m-2-1);
			\end{tikzpicture}
			
		\end{Tabular}
		\end{center}
\noindent Consider any action set $A$ and mapping $\alpha': S' \rightarrow \Delta(A)$, where
		$d = \alpha' \circ \sigma'$
		is a feasible distribution under $\sigma'$. Define
		$\alpha = \alpha' \circ \gamma$. 
		Then
		\[\alpha \circ \sigma = (\alpha' \circ \gamma) \circ \sigma = \alpha' \circ (\gamma \circ \sigma) = \alpha' \circ \sigma' = d\]
		using  associativity of the operation $\circ$. So $d$ is feasible also under $\sigma$, as depicted in the figure below: 	
		\begin{center}
	
		\begin{Tabular}{l}
			\begin{tikzpicture}
			\matrix (m) [matrix of math nodes,row sep=3em,column sep=4em,minimum width=2em]
			{
				\Theta & S \\
				S' & A \\};
			\path[-stealth]
			(m-1-1) edge node [left] {$\sigma'$} (m-2-1)
			edge node [above] {$\sigma$} (m-1-2)
			(m-2-1.east|-m-2-2) edge node [below] {$\alpha'$} (m-2-2)
			(m-1-2) edge node [right] {$\alpha$} (m-2-2)
			edge node [below] {$ \gamma $} (m-2-1);
			\end{tikzpicture}
			
		\end{Tabular}
		\end{center}
\medskip
	($ 3\Rightarrow 1 $) Let the action set be $S'$ and define $\alpha'$ to be the identity mapping $id_{S'}: S' \rightarrow \Delta(S')$ which satisfies $ id_{S'}(s')=\delta_{s'} $ for all $s' \in S'$ (where $\delta_{s'}$ denotes a point mass at $s'$). By 3, there must exist some $\alpha: S \rightarrow \Delta(S')$ such that
		\[\alpha \circ \sigma =  id_{S'} \circ \sigma'\]
		The RHS reduces to $\sigma'$ since for any $s' \in S'$,
		\[ id_{S'} \circ \sigma'(s' \mid \theta)  = \sum_{s \in S'}  id_{S'}(s' \mid s) \sigma'(s \mid \theta) = \sigma'(s' \mid \theta).\]
		Thus $\alpha \circ \sigma = \sigma'$. But this implies that $\sigma'$ is a garbling of $\sigma$, as depicted below.
		\begin{center}
		\begin{Tabular}{l}
			\begin{tikzpicture}
			\matrix (m) [matrix of math nodes,row sep=3em,column sep=4em,minimum width=2em]
			{
				\Theta & S \\
				S' & S' \\};
			\path[-stealth]
			(m-1-1) edge node [left] {$\sigma'$} (m-2-1)
			edge node [above] {$\sigma$} (m-1-2)
			(m-2-1.east|-m-2-2) edge node [below] {$id_{S'}$} (m-2-2)
			(m-1-2) edge [dashed] node [right] {$\alpha$} (m-2-2);
			\end{tikzpicture}
			
		\end{Tabular}
		\end{center}
		
		\medskip
	($3 \Rightarrow 2$) Clear.
		\medskip
		
	($2 \Rightarrow 3$) Suppose 3 fails. Then there is a finite action set $A$ and a vector $\lambda' \in \Lambda_{\sigma'}(A)$ such that $ \lambda'\not \in \Lambda_\sigma(A)$. The set $ \Lambda_\sigma $ is a compact and convex subset of $ \mathbb{R}^{\vert \Theta \vert \times \vert A \vert} $ (you will be asked to prove this in Exercise \ref{ex:LsA}). Thus by the Separating Hyperplane Theorem, there exists a vector $ v\in \mathbb{R}^{\vert \Theta \vert \times \vert A \vert} $ such that for all $\lambda \in \Lambda_\sigma(A)$,
		\begin{equation} \label{eq:SeparatingHyperplane}
		\sum v(a,\theta)\lambda(a, \theta) < \sum v(a,\theta)\lambda'(a, \theta)
		\end{equation}
		as depicted in Figure \ref{fig:SeparatingHyperplane}.
\begin{figure}[H]
\begin{center}
\includegraphics[scale=0.4]{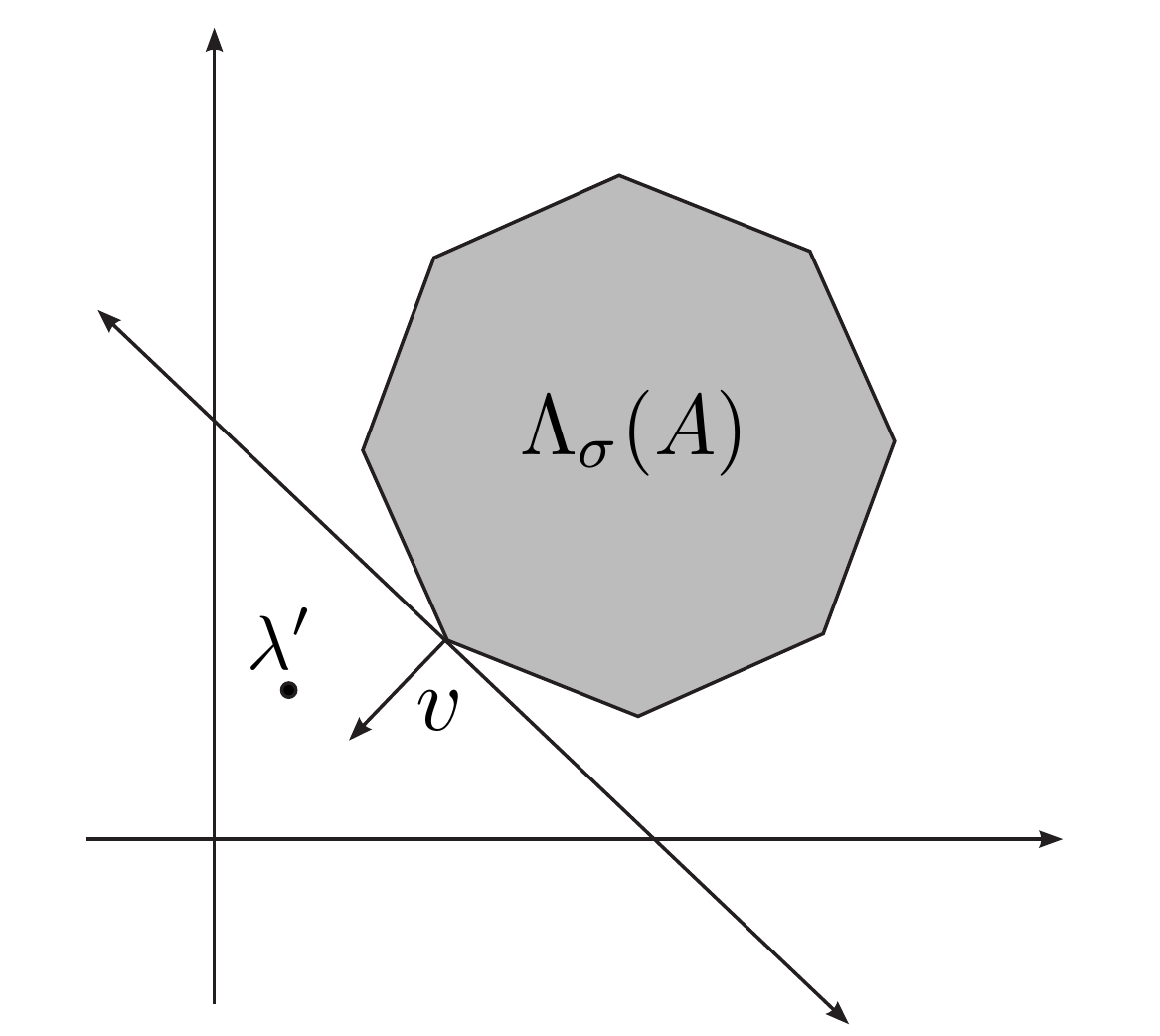}
\end{center}
\caption{Separation of $\lambda'$ from $\Lambda_\sigma(A)$.} \label{fig:SeparatingHyperplane}
\end{figure}
\noindent Consider an agent with a uniform prior $p$ on $\Theta$ and utility function $v$, and define $d(a\mid \theta) \equiv \lambda(a, \theta)$ and $d'(a \mid \theta) \equiv \lambda'(a,\theta)$. Then
\begin{align*}
\sup_{\alpha: S \rightarrow \Delta(A)} \sum_{a,s,\theta} v(a,\theta) \alpha(a \mid s)p(\theta,s)  & = \sup_{\alpha: S \rightarrow \Delta(A)} \frac{1}{\vert \Theta \vert} \sum_{\theta, a,s} \sigma(s \mid \theta) \alpha(a \mid s) v(a,\theta) \\
& = \sup_{d \in \Lambda_\sigma(A)} \frac{1}{\vert \Theta \vert} \sum_{\theta,a} d(a \mid \theta) v(a,\theta) \\
& < \frac{1}{\vert \Theta \vert} \sum_{\theta,a} d'(a \mid \Theta)v(a, \theta)
\end{align*}
using (\ref{eq:SeparatingHyperplane}) in the final inequality. Thus there is a decision problem and a prior for which an agent can achieve a strictly higher payoff by conditioning on $\sigma'$ rather than on $\sigma$, and so 2 fails.

\medskip
($1 \Rightarrow 4$) Let $X$ and $X'$ respectively denote the random realizations of the signals $\sigma$ and $\sigma'$. Since by assumption $\sigma'$ is a garbling of $\sigma$, we can generate $\theta$, $X$, $X'$ in a way such that $X'$ is independent of $\theta$ conditional on $X$ (see Remark \ref{remark:Garbling}).  On this probability space,  define $Z$ to be the random posterior belief of $\theta$ given $X$, i.e., the distribution of $\theta \mid X$, and define $Z'$ to be the random posterior belief of $\theta$ given $X'$, i.e., the distribution of $\theta \mid X'$. We need to show that $\mathbb{E}[Z \mid Z'] = Z'$. 

For any realization $\theta_i$ of $\theta$, define $Z_i \equiv \mathbb{E}[\mathbbm{1}_{\theta_i} \mid X] =  \mathbb{E}[\mathbbm{1}_{\theta_i} \mid X, X']$ (where the second equality is due to independence of $\theta$ and $X'$ conditional on $X$) and  define $Z'_i \equiv \mathbb{E}[\mathbbm{1}_{\theta_i} \mid X']$.
Then
\begin{align}
\mathbb{E}[Z_i \mid X'] & = \mathbb{E}[[\mathbbm{1}_{\theta_i} \mid X,X'] \mid X' ] \nonumber \\
& = \mathbb{E}[\mathbbm{1}_{\theta_i} \mid X' ] \nonumber \\
& = Z'_i \label{eq:reduceZ}
\end{align}
where the second equality follows from the law of iterated expectations (henceforth abbreviated to L.I.E.). Moreover,
\begin{align*}
\mathbb{E}[Z_i \mid Z'] & = \mathbb{E}[ \mathbb{E}[Z_i \mid X',Z'] \mid Z']  && \mbox{by L.I.E.}\\
& = \mathbb{E}[\mathbb{E}[Z_i \mid X'] \mid Z'] &&  \mbox{since $Z'$ is a function of $X'$} \\
& = \mathbb{E}[Z'_i \mid Z']  && \mbox{using (\ref{eq:reduceZ})}\\
& = Z'_i
\end{align*}
Repeating this argument for every $\theta_i$, we have the desired result. \\[2mm]
\medskip
($4 \Rightarrow 5$) Suppose $F$ is a MPS of $F'$ with associated random variables $Z$ and $Z'$ satisfying $\mathbb{E}(Z \mid Z') = Z'$. Then for any continuous and convex function $h: \Delta(\Theta) \rightarrow \mathbb{R}$, 
\begin{align*}
\int_{\Delta(\Theta)} h(p) dF(p) & = \mathbb{E}[h(Z)] \\
& = \mathbb{E}[\mathbb{E}[h(Z) \mid Z']] && \mbox{by L.I.E.}\\
& \geq \mathbb{E}[h(\mathbb{E}[Z \mid Z'])] && \mbox{by Jensen's inequality}\\
& = \mathbb{E}[h(Z')] && \mbox{by assumption of MPS}\\
& = \int_{\Delta(\Theta)} h(p) dF'(p)
\end{align*}
So $F$ dominates $F'$ in the convex order.

\medskip

($5 \Rightarrow 2$) Fix any action set $A$ and utility function $u$, and define $h: \Delta(\Theta) \rightarrow \mathbb{R}$ by
\[h(p) = \max_{a \in A} \sum_{\theta \in \Theta} p(\theta) u(a,\theta)\]
to be the maximum achievable payoff under belief $p$. The function $h$ is the pointwise maximum of linear functions, and hence it is continuous and convex. Letting $p\sim F$ denote the agent's posterior belief, the maximum \emph{ex-ante} payoff is
\[\int_{\Delta(\Theta)} h(p)dF(p)\]
So dominance in the convex order implies  ``more valuable." \end{proof}

\section{Additional Exercises}

\begin{exercise}[U] The state space is $\Theta = \{\theta_1,\theta_2,\theta_3\}$ and an agent's prior belief is $p=(1/2,1/4,1/4)$. The agent chooses from actions in the set $A = \{a_1,a_2,a_3\}$ and has the payoff function
\[u(a,\theta) = \left\{\begin{array}{cl}
1 & \mbox{ if } a=\theta \\
0 & \mbox{ otherwise }
\end{array} \right.\]

\begin{enumerate}
\item[(a)] What is the highest expected payoff that the agent can achieve?
\item[(b)] Now suppose the agent gets to see the outcome of the following signal structure before choosing an action:
\[\begin{array}{ccc}
& s_1 & s_2 \\
\theta_1 & 0 & 1\\
\theta_2 & 1/2 & 1/2 \\
\theta_3 & 1 & 0 
\end{array}\]
\begin{itemize}
\item[(i)] Suppose the realized signal outcome is $s_1$. Solve for the agent's posterior belief and optimal action.
\item[(ii)] Suppose the realized signal outcome is $s_2$. Solve for the agent's posterior belief and optimal action.
\item[(iii)] What is the agent's best expected payoff (where the expectation is taken prior to the realization of the signal outcome)?
\end{itemize}
\end{enumerate}
\end{exercise} 

\begin{exercise}[U] (This problem is based on \citet{Meyer}.) Consider the setting of Example \ref{ex:Meyer}. It turns out that we can make the second realization of $P$ strictly valuable again by biasing it in favor of the more likely signal realization. That is, let the realizations of $P$ be denoted $s_1$ and $s_2$, where
	\[P = \left(\begin{array}{cc}
		 3/4 & 1/4 \\
		 1/4 & 3/4\end{array}\right)\]
		 and modify the second signal in the following way: If the first realization is $s_1$, then the second signal realization is determined by
		 	\[Q_1 = \left(\begin{array}{cc}
		 3/4 + c& 1/4-c \\
		 1/4 + c & 3/4-c \end{array}\right) \]
		 and if the first realization is $s_2$, the second signal realization is determined by
		 
\[		 Q_2 = \left(\begin{array}{cc}
		 3/4 - c& 1/4 + c \\
		 1/4 - c & 3/4 + c \end{array}\right)\]
where in both cases the realization of the second signal is independent of the first conditional on the state. 
		 
\begin{itemize}
\item[(a)] Show that for any $c \in (0,1/4]$, the value of observing this second (biased) signal is strictly positive.
\item[(b)] Solve for the size of the bias $c \in (0,1/4]$ that leads to the highest expected payoff for the agent.
\end{itemize}
\end{exercise}

\begin{exercise}[G] \label{ex:LsA} Let the sets $A$, $\Theta$, and $S$ be finite, and prove that the set $\Lambda_\sigma(A)$ (from Definition \ref{def:Feasible}) is compact and convex for every $\sigma: \Theta \rightarrow \Delta(S)$.
\end{exercise}

\begin{exercise}[G] Consider two random variables
$
X=\theta+\varepsilon$ and
$Y=\theta +\varepsilon'$, where $\theta$, $\eps$, and $\eps'$ are mutually independent. 
\begin{itemize}
\item[(a)] Suppose that $\theta \sim \mathcal{N}(0,1)$ and $\varepsilon,\varepsilon' \in \mathbb{R}$ are distributed
$(\varepsilon, \varepsilon') \sim \mathcal{N}\left(\mu, \Sigma \right).$
Prove that $X$ and $Y$ are Blackwell comparable for all mean vectors $\mu$ and covariance matrices $\Sigma$. 
\item[(b)] Suppose that 
\[\theta \sim \mathcal{N}\left(\left(\begin{array}{c} 0 \\ 0 \end{array} \right), \left(\begin{array}{cc} 1 & 0 \\ 0 & 1 \end{array}\right) \right)\]
and $\varepsilon,\varepsilon' \in \mathbb{R}^2$ are distributed
$(\varepsilon, \varepsilon') \sim \mathcal{N}\left(\mu, \Sigma \right).$
Prove that $X$ and $Y$ are not always Blackwell ranked by demonstrating a pair $(\mu,\Sigma)$ such that $X$ allows for a strictly higher expected payoff for one decision problem, and $Y$ allows for a strictly higher expected payoff given another. 
\end{itemize}
\end{exercise}

\begin{exercise}[G]  In each of the following parts, determine whether the statement is true or false and prove your claim in either case. 
\begin{itemize}
\item[(a)] The state $\theta$ belongs to $\{\theta_1,\theta_2\}$ and the two signals are defined as 
\begin{align*}
X = \theta + \eps_1, \quad \eps_1 \sim U([-1/2,1/2]) \\
\widetilde{X} =  \theta + \eps_2, \quad \eps_2 \sim U([-1/3,1/3]) 
\end{align*}
where $U$ denotes the uniform distribution. The signals $X$ and $\widetilde{X}$ can be Blackwell ranked. 

\item[(b)] The state $\theta$ belongs to $\{0,1/3,2/3,1\}$ and the two signals are defined as 
\begin{align*}
X = \theta + \eps_1, \quad \eps_1 \sim U([-1/2,1/2]) \\
\widetilde{X} =  \theta + \eps_2, \quad \eps_2 \sim U([-1/3,1/3]) 
\end{align*}
The signals $X$ and $\widetilde{X}$ can be Blackwell ranked. 
\end{itemize}

\end{exercise}

\begin{exercise}[G] (This problem is based on \citet{BrooksFrankelKamenica2}.)  Consider the following strengthening of the Blackwell order. Let $\theta$, $X$, and $X'$ be random variables defined on the same probability space $(\Omega, \Sigma, P)$. 

\begin{definition} Say that $X$ \emph{strongly Blackwell dominates} $X'$ if $(X,\widetilde{X})$ Blackwell dominates $(X',\widetilde{X})$ for every random variable $\widetilde{X}$ also defined on $(\Omega, \Sigma, P)$. 
\end{definition}

\noindent Clearly a necessary condition is for $X$ to Blackwell dominate $X'$ (choose $\widetilde{X}$ to be null information). A sufficient condition is for the realization of $X'$ to be known from the realization of $X'$, i.e., for the distribution of $X'\mid X$ to be degenerate for every realization of $X$ (what \citet{BrooksFrankelKamenica} call the \emph{refinement order}). Provide an example in between, namely a signal $X$ that strongly Blackwell dominates $X'$, where the realization of $X'$ is not known from $X$.
\end{exercise}

\chapter{Comparing Information II: Cost of Information} \label{sec:CostofInformation}

So far we have considered decision problems in which the signal informing the agent's decision is given exogenously. In many economic applications, agents can acquire information at a cost and thereby control the signal that they observe. The full problem the agent faces is often specified as
\[\max_{\sigma: \Theta \rightarrow \Delta(S)} \int_{\Delta(\Theta)} \max_{a \in A} \mathbb{E}_q[u(a,\theta)] d\tau_\sigma(q) - \mbox{cost of acquiring $\sigma$}\]
where $\tau_\sigma$ denotes the distribution over posterior beliefs induced by signal $\sigma$. 

This chapter discusses how to model the cost of information, and is divided into two sections. Section \ref{sec:PriorDependent} considers \emph{prior-dependent} cost functions that are a function both of the agent's prior $p \in \Delta(\Theta)$ and of the signal $\sigma: \Theta \rightarrow \Delta(S)$. Section \ref{sec:PriorIndependent} considers \emph{prior-independent} cost functions that depend only the signal $\sigma$. The former are often interpreted as costs of information processing while the latter are often associated with a physical or exogenous cost of producing information. Both approaches draw from information theory, and we review relevant concepts in Section \ref{sec:InformationTheory}.

Two useful benchmarks to keep in mind are the following. 

\begin{example}[Binary] \label{ex:BinaryCost} The unknown state $\theta$ is equally likely to take the value 0 or 1, and the agent chooses an action $a \in \{0,1\}$ with payoff $u(a,\theta) = \mathbbm{1}(a=\theta)$. This action is based on the signal
\[\begin{array}{ccc}
& s=0 & s=1 \\
\theta=0 & \varphi & 1-\varphi  \\
\theta=1 & 1-\varphi  & \varphi 
\end{array}\]
where the agent chooses $\varphi \in [0,1]$.
\end{example}

\begin{example}[Gaussian] \label{ex:GaussianCost}
An agent chooses an action $a \in \mathbb{R}$  and receives the payoff $-(a-\theta)^2$, where $\theta \sim \mathcal{N}(\mu, \sigma_\theta^2)$ is an unknown state. This action is based on a signal $X= \theta + \eps$ where $\eps \sim \mathcal{N}(0, \sigma_\eps^2)$, and the signal noise $\sigma_\eps^2$ is chosen by the agent. 
\end{example}

\section{Information Theoretic Preliminaries} \label{sec:InformationTheory}

This section reviews the definitions of entropy and KL divergence. 

\subsection{Entropy} \label{sec:Entropy}

First assume a finite set of states $\Theta$ with $n \equiv \vert \Theta \vert$, and consider beliefs $p=(p_1, \dots, p_n)$ defined over this set.

\begin{definition}[\citet{shannon}] \label{def:Entropy}
Let $\Theta = \{\theta_1, \dots, \theta_n\}$ for any $n<\infty$. The \emph{entropy} of belief $p \in \Delta(\Theta)$ is
\[H(p) = - \sum_{\theta \in \Theta} p(\theta) \ln(p(\theta)) = \mathbb{E}_{\theta \sim p}[-\ln(p(\theta))]\]
where $0\ln0 = 0$.
\end{definition}

\begin{remark} Entropy is also sometimes defined as a function of the random variable rather than its distribution, i.e., $H(\theta) =  \mathbb{E}[-\ln(p(\theta))].$
\end{remark}

\medskip

Entropy is a quantification of uncertainty in a distribution. The higher the entropy of the distribution, the more information is contained in the realization of a random variable it governs. (Entropy is also often interpreted as the ``surprise factor" of the outcome.)

\begin{example} \label{ex:BinaryEntropy} Suppose $\Theta = \{\theta_1, \theta_2\}$. The entropy of any belief $(q,1-q)$ is 
\begin{equation} \label{eq:BinaryEntropy}
H(q) = -q\ln(q) - (1-q) \ln(1-q).
\end{equation}
This curve is depicted in Figure \ref{fig:Entropy} below. It is concave, minimized at the two degenerate distributions $(0,1)$ and $(1,0)$, and maximized at the uniform distribution $(1/2,1/2)$.
\begin{figure}[H] 
\begin{center}
\includegraphics[scale=0.45]{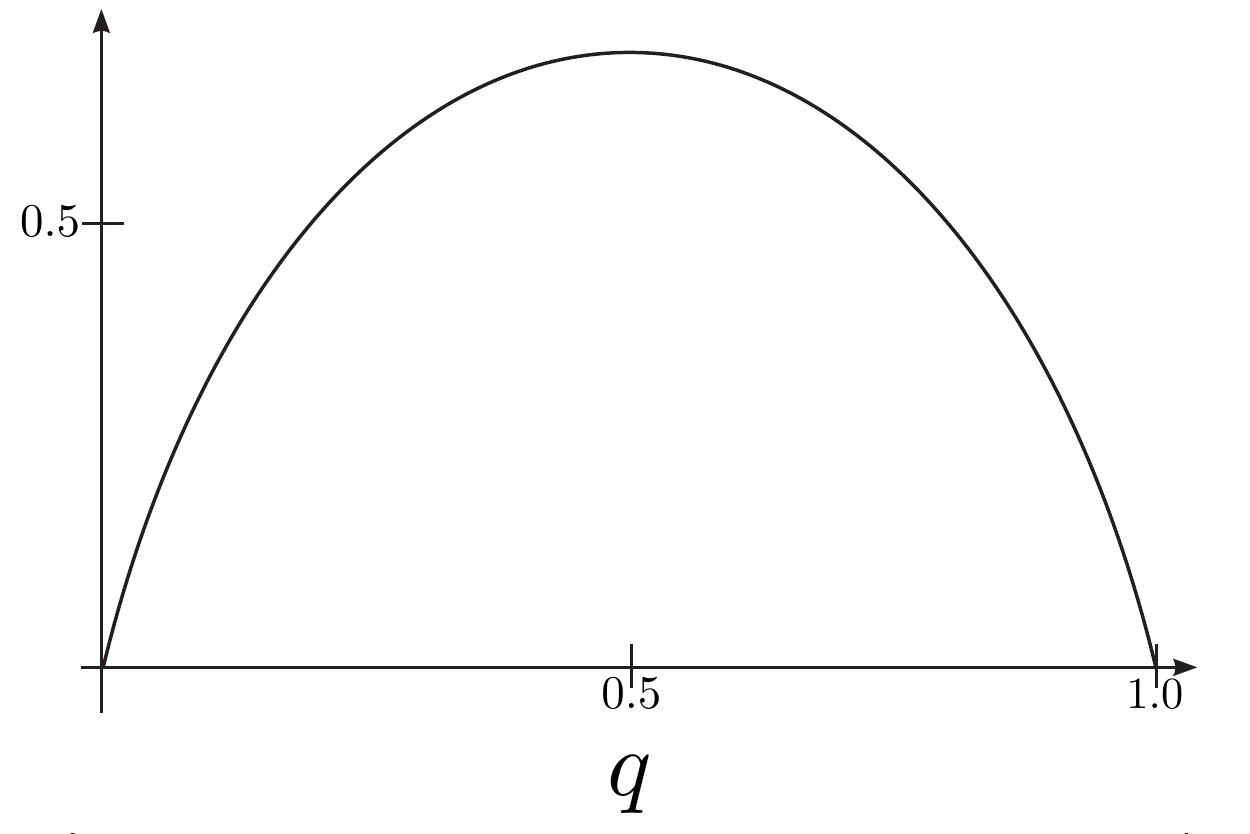}
\caption{Plot of the entropy of the distribution $(q,1-q)$ as $q$ varies in $[0,1]$.} \label{fig:Entropy}
\end{center}
\end{figure}
\end{example}

Several key properties of entropy are collected below. 

\medskip
\begin{property}[Maximal Value] \label{propEntr:Max} $H(p) \leq H\left(\frac{1}{n}, \dots, \frac{1}{n}\right)$ for every $n<\infty$ and $p \in \Delta(\{\theta_1, \dots, \theta_n\})$; that is, entropy is maximized at the uniform distribution.
\end{property}

\medskip

\begin{property}[Probability Zero States] $H(p) = H(p_1, \dots, p_n,0)$ for every $n<\infty$ and $p \in \Delta(\{\theta_1, \dots, \theta_n\})$; that is, entropy is unchanged by an expansion of the state space to include probability-zero outcomes.
\end{property}

\medskip

\begin{property}[Continuity] $H$ is continuous with respect to all of its arguments.
\end{property} 

\medskip

\begin{property}[Chain Rule] \label{propEntr:Chain} Suppose $(X,Y) \in \mathcal{X} \times \mathcal{Y}$ with $\mathcal{X} = \{x_1, \dots, x_n\}$ and $\mathcal{Y} = \{y_1, \dots, y_m\}$, where the joint distribution of $(X,Y)$ is denoted $p$, the marginal distribution of $X$ is  $p_X$, and the conditional distribution of $Y$ given $X$ is $p_{Y \mid X}$. Then
\[H(p) = H(p_X) + \sum_{i=1}^n p_X(x_i) H(p_{Y \mid X=x_i})\]
or more simply
\[H(X,Y) = H(X) + H(Y \mid X)\]
where $H(X,Y) \equiv H(p)$ is the entropy of the joint distribution, $H(X) \equiv H(p_X)$ is the entropy of of the marginal distribution of $X$, and $$H(Y \mid X) \equiv \sum_{i=1}^n p_X(x_i) H(p_{Y \mid X=x_i})$$ is the expected entropy of the conditional distribution of $Y$ given $X$, also known as the \emph{conditional entropy} of $Y$ given $X$.
\end{property}

\begin{remark} In the special case where $X$ and $Y$ are independent, Property \ref{propEntr:Chain} implies $H(X,Y) = H(X) + H(Y)$.
\end{remark}

\medskip

\begin{property}[Nonnegativity] \label{propEntr:Positive} $H(p)\geq 0$ for all distributions $p$.
\end{property}

\medskip

\begin{property}[Degenerate Distributions] \label{propEntr:Degenerate} $H(p)=0$ for all degenerate distributions $p$.
\end{property}

\medskip

\begin{property}[Concavity] \label{propEntr:Concave} $H$ is concave.
\end{property}

\medskip

\begin{property}[Relabelling of States] \label{propEntr:Relabel} $H(p_1, \dots, p_n) = H(p_{\pi(1)}, \dots, p_{\pi(n)})$ for any bijection $\pi$ from $\{1,\dots, n\}$ to itself; that is, entropy is invariant to a relabelling of states. 
\end{property}

\medskip

\begin{property}[Information Reduces Uncertainty] \label{propEntr:Information} $H(Y\mid X) \leq H(Y)$ with equality if and only if $X$ and $Y$ are independent; that is, conditioning on information reduces expected entropy.
\end{property}

Properties \ref{propEntr:Max}-\ref{propEntr:Chain} constitute a set of necessary and sufficient conditions for the form of $H$ given in (\ref{def:Entropy}), up to rescaling.

\begin{proposition}[\citet{Khinchin}] Let $H(p_1, \dots, p_n)$ be a function defined for any $n \in \mathbb{Z}_{+}$ and for all values $p_1, \dots, p_n$ satisfying $p_i \geq 0$ for each $i=1, \dots, n$ and $\sum_{i=1}^n p_i = 1$. Then $H$ satisfies Properties \ref{propEntr:Max}-\ref{propEntr:Chain} if and only if 
\[H(p_1, \dots, p_n) = -\lambda \sum_{i=1}^n p_i \ln(p_i)\]
for some constant $\lambda >0$.\footnote{Recalling that $\log_b(x) = \frac{\log_a(x)}{\log_a(b)}$ for any two bases $a, b >0$, changing the logarithm to a different basis simply rescales the measure. Choice of base $2$ and of base $e$ are both common.}
\end{proposition}

Properties \ref{propEntr:Positive}, \ref{propEntr:Degenerate}, and \ref{propEntr:Relabel} are immediate from the functional form of entropy. Property \ref{propEntr:Concave} (concavity) follows because $-x\log(x)$ is concave, and the sum of concave functions is concave. (In fact, the same argument shows that entropy is \emph{strictly} concave, so Property \ref{propEntr:Max} can be strengthened to the statement that the uniform distribution is the unique maximum.) The following exercise asks you to prove that entropy satisfies Property \ref{propEntr:Information}.

\begin{exercise}[G] Suppose $(X,Y) \in \mathcal{X} \times \mathcal{Y}$ with $\vert \mathcal{X} \vert = n$ and $\vert \mathcal{Y} \vert = m$, where $p_X$ and $p_Y$ denote the marginal distributions of $X$ and $Y$, and $p_{Y \mid X}$ denotes the conditional distribution of $Y$ given $X$. Let  $H(Y) \equiv H(p_Y)$ be the entropy of of the marginal distribution of $Y$, and $H(Y \mid X) \equiv \sum_{i=1}^n p_X(x_i) H(p_{Y \mid X=x_i})$ be the conditional entropy of $Y$ given $X$. Prove that 
$H(Y \mid X) \leq H(Y)$.
\end{exercise}

\citet{shannon} defines a continuous version of entropy.

\begin{definition}
The entropy of probability density $p$ on $\Theta \subseteq \mathbb{R}$ is
\[H(p) = - \int_{\theta \in \Theta} p(\theta) \ln(p(\theta)) d\theta \]
\end{definition}

\begin{example}
Recall that the normally distributed variable $\theta \sim \mathcal{N}(\mu,\sigma^2)$ has density $p(\theta) = \frac{1}{\sigma \sqrt{2\pi}}e^{-\frac12 \left(\frac{\theta-\mu}{\sigma}\right)^2}$. The entropy of this distribution is
\begin{align}
\mathbb{E}\left[ -\ln\left(\frac{1}{\sigma \sqrt{2\pi}}e^{-\frac12 \left(\frac{\theta-\mu}{\sigma}\right)^2}\right)\right] & = -\ln\left(\frac{1}{\sigma \sqrt{2\pi}}\right) + \frac{1}{2\sigma^2}\mathbb{E}\left[(\theta - \mu)^2\right] \nonumber \\
& = \frac12 \ln\left(2\pi\sigma^2\right) + \frac12 \label{eq:GaussianEntropy}
\end{align}
using in the second equality that $\mathbb{E}[(\theta-\mu)^2] = \sigma^2$. So entropy and variance order normal distributions in the same way.

\end{example}

\subsection{Kullback-Liebler Divergence} \label{sec:KL}

The \emph{Kullback-Liebler Divergence (KL divergence)}, also known as \emph{relative entropy}, quantifies how different two distributions are.

\begin{definition}[KL-Divergence] \label{def:KL}
Let $\Theta = \{\theta_1, \dots, \theta_n\}$ for any $n<\infty$, and let $p,q \in \Delta(\Theta)$. Then the KL divergence from $q$ to $p$ is 
\[D(p \| q) = \sum_{\theta \in \Theta} p(\theta) \ln\left(\frac{p(\theta)}{q(\theta)}\right) = \mathbb{E}_{\theta \sim p}\left[\ln\left(\frac{p(\theta)}{q(\theta)}\right)\right]\]
where $0\ln 0 =0$.
\end{definition}

\begin{example}[Binary] Let $\Theta = \{\theta_1,\theta_2\}$ with $(p,1-p)$ and $(q,1-q)$ be two distributions on this set. Then
\[D(p \| q) = p \ln\left(\frac{p}{q}\right) + (1-p)\ln\left(\frac{1-p}{1-q}\right).\]
Intuitively, larger log likelihood ratios $\ln\left(\frac{p}{q}\right)$ and $\ln\left(\frac{1-p}{1-q}\right)$ reflect distributions that are more different. KL divergence aggregates these log likelihood ratios by weighting them with respect to their probabilities under a reference distribution, which is chosen to be either of $p$ or $q$. 
\end{example}

\begin{example}[Gaussian] \label{ex:GaussianKL} Let $p$ and $q$ denote two Gaussian densities with common variance $\sigma$ and different means $\mu_p$ and $\mu_q$. Then
\begin{align*}
D(p \| q) &=  \mathbb{E}_{\theta \sim p}\left[\ln\left(\frac{e^{-\frac12\left(\frac{\theta-\mu_p}{\sigma}\right)^2}}{e^{-\frac12\left(\frac{\theta-\mu_q}{\sigma}\right)^2}}\right)\right] \\
& = \frac{\mu_q^2 - \mu_p^2}{2\sigma^2} - \frac{\mu_q - \mu_p}{\sigma^2}\cdot \mathbb{E}_{\theta \sim p}(\theta) = \frac{(\mu_q - \mu_p)^2}{2\sigma^2}
\end{align*}
So as we might expect, the further the two means, the larger the KL divergence between the two distributions.
\end{example}

\noindent KL divergence is not in general symmetric (with Example \ref{ex:GaussianKL} being a notable exception) and hence it is not a metric. Other key properties of the KL divergence include:

\begin{property}[Nonnegativity] \label{propKL:Nonnegative} $D(p \| q) \geq 0$ for all $p,q\in \Delta(\Theta)$, with equality if and only if $p=q$.
\end{property}

\noindent To prove this, observe that 
\begin{align*}
-D(p \| q) & = \mathbb{E}_{\theta \sim p}\left[ \ln\left(\frac{q(\theta)}{p(\theta)}\right) \right] \\
& \leq \ln \left( \mathbb{E}_{\theta \sim p}\left[ \frac{q(\theta)}{p(\theta)} \right]\right) && \mbox{by Jensen's inequality} \\
& = \ln(1) = 0 && \mbox{since $\sum_{\theta \in \Theta} p(\theta) \left(\frac{q(\theta)}{p(\theta)}\right) =1$}
\end{align*}

\begin{property}[Additivity for Independent Distributions] Suppose $p_1\in \Delta(\mathcal{X}_1)$ and $p_2 \in \Delta(\mathcal{X}_2)$ are independent distributions, with $p(x_1,x_2)=p_1(x_1)p_2(x_2)$. Likewise suppose $q_1 \in \Delta(\mathcal{X}_1)$ and $q_2 \in \Delta(\mathcal{X}_2)$ are independent distributions with $q(x_1,x_2)=q(x_1)q(x_2)$. Then
\[D(p \| q) = D(p_1 \| q_1) + D(p_2 \| q_2),\]
i.e., KL divergence is additive for independent distributions.
\end{property}

This property follows from straightforward algebra:
\begin{align*}
D(p \| q) &  = \sum_{x_1 \in \mathcal{X}_1} \sum_{x_2 \in \mathcal{X}_2} p(x_1,x_2) \ln \left(\frac{p(x_1,x_2)}{q(x_1,x_2)}\right) \\
& = \sum_{x_1 \in \mathcal{X}_1} \sum_{x_2 \in \mathcal{X}_2} p_1(x_1)p_2(x_2)\ln \left(\frac{p_1(x_1)p_2(x_2)}{q_1(x_1)q_2(x_2)}\right) \\
& = \sum_{x_2 \in \mathcal{X}_2} p_2(x_2) \left(\sum_{x_1 \in \mathcal{X}_1} p_1(x_1) \ln \left(\frac{p_1(x_1)}{q_1(x_1)}\right)\right) \\
& \quad  + \sum_{x_1 \in \mathcal{X}_1} p_1(x_1) \left(\sum_{x_2 \in \mathcal{X}_2} p_2(x_2) \ln \left(\frac{p_2(x_2)}{q_2(x_2)}\right)\right) = D(p_1 \| q_1) + D(p_2 \| q_2)
\end{align*}
where independence is invoked in the second equality. 

\begin{property}[Convexity] $D$ is convex: For any two pairs $(p,q)$ and $(p',q')$, and any $\alpha \in [0,1]$, we have
\[D\left(\alpha p + (1-\alpha) p' \| \alpha q + (1-\alpha) q' \right) \leq \alpha D(p \| q) + (1-\alpha) D(p' \| q')\]
\end{property}

\begin{exercise}[G] Prove the above property using the following fact:
\begin{fact}[Log-Sum Inequality] Let $a_1, \dots a_n$ and $b_1, \dots b_n$ be nonnegative real numbers. Then
\[\sum_{i=1}^n a_i \ln\left(\frac{a_i}{b_i}\right) \geq \left(\sum_{i=1}^n a_i\right) \ln\left(\frac{\sum_{i=1}^n a_i}{\sum_{i=1}^n b_i}\right).\]
\end{fact}
\end{exercise}

There is a close relationship between KL divergence and entropy. First, the entropy of a distribution $p \in \Delta(\Theta)$ with $n \equiv \vert \Theta \vert < \infty$ can be rewritten directly in terms of KL divergence:
\[H(p) = \ln n - D(p \| U)\]
where $U$ denotes the uniform distribution on $\Theta$. Thus, the larger the KL divergence from the uniform distribution to $p$, the lower the entropy of $p$. This is proved by observing that
\begin{align*} 
\ln n - D(p \| U) &= \ln n - \sum_{\theta \in \Theta} p(\theta) \ln\left(\frac{p(\theta)}{1/n}\right) \\
& = \sum_{\theta \in \Theta} p(\theta) (\ln n - \ln\left(n p(\theta) \right)) && \mbox{since $\sum_{\theta \in \Theta} p(\theta) =1$} \\
& = -\sum_{\theta \in \Theta} p(\theta) \ln(p(\theta)) = H(p) 
\end{align*}

\begin{remark} Together with Property \ref{propKL:Nonnegative}, the above relationship implies that entropy is maximized at the uniform distribution (Property \ref{propEntr:Max}).
\end{remark}

KL divergence cannot be rewritten directly in terms of entropy, although
\[D(p \| q) = - \sum_{\theta \in \Theta} p(\theta) \ln\left(q(\theta)\right) - H(p)\]
where $-\sum_{\theta \in \Theta} p(\theta) \ln\left(q(\theta)\right)$ is the \emph{cross-entropy} of distribution $q$ relative to $p$. 

\section{Prior-Dependent Costs} \label{sec:PriorDependent}

Returning to the question of how to model the cost function, we begin with \emph{prior-dependent} cost functions. Dependence on the prior belief means that the cost of absorbing the information content of a signal varies with what the agent already knows. This feature may be justified if we view the cost of information as an information processing or cognitive cost: For example, processing a news article about a proposed tax change may be relatively easy for someone who already understands this tax change well, but cognitively taxing for someone who does not. 
 
It will be convenient to represent signals as distributions over posterior beliefs, as in  Section \ref{sec:BayesPlausibility}. Following Definition \ref{def:BayesPlausible}, we use $\mathcal{T}(p)$ to denote the set of Bayes plausible distributions given prior $p$, and we further define
\[\mathcal{S} = \{(p, \tau) : p \in \Delta(\Theta), \tau \in \mathcal{T}(p)\}\]
to be the domain of prior beliefs and Bayes plausible distributions. The cost functions in this section will take the form $C: \mathcal{S} \rightarrow \mathbb{R}$.

\subsection{Uniform Posterior Separability} \label{sec:ReductionUncertainty}

One popular class of cost functions are those that are \emph{uniformly posterior separable}.

\begin{definition}[\citet{CaplinDean2013,CaplinDeanLeahy2022}] \label{def:UPS} The cost function $C: \mathcal{S} \rightarrow \mathbb{R}$ is \emph{uniformly posterior separable} (henceforth UPS) if there is a strictly concave function $\Phi$ such that
\[C(p,\tau) = \Phi(p) - \mathbb{E}_{q \sim \tau}[\Phi(q)] \quad \forall (p,\tau) \in \mathcal{S}.\]
\end{definition}
We can interpret this cost of information as the expected reduction of uncertainty, where $\Phi: \Delta(\Theta) \rightarrow \mathbb{R}$ measures how uncertain the belief is.

\begin{remark} The cost of ``no information" is zero, since $\Phi(p) - \mathbb{E}_{q \sim \delta_p}[\Phi(q)]=\Phi(p)-\Phi(p)=0$ (with $\delta_p$ denoting the degenerate distribution at the prior $p$).
\end{remark}

\begin{remark} Concavity of $\Phi$ guarantees that uncertainty decreases in expectation when more information is received. Together with Bayes plausibility of $\tau$, this further implies that UPS cost functions are everywhere positive:
\begin{align*}
\Phi(p) - \mathbb{E}_{q \sim \tau}[\Phi(q)] & \geq  \Phi(p) - \Phi(\mathbb{E}_{q \sim \tau}[q]) && \mbox{by Jensen's inequality}\\
& = \Phi(p) - \Phi(p) && \mbox{by Bayes plausibility of $\tau$} \\
& = 0 
\end{align*} 
\end{remark}

\begin{remark} UPS cost functions are consistent with the Blackwell order. That is, let $\sigma$ and $\sigma'$ be arbitrary signals where $\sigma$ Blackwell dominates $\sigma'$. Fix any prior $p$, and let $\tau_\sigma$ and $\tau_{\sigma'}$ denote the distributions over posteriors that are induced by  $\sigma$ and $\sigma'$. Then for any UPS cost function $C$, we have $C(p,\tau_\sigma) \geq C(p,\tau_{\sigma'})$ since 
\[C(p,\tau) = \int (\Phi(p) - \Phi(q))d\tau(q)\]
where $\Phi(p) - \Phi(q)$ is convex in $q$, and $\tau_\sigma$ dominates $\tau_{\sigma'}$ in the convex order (see the characterization of the Blackwell order in Section \ref{sec:ConvexOrder}). 
\end{remark}

The leading specification of $C$ is the expected reduction of the entropy of the agent's belief.

\begin{example}[Entropy Reduction] Let $H$ be the entropy function given in Definition \ref{def:Entropy}. Then define
\begin{equation} \label{def:EntropyCost}
C_{\text{Ent}}(p, \tau) = H(p) - \mathbb{E}_{q \sim \tau}[H(q)] \quad \forall (p,\tau) \in \mathcal{S}
\end{equation}
to be the expected reduction in the entropy of the agent's belief. 
\end{example}

Initially proposed as an information cost in \citet{Sims2003}, this cost function is a cornerstone of the rational inattention literature \citep{CaplinDean2013,CaplinDean2015,HebertWoodfordAER,HebertLaO}.  Various conceptual foundations for entropic costs and uniformly posterior separable cost functions (as well as the broader class of posterior separable cost functions discussed in Section \ref{sec:PosteriorSeparable}) can be found in \cite{CaplinDean2013}, \citet{MatejkaMcKay2015}, \citet{MorrisStrack}, \citet{HebertWoodford}, \citet{BloedelZhong}, and \citet{Denti2022} among others.

\begin{example} In the setting of Example \ref{ex:BinaryCost}, we have
\[C_{\text{Ent}}(p,\tau_\varphi) = -\ln\left(\frac{1}{2}\right)  + \left(\varphi\ln(\varphi) + (1-\varphi) \ln(1-\varphi)\right)\]
where $\tau_\varphi$ denotes the distribution over posterior beliefs induced by the signal indexed to $\varphi$.
The cost of the signal is largest when $\varphi\in \{0,1\}$ (corresponding to a fully revealing signal) and smallest when $\varphi=1/2$ (corresponding to an uninformative signal).
\end{example}

\bigskip

Besides entropy, another natural choice of $\Phi$ is variance.

\begin{example}[Variance Reduction] Let 
\begin{equation} \label{def:VarCost}
C_{\text{Var}}(p, \tau) = \Var(p) - \mathbb{E}_{q \sim \tau}[\Var(q)]
\end{equation}
be the expected reduction in the variance of the agent's belief.
\end{example}

\begin{exercise}[G] Prove that variance is strictly concave, so $C_{\text{Var}}$ is a UPS cost function.
\end{exercise}

\begin{example} Consider the setting of Example \ref{ex:GaussianCost} (where we use $\tau_{\sigma_\eps^2}$ to denote the distribution over posterior beliefs induced by observing the signal $X=\theta +\eps$, $\eps \sim \mathcal{N}(0,\sigma_\eps^2)$). Applying (\ref{eq:GaussianEntropy}),
\begin{align*}
C_{Ent}(p,\tau_{\sigma_\eps^2}) & = \left(\frac12 \ln(2\pi\sigma_\theta^2) + \frac12 \right) -  \left(\frac12 \ln\left(2\pi\left(\frac{\sigma_\theta^2\sigma_\eps^2}{\sigma_\theta^2 + \sigma_\eps^2}\right)\right) + \frac12 \right)\\
& = \frac12 \ln\left(\frac{\sigma_\theta^2 + \sigma_\eps^2}{\sigma_\eps^2}\right)
\end{align*}
while
\begin{align*}
C_{Var}(p,\tau_{\sigma_\eps^2}) & = \sigma_\theta^2 - \frac{\sigma_\theta^2 \sigma_\eps^2}{\sigma_\theta^2 + \sigma_\eps^2} = \frac{\sigma_\theta^4}{\sigma_\theta^2 + \sigma_\eps^2}.
\end{align*}
For every fixed prior variance $\sigma_\theta^2$, both cost functions are strictly decreasing in the noise variance $\sigma_\eps^2$, and thus correspond to different cardinal representations of the same ordering over signals. One interesting contrast is that $C_{Ent}(p,\sigma_\eps^2) \rightarrow \infty$ as $\sigma_\eps^2 \rightarrow 0$, while $C_{Var}(p,\sigma_\eps^2) \rightarrow \sigma^2$. That is, the cost of information using $C_{\Var}$ is bounded above by the agent's prior uncertainty, while entropy cost is unbounded.
\end{example}

\subsection{Decision-Theoretic Foundations}

The function $\Phi$ is interpreted in the previous section as a ``pure" measure of uncertainty,  without reference to why this uncertainty matters. Parallel to Section \ref{sec:DecisionProblem}'s assessment of the value of information using decision problems, \citet{FrankelKamenica} microfound the function $\Phi$ as measuring the instrumental loss of uncertainty for a specific decision problem.

\begin{definition} For any belief $q \in \Delta(\Theta)$ and decision problem $\mathcal{D}=(A,u)$, let
\[\Phi_{\mathcal{D}}(q) = \mathbb{E}_q\left[\max_{a \in A} u(a, \theta)\right] - \max_{a \in A} \mathbb{E}_q\left[u(a,\theta)\right].\]
\end{definition}
The first term of this expression is the agent's best expected payoff when conditioning his action directly on the realized state (which is random and distributed according to the agent's belief $q$). The second term is the best expected payoff that the agent with belief $q$ can achieve given no additional information on which to condition his action. Thus $\Phi_{\mathcal{D}}$ quantifies the agent's payoff loss from not knowing a state which is distributed according to $q$. 

\begin{definition}[\citet{FrankelKamenica}] Say that $\Phi: \Delta(\Theta) \rightarrow \mathbb{R}$ is \emph{valid} if there is a decision problem $\mathcal{D}$ such that $\Phi=\Phi_{\mathcal{D}}$. 
\end{definition}

Any function $\Phi$ that is concave and takes value zero at degenerate distributions (i.e., satisfies Properties \ref{propEntr:Degenerate} and \ref{propEntr:Concave}) can be microfounded using a decision problem in this way. 

\begin{proposition}[\citet{FrankelKamenica}] \label{prop:FK} $\Phi: \Delta(\Theta) \rightarrow \mathbb{R}$ is valid if and only if it satisfies Properties \ref{propEntr:Degenerate} and \ref{propEntr:Concave}.
\end{proposition}

This result follows from the subsequent lemma, which is of independent interest.

\begin{lemma} \label{lemm:Support} Let $\Theta$ be a finite set. Then every convex function $V: \Delta(\Theta) \rightarrow \mathbb{R}$ can be represented as
\begin{equation} \label{eq:VDecision}
V(q) = \sup_{a \in A} \mathbb{E}_{q}[u(a,\theta)] \quad \forall q \in \Delta(\Theta)
\end{equation}
for some decision problem $(A,u)$, where $A$ is a set (not necessarily finite) and $u$ is a map $u: \Theta \times A \rightarrow [-\infty,+\infty]$.
\end{lemma}

\noindent The key points in the proof of this lemma are that $\mathbb{E}_q(u(a,\theta))$ is affine in $q$, and that every convex function is the supremum of affine functions lying below it. We'll prove this lemma assuming that $V$ is continuous and has a nonvertical supporting hyperplane at every point $ q\in \Delta(\Theta)$, leaving the completion of the proof when these assumptions fail as Exercise \ref{ex:Vertical}.\footnote{Under these assumptions,  the supremum in (\ref{eq:VDecision}) can be replaced with maximum, as the following proof demonstrates.}

\medskip

\begin{proof} Our approach is to construct a set of actions indexed to beliefs, $A = \{a_q \, : \, q \in \Delta^{n}\}$, and to construct a utility function such that each action $a_{q}$ is optimal at the belief $q$. To do this, define a family of affine functions $(U_{a_{q}})_{q \in \Delta^{n}}$, where each $U_{a_{q}}: \Delta^{n} \rightarrow \mathbb{R}$ is a supporting hyperplane of the epigraph of $V$ at $q$, as depicted below in Figure \ref{fig:Support}.\footnote{Recall that the epigraph of $V$ is $\{(q,v): v \geq V(q)\}$, the set of points lying on or above $V$.} 

\begin{figure}[h]
\begin{center}
\includegraphics[scale=0.65]{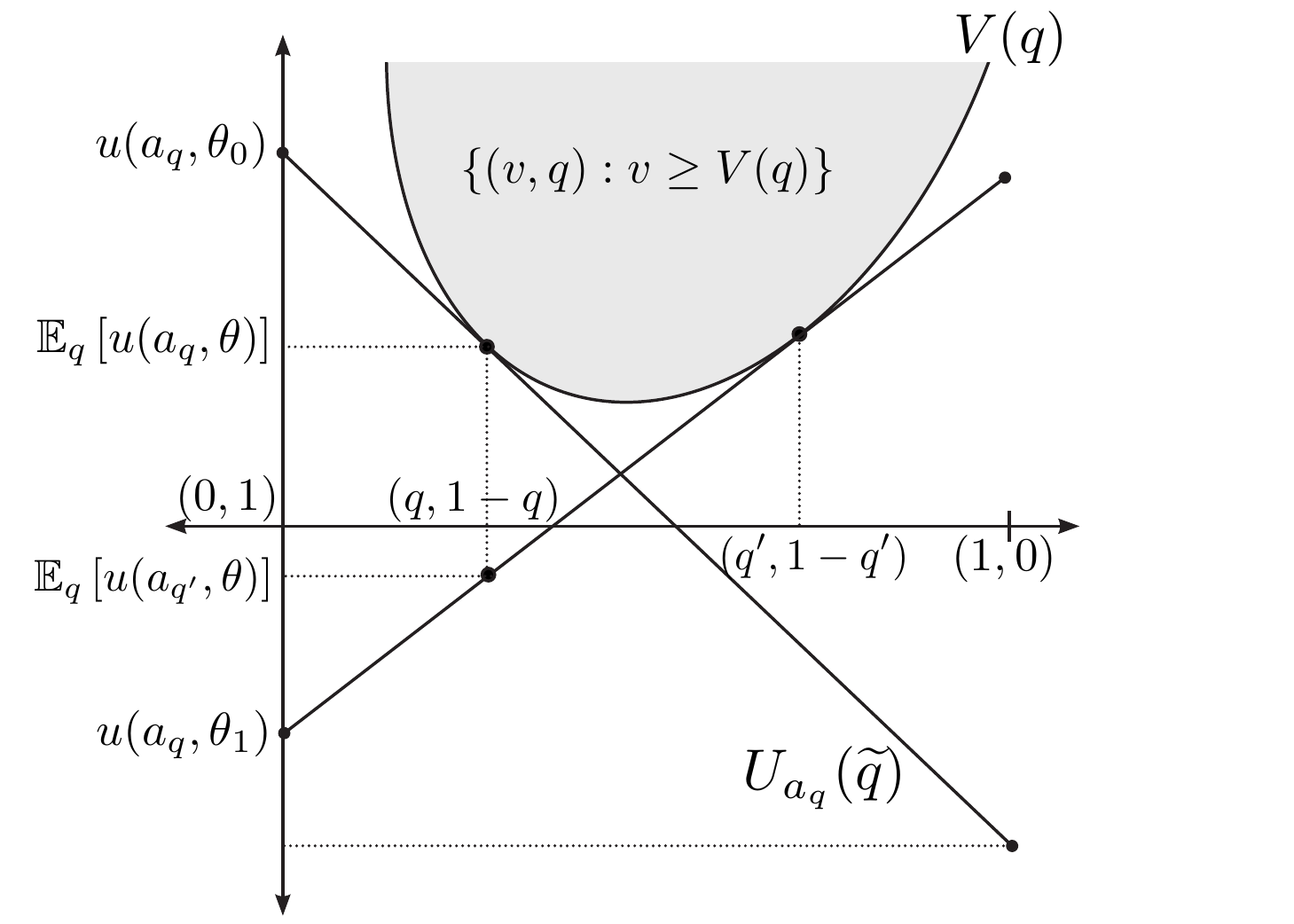}
\end{center}
\caption{Example construction for a binary state space $\Theta = \{\theta_0,\theta_1\}$. The action $a_q$ is optimal at belief $q$; that is, for every other belief $q'$ we have $\mathbb{E}_q[u(a_q,\theta)]\geq \mathbb{E}_q[u(a_{q'},\theta)]$, as depicted here.} \label{fig:Support}
\end{figure}

Since $V$ is continuous and convex, it can be represented on its domain as the supremum of all affine functions lying below it. Since each $U_{a_{q}}$ is affine and lies below $V$, we have that
\[V(q) \geq U_a(q) \quad \forall a \in A,q\in \Delta^{n}.\]
Moreover (by definition) $U_{a_{q}}$ supports $V$ at $q$, so $U_{a_{q}}(q) = V(q)$. This implies that 
\begin{equation} \label{eq:V}
U_{a_{q}}(q) = \max_{a \in A} U_a(q) \quad \forall q \in \Delta(\Theta)
\end{equation}
We now need to express $U_{a_q}$ as an expected utility function. Since each belief $q'$ is a convex combination of the degenerate beliefs $(\delta_{\theta})_{\theta \in \Theta}$ (with weights given by $q'(\theta)$), and $U_{a_q}$ is affine, it follows that
\begin{equation} \label{eq:SetU}
U_{a_{q}}(q') = \sum_{\theta \in \Theta} q'(\theta) U_{a_q}(\delta_\theta) \quad \forall q' \in \Delta(\Theta)
\end{equation}
Now define the utility function $u: \mathbb{R}^n \rightarrow \mathbb{R}$ to satisfy $u(a, \theta) = U_a(\delta_\theta)$ for every $a \in A$ and $\theta \in \Theta$.
Then from (\ref{eq:SetU}),
\[U_{a_q}(q') = \sum_{\theta \in \Theta} q'(\theta) u(a_q,\theta)\]
and so (\ref{eq:V}) implies that
\[\mathbb{E}_q[u(a_q,\theta)] \geq \mathbb{E}_q[u(a,\theta)]\]
 for every $q \in \Delta(\Theta)$ and $a \in A$. Thus each action $a_q$ is optimal at belief $q$, and achieves the expected utility $U_{a_q}(q) = V(q)$ as desired.
\end{proof}

\begin{exercise}[G] \label{ex:Vertical} Complete the proof by showing that the statement of Lemma \ref{lemm:Support} continues to hold when $V$ is discontinuous and/or there exists a belief $q$ at which every supporting hyperplane of $V$ is vertical.

\begin{hint} Observe that vertical supporting hyperplanes can only exist on the boundary of $\Delta(\Theta)$, and that discontinuities can only occur at degenerate beliefs.
\end{hint}
\end{exercise}

We'll now use this lemma to prove Proposition \ref{prop:FK}. 

\medskip

\begin{proof} Suppose $\Phi$ satisfies Assumptions \ref{propEntr:Degenerate} and \ref{propEntr:Concave}. Then $-\Phi$ is convex, so by Lemma \ref{lemm:Support}, there is a set of actions $A$ and a utility function $u:A \times \Theta \rightarrow \mathbb{R}$ such that 
\begin{equation} \label{eq:ApplyLemma}
-\Phi(q) = \max_{a \in A} \mathbb{E}_q[u(a,\theta)] \quad \forall q\in \Delta(\Theta).
\end{equation}

\noindent We need to verify that
\begin{equation} \label{eq:ToShowPhi}
\Phi(q) = \mathbb{E}_q\left[\max_{a \in A} u(a,\theta)\right] - \max_{a\in A} \mathbb{E}_q[u(a,\theta)]
\end{equation}
for every $q \in \Delta(\Theta)$. Again index the states by $\theta_1, \dots, \theta_n$ (where $n \equiv \vert \Theta \vert$), and define $\delta_{\theta_i}$ to be the belief that is degenerate at state $\theta_i$. Then for any $\theta_i \in \Theta$
\begin{align*}
\max_{a \in A} u(a,\theta_i) & = \max_{a \in A}  \mathbb{E}_{\delta_{\theta_i}}[u(a,\theta)] \\
& = - \Phi(\delta_{\theta_i}) && \mbox{by (\ref{eq:ApplyLemma})} \\
& = 0 && \mbox{by Assumption \ref{propEntr:Degenerate}}
\end{align*}
Thus also $\mathbb{E}_q\left[\max_{a \in A} u(a,\theta)\right]=0$ for any belief $q$, which together with (\ref{eq:ApplyLemma}) implies that (\ref{eq:ToShowPhi}) reduces to $\Phi(a) = 0 - (-\Phi(a))$ and is thus true.

In the other direction, 
 \[\Phi(\delta_{\theta}) = \max_{a \in A} u(a,\theta) - \max_{a \in A} u(a,\theta)=0 \quad \forall \theta \in \Theta\]
implying Property \ref{propEntr:Degenerate}. Concavity of $\Phi$ (Property \ref{propEntr:Concave}) follows by construction of $\Phi$ since $\mathbb{E}_q\left[\max_{a \in A} u(a,\theta)\right]$ is affine while $\sup_{a\in A} \mathbb{E}_q[u(a,\theta)]$ is a pointwise supremum of affine functions, and thus convex. \end{proof}

\bigskip
By Proposition \ref{prop:FK}, the two example cost functions from the previous section, $C_{Ent}$ and $C_{Var}$, can be microfounded using decision problems. These decision problems are given below. 

\begin{example}[Microfoundation for Entropy Cost] Set $A = \Delta(\Theta)$ and $u(a,\theta) =  \ln(a(\theta))$, where $\ln0=-\infty$. Then the cost of uncertainty is
\[\Phi_{\mathcal{D}}(q) = \mathbb{E}_q\left[\max_a\left[ \ln(a(\theta))\right]\right] - \max_a \mathbb{E}_q\left[\ln(a(\theta))\right] = H(q).\]
\end{example}

\begin{example}[Microfoundation for Variance Cost] Set $A = \Theta \subseteq \mathbb{R}$ and $u(a, \theta) = -(a-\theta)^2$. Then
\[\Phi_{\mathcal{D}}(q) = \mathbb{E}_q\left[\max_a\left[ - (a- \theta)^2\right]\right] - \max_a \mathbb{E}_q\left[-(a-\theta)^2\right] = Var_q(\theta)\]
\end{example}

\subsection{Posterior Separability} \label{sec:PosteriorSeparable}

A weaker requirement than uniform posterior separability is that the cost of $\tau$ can be written in a way that is separable in the realized posteriors. 

\begin{definition}[\citet{CaplinDean2013,CaplinDeanLeahy2022}] \label{def:PosteriorSeparable} The cost function $C: S \rightarrow \mathbb{R}$ is \emph{posterior separable} if
\[C(p,\tau) = \mathbb{E}[\Phi_p(q)]\]
for some family of convex functions $(\Phi_p)_{p \in \Delta(\Theta)}$ where each $\Phi_p: \Delta(\Theta) \rightarrow \mathbb{R}$ is everywhere weakly positive, and $\Phi_p(p)=0$ for every $p$.
\end{definition}

\begin{remark} When the cost function is posterior separable but not uniformly posterior separable, the cost of acquiring two signals in sequence may depend on the order in which these signals are acquired. This is not true for for UPS cost functions \citep{FrankelKamenica,BloedelZhong}.
\end{remark}

When the cost function is posterior separable, then the agent's payoff from choosing signal $\sigma:\Theta \rightarrow \Delta(S)$ and strategy $\alpha: S \rightarrow \Delta(A)$ is
\[\int_{\Delta(\Theta)}  \int_{a \in A} \alpha(a \mid q) \mathbb{E}_q[u(a,\theta)] d\tau_\sigma(q) - C(p,\tau_\sigma),\]
and can be rewritten as
\[\int_{\Delta(\Theta)} \int_{a \in A} \alpha(a \mid q) \left( \mathbb{E}_q[u(a,\theta)] - \Phi_p(q)\right) d\tau_\sigma(q)\]
where the concave function $\mathbb{E}_q[u(a,\theta)] - \Phi_p(q)$ is the ``net utility" of action $a$ under posterior $q$. So maximizing the value function is equivalent to maximizing the expected net utility over all Bayes-plausible distributions and strategies,  which is an optimization problem that can be solved using standard methods. This tractability is a part of the appeal of this family of cost functions.

A closely related concept appears in \citet{FrankelKamenica}, where $\Phi_p(q)$ is interpreted as the  amount of information in news that moves an agent's belief from $p$ to $q$. \citet{FrankelKamenica} define the pair $(\Phi_p,\Phi)$ as \emph{coupled} if $\mathbb{E}[\Phi_p(q)] = \mathbb{E}[\Phi(p)-\Phi(q)]$, in which case the cost function is not only posterior separable but also uniformly posterior separable.

That uniform posterior separability is strictly stronger than posterior separability is nearly immediate, except for the requirement in the definition of posterior separable cost functions that $\Phi_p(q)$ is everywhere positive. We cannot therefore simply convert a UPS cost function $C(p,\tau) = \Phi(p) - \mathbb{E}_{q \sim \tau}[\Phi(q)]$ into a posterior separable cost function $C(p,\tau) = \mathbb{E}[\Phi_p(q)]$ by setting  $\Phi_p(q) \equiv \Phi(p ) - \Phi(q)$, as this quantity may be negative for some posterior beliefs $q$. The correct construction is instead to choose  $\Phi_p$ to be a \emph{Bregman divergence} of $\Phi$ \citep{FrankelKamenica,CaplinDeanLeahy2022}.

\begin{definition} Let $\Phi: \Delta(\Theta) \rightarrow \mathbb{R}$ be a concave function. A \emph{supergradient} of $\Phi$ at $p \in \Delta(\Theta)$ is any vector $\nabla \Phi(p)$ such that
\[\Phi(p) + \nabla  \Phi(p) \cdot (q-p) \geq \Phi(q)\]
for every $q \in \Delta(\Theta)$.
\end{definition}

\begin{remark} When $\Phi$ is concave, then a supergradient $\nabla \Phi(q)$ exists for every $q$. When $\Phi$ is smooth at $q$, then $\nabla \Phi(q)$ is unique and equal to $\Phi'(q)$. 
\end{remark}

\begin{definition}
Let $\Phi: \Delta(\Theta) \rightarrow \mathbb{R}$ be a concave function.  A \emph{Bregman divergence} of $\Phi$ is any map $D_\Phi: \Delta(\Theta) \times \Delta(\Theta) \rightarrow \mathbb{R}$ satisfying
\[D_\Phi(p,q) = \Phi(p) - \Phi(q) + \nabla \Phi(p) \cdot (q-p) \quad \forall (p,q) \in \Delta(\Theta) \times \Delta(\Theta)\]
where $\nabla \Phi(q)$ is a supergradient of $\Phi$ at $q$. \label{def:Bregman}
\end{definition}
\noindent This is the difference between the value of $\Phi$ at $q$ and the value of the first-order Taylor expansion of $\Phi$ around $p$ evaluated at point $q$. 

Setting $\Phi_p(q) = D_\Phi(p,q)$ from Definition \ref{def:Bregman}, we have 
\[\Phi_p(q) = \left(\Phi(p)  + \nabla \Phi(p) \cdot (q-p)\right) - \Phi(q) \geq 0 \quad \forall q \in \Delta(\Theta)\]
since $\nabla \Phi(p)$ is a supergradient of $\Phi$, and also 
\begin{align*}
\mathbb{E}_{q \sim \tau} [\Phi_p(q)] & = \mathbb{E}_{q \sim \tau}[\Phi(p) - \Phi(q) + \nabla \Phi(p) \cdot (q-p)] \\
& = \Phi(p) - \mathbb{E}_{q\sim \tau}[\Phi(q)]
\end{align*}
using in the second inequality that $\mathbb{E}_{q \sim \tau}(q-p)=0$. The relationship between $\Phi_p$ and $\Phi$ is depicted in Figure \ref{fig:Bregman}.

\begin{figure}[H]
\begin{center}
\includegraphics[scale=0.65]{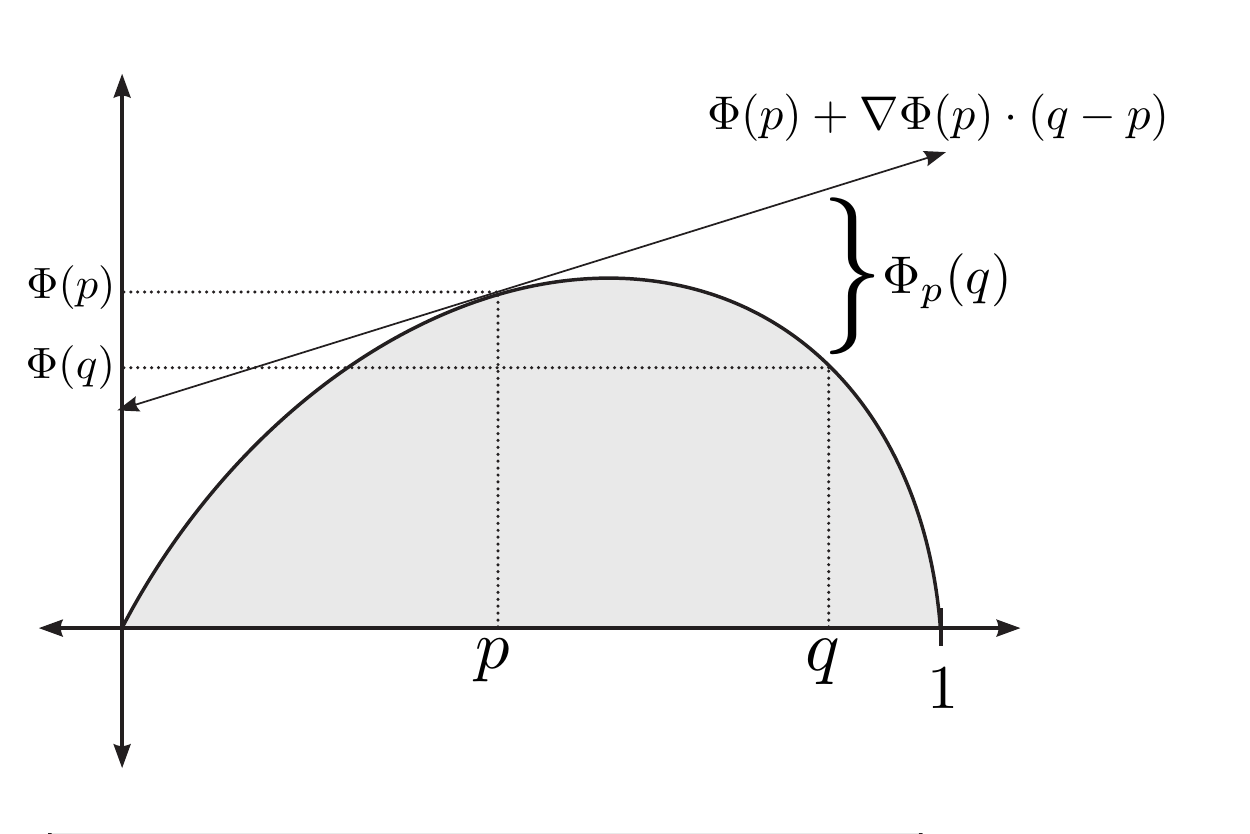}
\end{center}
\caption{Relationship between $\Phi_p$ and $\Phi$.} \label{fig:Bregman}
\end{figure}

\begin{example}
Consider entropy cost $C_{\text{Ent}}(p, \tau) = H(p) - \mathbb{E}_{q \sim \tau}[H(q)]$. The Bregman divergence of entropy is KL divergence \citep{Bregman1967}, so 
\[C_{\text{Ent}}(p, \tau) = H(p) - \mathbb{E}_{q \sim \tau}[H(q)] = \mathbb{E}[D(p \| q)].\]
Thus we can view the cost of a signal that generates the distribution of beliefs $\tau$ either as the expected reduction in the entropy of the agent's belief, or as the expected KL divergence from the agent's prior to the realized posterior belief.
\end{example}

\section{Prior-Independent Costs} \label{sec:PriorIndependent}

We now turn to cost functions that do not depend on the agent's prior belief. If the cost of information is exogenous to the agent---for example, a price determined within a market, or a physical cost of producing information---then we may expect the cost of acquiring information to be the same for all consumers regardless of their beliefs or expertise in the area, and thus prior independent.

One common cost specification is the following.

\begin{example} \label{ex:GaussianPrecision} In the setting of Example \ref{ex:GaussianCost}, let
\begin{equation} \label{cost:Precision}
C(\sigma_\eps^2) = \frac{\kappa}{\sigma_\eps^2}
\end{equation}
Then the cost of the signal scales linearly with the precision of the signal, $1/\sigma_{\eps}^2$. This formulation of the cost is especially sensible if we interpret $\theta$ as an unknown population parameter (for instance, the average height in a population) and the signal as a sample of individuals from this population. Modeling each observation as $X_i = \theta + \eps_i$ with  $\eps_i \sim \mathcal{N}(0,\sigma^2)$ independent of $\theta$ and independent across agents, the conditional distribution of $\theta$ given the sample $(X_1, \dots, X_n)$ is the same as the conditional distribution of $\theta$ given the signal $X = \theta + \delta, \,\, \delta \sim \mathcal{N}(0, \sigma^2/n)$ (see Exercise \ref{ex:Average}). So (\ref{cost:Precision}) corresponds to a fixed cost of $\kappa/\sigma^2$ for each individual in the sample. This cost function is used in Wald's classic model of sequential sampling \citep{Wald,ArrowBlackwellGirshick}, and is a common modeling choice in continuous-time sequential sampling problems where the signal corresponds to observation of a Brownian motion \citep{FudenbergStrackStrzalecki,LiangMuSyrgkanis}.
\end{example}

We now present a generalization of the above cost function due to \citet{PomattoStrackTamuz}. Let $\Theta$ be a finite set and $S$ be a set of signal realizations equipped with $\sigma$-algebra $\Sigma$, with $\Delta(S)$ denoting the set of measurable probability distributions on $S$. A signal is a mapping $ \sigma: \Theta \to \Delta(S)$, and we use $\sigma_\theta \equiv \sigma( \cdot \mid \theta) \in \Delta(S)$ to denote the conditional distribution over signal realizations when the state is $\theta$. 

\begin{definition} The log-likelihood ratio between states $\theta$ and $\theta'$ at signal realization $s$ is
\[\ell^\sigma_{\theta,\theta'}(s) = \ln\left(\frac{d\sigma_\theta(s)}{d\sigma_{\theta'}(s)}\right)\]
\end{definition} 

\begin{definition} For any state $\theta \in \Theta$ and map $\alpha: \Theta \rightarrow \mathbb{N}$, define
\[M_\theta^\sigma(\alpha) = \int_S \left\vert \prod_{\theta' \neq \theta} \left(\ell_{\theta,\theta'}^\sigma(s)\right)^{\alpha(\theta')} \right\vert d\sigma_\theta\] 
\end{definition}
 
\begin{assumption} \label{assp:FiniteMoment} The expectation $M_\theta^\sigma(\alpha)$ is finite for every $\theta$ and every $\alpha: \Theta \rightarrow \mathbb{N}$.
\end{assumption}

This assumption says that the log-likelihood ratios have finite moments, ruling out for example the signal structure
\[\begin{array}{ccc}
& s_1 & s_2 \\ 
\theta_1 & 0 & 1 \\
\theta_2 & \frac12 & \frac12 
\end{array}\]
where the signal realization $s_1$ is perfectly revealing of the state $\theta_2$.

Let $\mathcal{E}$ be the class of all signals satisfying Assumption \ref{assp:FiniteMoment}. An \emph{information cost function} is any map $C: \mathcal{E} \rightarrow [0,\infty)$. \citet{PomattoStrackTamuz} propose four axioms that such a cost function should further satisfy.

\begin{axiom}[Consistency with the Blackwell order] \label{axiom:Blackwell} If $\sigma$ dominates $\sigma'$ in the Blackwell order, then $C(\sigma)\geq C(\sigma')$. 
\end{axiom}

That is, more informative signals are more costly to acquire.

\begin{definition}[Combining Independent Signals] For any two signals $\sigma: \Theta \rightarrow \Delta(S)$ and $\sigma': \Theta \rightarrow \Delta(S')$, let $\sigma \otimes \sigma'$ denote the product signal 
\[\sigma \otimes \sigma': \Theta \rightarrow \Delta(S \times S')\]
where $\sigma \otimes \sigma'(s,s' \mid \theta)=\sigma(s\mid \theta)\sigma(s'\mid \theta)$.
\end{definition}

\begin{axiom}[Additivity with respect to independent experiments] \label{axiom:AdditiveCost} For any two signals $\sigma$ and $\sigma$', $C(\sigma \otimes \sigma')= C(\sigma) + C(\sigma')$.
\end{axiom}

That is, the cost of acquiring two (conditionally) independent signals is equal to the sum of their costs. This axiom imposes a constant marginal cost on information similar to the one used to motivate Example \ref{ex:GaussianPrecision}.

\begin{definition}[Diluting Signals] For any signal $\sigma$, the $\alpha$-\emph{dilution} of $\sigma$, denoted $\alpha \cdot \sigma$, is a signal where with probability  $\alpha$ the realization of $\sigma$ is observed, and otherwise a completely uninformative signal is observed. Formally, $\alpha \cdot \sigma$ is a map from $\Theta$ to $S \cup \{\emptyset\}$ where the signal outcome $\emptyset$ has a constant $1-\alpha$ probability at every state $\theta \in \Theta$, and the remaining probability is assigned to $S$ in proportion to $\sigma$. 
\end{definition}

\begin{axiom}[Linearity in the ``dilution" of the experiment] \label{axiom:LinearDilution} $C(\alpha \cdot \sigma) = \alpha \cdot C(\sigma)$ for every signal $\sigma$ and weight $\alpha \in [0,1]$. 
\end{axiom}

That is, the cost of a signal is linear in the probability that it generates information. 

\begin{remark} Every posterior separable cost function $C(p,\tau) = \mathbb{E}_{q \sim \tau}[\Phi_p(q)]$ satisfies Axiom \ref{axiom:LinearDilution}. To see this, observe that the distribution over posterior beliefs given the diluted signal $\alpha \cdot \sigma$, denoted $\tau_{\alpha \cdot \sigma}$, is the convex combination that puts weight $\alpha$ on the distribution $\tau_\sigma$ generated by $\sigma$, and weight $1-\alpha$ on the prior. So
\begin{align*}
C(p,\tau_{\alpha \cdot \sigma}) & = \mathbb{E}_{q \sim \alpha \tau_\sigma + (1-\alpha) \delta_p}[\Phi_p(q)] \\
& = \alpha \mathbb{E}_{q \sim \tau_\sigma}[\Phi_p(q)] + (1-\alpha) \Phi_p(p)\\
& = \alpha \cdot C(p,\tau_\sigma)
\end{align*}
where the second equality uses that $C$ is affine in $\tau$ and the third uses that $\Phi_p(p)=0$ in the definition of a posterior separable cost function.
\end{remark}

The final axiom imposes continuity of the cost function with respect to a nonstandard (pseudo)-metric given below.\footnote{This is a pseudometric rather than a metric, since  $d_N(\sigma,\sigma')$ is equal to zero for $\sigma \neq \sigma'$ if they induce the same distribution over posterior beliefs.} 

\begin{definition} Given an upper bound $N \geq 1$, define 
\[d_N(\sigma,\sigma') = \max_{\theta \in \Theta} d_{TV}(\sigma_\theta,\sigma'_\theta) + \max_{\theta \in \Theta} \max_{\alpha \in \{0, \dots, N\}^n} \vert M_\theta^\sigma(\alpha) - M_\theta^{\sigma'}(\alpha)\vert\]
where $d_{TV}$ denotes the total variation distance.
\end{definition}

Two signals $\sigma$ and $\sigma'$ are close under this pseudo-metric if for every state $\theta$, the induced distributions of log-likelihood ratios are close in total-variation distance and additionally have similar moments, for any vector of moments lower or equal to $(N,\dots, N)$.  

\begin{axiom}[Continuity.] \label{axiom:Continuity} The function $C$ is uniformly continuous with respect to $d_{N}$.
\end{axiom}

\begin{remark} The topology of weak convergence of likelihood ratios and the topology of convergence of likelihood ratios in total variation distance are both more standard. But no cost function which satisfies Axioms \ref{axiom:Blackwell}-\ref{axiom:LinearDilution} is continuous in these alternative topologies. To see this, let $\theta$ be the unknown bias of a coin, and let $\sigma_n$ be the signal where with probability $1/n$ the outcome of $n$ independent flips of this coin is observed, and otherwise no information is revealed. Axioms \ref{axiom:Blackwell}-\ref{axiom:LinearDilution} imply that $C(\sigma_n) = C(\sigma_{n'})$ for all finite $n,n'$. But the likelihood ratios of these signals converge in the weak topology (and in the total variation topology) to those of the signal that produces no information, and thus a stronger form of Axiom \ref{axiom:Continuity} based on either of these alternative topologies would require these signals to all have zero cost.  

\end{remark}

\begin{proposition} \label{prop:PST} The cost function $C: \mathcal{E} \rightarrow \mathbb{R}$ satisfies Axioms \ref{axiom:Blackwell}-\ref{axiom:Continuity} if and only if there exists a unique collection of $\mathbb{R}_+$-valued parameters $(\beta_{\theta,\theta'})_{\theta, \theta' \in \Theta}$ such that
\begin{equation} \label{eq:PST}
C(\sigma) = \sum_{\theta,\theta' \in \Theta} \beta_{\theta,\theta'} \times  \underbrace{\int_S \ln \frac{d\sigma_\theta(s)}{d\sigma_{\theta'}(s)} d\sigma_\theta(s)}_{\mbox{KL-divergence from $\sigma(\cdot \mid \theta')$ to $\sigma(\cdot \mid \theta)$}}
\end{equation}
\end{proposition}

As discussed in Section \ref{sec:KL}, the KL-divergence from $\sigma( \cdot \mid \theta')$ to $\sigma(\cdot \mid \theta)$ is a measure of how different the distributions are. The larger this divergence is, the easier it is to reject the hypothesis that the state is $\theta'$ when it truly is $\theta$.

\begin{remark} Axiom \ref{axiom:Continuity} can be dispensed with if $\Theta = \{\theta_0,\theta_1\}$, in which case Proposition \ref{prop:PST} simplifies to the statement that $C$ satisfies Axioms \ref{axiom:Blackwell}-\ref{axiom:LinearDilution} if and only if there exist parameters $\beta_{01},\beta_{10} \geq 0$ such that
\[C(\sigma) = \beta_{01} D(\sigma(\cdot \mid \theta_0) \| \sigma(\cdot  \mid \theta_1)) + \beta_{10} D(\sigma(\cdot \mid \theta_1) \| \sigma(\cdot  \mid \theta_0)).\]
\end{remark} 

A notable contrast with entropy cost is that this cost function permits differentiation between states.

\begin{example}[\citet{PomattoStrackTamuz}] \label{ex:DistinguishStates} Suppose the unknown state $\theta$ is the US GDP per capita, and the agent holds a uniform prior over $\Theta = \{20,000, \dots, 80,000\}$. Then under entropy cost $C_{Ent}$, it is equally costly to acquire the signal that reveals whether $\theta$ is above or below \$50,000, or the signal that reveals whether $\theta$ is even or odd.
\end{example}

\noindent The free parameters $\beta_{\theta,\theta'}$ in the  representation in (\ref{eq:PST})  reflect potentially different costs to distinguishing between different pairs of states. Specifically, we can interpret each $\beta_{\theta,\theta'}$ as the marginal cost of increasing the expected log-likelihood ratio of a signal with respect to states $\theta$ and $\theta'$ (when $\theta$ is the true state). Thus in Example \ref{ex:DistinguishStates}, we may specify (for example) that it is easier to distinguish between states that are far apart than those that are nearby, i.e., if GDP is in fact 80,000 then it is easier to rule out that GDP is 20,000 than it is to rule out that it is 79,999.  In the special case where no pair of states is a priori harder to distinguish than another, then all coefficients are equal to one another.

\begin{example} Returning to the setting of Example \ref{ex:GaussianCost}, where we now use $C(\sigma_\eps^2)$ to mean the cost of acquiring the signal $X=\theta +\eps$, $\eps \sim \mathcal{N}(0,\sigma_\eps^2)$, we have
\[C(\sigma_\eps^2) = \sum_{\theta,\theta' \in \Theta} \beta_{\theta,\theta'} \frac{(\theta - \theta')^2}{2\sigma_\eps^2}.\]
This nests the precision of the signal $(1/\sigma_\eps^2)$ as a special case when 
$\beta_{\theta,\theta'} = \frac{1}{(\theta - \theta')^2}$, with the  interpretation that  states that are closer (in squared distance) are harder to distinguish.
\end{example}

\begin{remark} The class of cost functions identified in Proposition \ref{prop:PST} does not presuppose that the agent is Bayesian and has a prior belief over the state space. But if the agent does have a prior $p$, then the cost of the signal that induces distribution $\tau$ over posterior beliefs can be restated as
\begin{equation} \label{eq:PSVersion}
\mathbb{E}_{q \sim \tau}[\Phi_p(q)]
\end{equation}
where 
\begin{equation} \label{eq:Phi}
\Phi_p(q) = \Phi(p) - \sum_{\theta,\theta'} \beta_{\theta,\theta'} \frac{q_\theta}{p_\theta} \ln \left(\frac{q_\theta}{q_{\theta'}}\right)
\end{equation}
so this family of cost functions belongs to the class of posterior-separable cost functions (Definition \ref{def:PosteriorSeparable}), although not to the class of uniform posterior separable cost functions (Definition \ref{def:UPS}).\footnote{\citet{PomattoStrackTamuz} show that a generalization of the representation in (\ref{eq:PST}), which permits the parameters $\beta_{\theta,\theta'}$ to depend on the prior, can accommodate uniformly posterior separable cost functions.}

\begin{exercise}[G] Verify that (\ref{eq:PSVersion}) is equivalent to the original representation in (\ref{eq:PST}) when $\Phi$ is defined according to (\ref{eq:Phi}).
\begin{hint} Recall from Section \ref{sec:Bayes} that the prior $p$ and posterior $q$ at signal realization $s$ are related by $\log\left(\frac{q(\theta)}{q(\theta')}\right) = \log\left(\frac{p(\theta)}{p(\theta')}\right) + \log\left(\frac{d\sigma_\theta}{d\sigma_{\theta'}}(s)\right)$. 
\end{hint}
\end{exercise} 
\end{remark}

\section{Additional Exercises}

\begin{exercise}[G] Suppose $p,q \in \Delta(\mathcal{X} \times \mathcal{Y})$ with $p_X$ and $q_X$ denoting the marginal distributions on $\mathcal{X}$, and $p_{Y \mid X}$ and $q_{Y \mid X}$ denoting the respective conditional distributions. Prove that
\[D(p \| q) = D(p_X \| q_X) + D(p_{Y \mid X} \| q_{Y \mid X}).\]
This is known as the chain rule for KL divergence.

\end{exercise}

\begin{exercise}[G] Prove that the entropy cost function in Definiton \ref{def:EntropyCost} fails \citet{PomattoStrackTamuz}'s Axiom \ref{axiom:AdditiveCost}.
\end{exercise}

\part{\sc{Learning}}

\chapter{Learning} \label{sec:Learning}

We now extend the Bayesian framework described in \ref{sec:Preliminaries} to accommodate learning from a sequence of signals. Section \ref{sec:Doob} asks whether an agent will eventually learn the state. Section \ref{sec:Merging} asks whether agents with different prior beliefs will eventually hold similar beliefs. Section \ref{sec:KLS} asks whether agents with different priors expect their disagreement to reduce given information (thus studying a second-order belief). Section \ref{sec:CommonLearning} asks whether agents will commonly learn, i.e., whether agents will eventually believe that other agents believe that they ... have learned the state.

\section{Preliminaries} \label{sec:LearningFramework}

Let $(\Theta,d_\Theta)$ be a complete separable metric space endowed with its Borel $\sigma$-algebra $\Sigma$, and let $p \in \Delta(\Theta)$ be a ($\Sigma$-measurable) probability measure on $\Theta$. As before, we interpret $\theta \sim p$ as an unknown parameter of interest. 

The space of signal realizations $(\mathcal{X}, d_X)$ is again a complete separable metric space  endowed with its Borel $\sigma$-algebra $\mathcal{B}$.  There is an infinite sequence of signal realizations $X_1, X_2, \dots$ taking values in the set $\mathcal{X}^\infty = \mathcal{X}_1 \times \mathcal{X}_2 \times \dots$ where each $\mathcal{X}_t$ is a copy of $\mathcal{X}$. Conditional on the realized $\theta$, signals $X_1, X_2, \dots$ are generated iid according to a conditional density $f_\theta$, and we refer to each $X_t$ as the period-$t$ signal.

The full state space is $\Omega = \Theta \times \mathcal{X}^\infty = \Theta \times \mathcal{X}_1 \times \mathcal{X}_2 \times \dots$ and it is equipped with the product $\sigma$-algebra $\Sigma \times \mathcal{B}_1 \times \mathcal{B}_2 \times \dots$ where each $\mathcal{B}_t$ is a copy of $\mathcal{B}$. Throughout, we use $P$ to denote the measure on $\Omega$ induced by $p$ and the family $(f_\theta)_{\theta \in \Theta}$, and we use $P_\theta$ to denote the conditional measure on $\mathcal{X}^\infty$ when the parameter is $\theta$.

\section{Binary Example} \label{sec:LearningExample}

First consider a single-agent environment with  two possible parameter values $\theta \in \{A,B\}$. Each period $t\in \mathbb{Z}_+$ a signal realization from $\{a,b\}$ is generated iid according to
\[\begin{array}{ccc}
& a & b\\
A & q& 1-q \\
B & 1-q & q
\end{array}\]
where $q>1/2$. Will an agent who holds a prior belief that the probability of $A$ is $p\in (0,1)$ eventually learn the value of the parameter?

Suppose first that the parameter is $\theta=A$, in which case signals are drawn iid according to $f_A = (q,1-q)$. For any infinite sequence $\bold{x} \in \{a,b\}^\infty$ and any $t \in \mathbb{Z}_+$, let
\[n_t(\bold{x}) \equiv \#\{ 1 \leq t' \leq t : x_{t'} =a\}\]
denote the number of $a$-realizations among the first $t$ realizations of $\bold{x}$. By the strong law of large numbers, there is a set $\mathcal{X}_0^\infty \subseteq \mathcal{X}^\infty$ of $P_A$-measure 1 such that 
\begin{align*} 
\lim_{t \rightarrow \infty} \frac{n_t(\bold{x})}{t} = q \quad \forall \bold{x} \in \mathcal{X}_0^\infty.
\end{align*}
That is, the limiting fraction of $a$-realizations is $q$ along each sequence in $\mathcal{X}_0^\infty$.

Since signals are assumed to be conditionally independent, the agent's posterior belief about $A$ following any sequence $(x_1, \dots, x_t)$ depends only on the count of $a$ and $b$-realizations. Let $n$ denote the number of $a$-realizations. Then applying Bayes' rule (Section \ref{sec:BayesRule}), the agent's posterior belief is \begin{align}
P(\theta = A \mid x_1, \dots, x_t) & = \frac{p q^n (1-q)^{t-n}}{p q^n (1-q)^{t-n} + (1-p) (1-q)^n q^{t-n}} \nonumber \\
& = \frac{1}{1 + \frac{1-p}{p} \left(\frac{1-q}{q}\right)^{2n-t}} \label{eq:PosteriorBelief}
\end{align}
Along any $\bold{x} \in \mathcal{X}_0^\infty$ we have 
\[\lim_{t \rightarrow \infty} P(\theta=A \mid x_1, \dots, x_t) = \lim_{t \rightarrow \infty} \left(1+\frac{1-p}{p} \left[\left(\frac{1-q}{q}\right)^{2\frac{n_t(\bold{x})}{t} - 1}\right]^t \right)^{-1} = 1\]
recalling that $q>1/2$ by assumption.

So the agent's posterior belief $P_A$-almost surely converges to certainty of the correct value of the parameter, $A$. An identical argument shows that when the parameter is $B$ then the agent's posterior belief $P_B$-almost surely converges to certainty of $B$. Thus the agent (eventually) learns the parameter.

\section{Doob's Consistency Theorem} \label{sec:Doob}

A classic result due to \citet{Doob} generalizes the individual learning result from the previous section.\footnote{Our presentation of this material follows \citet{Miller2018}.}

\begin{assumption}[Identifiability] \label{assp:Identifiability}
If $\theta \neq \theta'$, then $P_\theta \neq P_{\theta'}$.
\end{assumption}

In words, Assumption \ref{assp:Identifiability} is satisfied if no pair of parameter values induce the same distribution over signals, meaning the parameter is identifiable from its observable implications.

\begin{proposition} \label{prop:Doob1} Suppose Assumption \ref{assp:Identifiability} is satisfied, and let $g:\Theta \rightarrow \mathbb{R}$ be any measurable function satisfying $\mathbb{E}\vert g(\theta)\vert < \infty$. Then 
\[\lim_{t \rightarrow \infty} \mathbb{E}(g(\theta) \mid X_1, X_2, \dots, X_t) = g(\theta) \quad P\text{-a.s.}\]
\end{proposition}

In the special case where $g(\theta) = \theta$, the result implies that the posterior expectation of $\theta$ converges to its true value almost surely. The following proposition is a Bayesian analogue of the above result, and says that posterior beliefs converge almost surely to a degenerate measure at the true state.

\begin{proposition}[Posterior Consistency] \label{prop:PosteriorConsistency} Suppose Assumption \ref{assp:Identifiability} holds. Then, there exists a set $\Theta' \subseteq \Theta$ with $p(\Theta')=1$ such that for every $\theta_0 \in \Theta'$ and every neighborhood $B$ of $\theta_0$,
\[\lim_{t \rightarrow \infty} \mathbb{P}(\theta \in B \mid X_1, X_2, \dots, X_t) =  1 \quad P_{\theta_0} \text{-a.s.}\]
\end{proposition}

That is, for any prior distribution, the posterior belief is guaranteed to concentrate in a neighborhood of the true parameter $\theta$---except possibly on a set of parameter values that has measure zero under the agent's prior.

\begin{remark} The qualification that learning occurs except on a set of ``measure zero under the agent's prior" is less harmless than it might initially seem. Consider $\Theta = \mathbb{R}$ where the agent's prior $p \in \Delta(\Theta)$ is a point mass at $\theta=0$. Then the posterior is also a point mass at zero, so the agent will fail to learn any parameter which is different from $0$. But because the set $\mathbb{R} \backslash \{0\}$ has measure zero under the agent's prior, the statement of the result holds in a trivial sense. See also the subsequent discussion in Section \ref{sec:Berk}.
\end{remark} 

\begin{remark} Proposition \ref{prop:PosteriorConsistency} implies that the agent's posterior belief converges almost surely to a point mass on the true parameter in the topology of weak convergence, i.e., there is a $P_\theta$-measure 1 set of sequences of signal realizations such that 
\[d(P^t,\delta_\theta) \rightarrow 0\]
 along each of these sequences, where $d$ denotes the Levy-Prokhorov metric and $P^t \in \Delta(\Theta)$ denotes the posterior belief after observing the first $t$ coordinates of the sequence. Since $d$ is a metric, we also have that for any alternative prior $\widetilde{p} \in \Delta(\Theta)$ and corresponding posterior belief $\widetilde{P}^t \in \Delta(\Theta)$ (updating to the same $t$ realizations),
\[d(P^t,\widetilde{P}^t) \leq d(P^t,\delta_\theta) + d(\delta_\theta,\widetilde{P}^t).\]
Since the RHS converges to zero almost surely (by Proposition \ref{prop:PosteriorConsistency}), the two agents' posterior beliefs converge to one another almost surely in the topology of weak convergence. The subsequent section provides an even stronger version of this result.
\end{remark}

\section{Merging of Beliefs} \label{sec:Merging}

Assume that for each $t \geq 1$, a unique conditional probability distribution $P^t(x_1, \dots, x_t)(C)$  exists for all realized sequences $x_1, \dots, x_t \in \mathcal{X}_1 \times \dots \times \mathcal{X}_t$ and unknown events $C \in \mathcal{B}_{t+1} \times \mathcal{B}_{t+2} \times \dots $.\footnote{\citet{BlackwellDubins} work with the more general notion of ``predictive probabilities" $P$ where conditional probabilities can be defined.}  \citet{BlackwellDubins} show that even if players start out with different prior beliefs, their conditional beliefs will merge to one another in a strong sense. 

To state the result formally, recall that for any two probability measures $\mu_1,\mu_2$ defined on the same $\sigma$-algebra $\mathcal{F}$,  \emph{total variation distance} and \emph{absolute continuity} are defined as follows.

\begin{definition} The \emph{total variation distance} between $\mu_1$ and $\mu_2$ is
		\[d_{TV}(\mu_1, \mu_2) = \sup_{ D \in \mathcal{F}} \vert \mu_1(D) - \mu_2(D) \vert\]
\end{definition}

\begin{definition} If $\mu_2(D)=0$ implies $\mu_1(D)=0$ for every $D \in \mathcal{F}$, then $\mu_1$ is \emph{absolutely continuous} with respect to $\mu_2$, denoted $\mu_1 \ll \mu_2$. \end{definition}

Now we are ready to state the main result:
\begin{proposition} \label{prop:BlackwellDubins} Suppose $p,\widetilde{p} \in \Delta(\Theta)$ are absolutely continuous with respect to one another, and define $P$, $\widetilde{P}$ to be the measures on $\Omega$ induced by the respective priors $p,\tilde{p}$, and the family $(P_\theta)_{\theta \in \Theta}$. Then \[ \lim_{t \rightarrow \infty} d_{TV}(P^t(x_1, \dots, x_t), \widetilde{P}^t(x_1, \dots, x_t))=0  \quad P\mbox{-almost surely}\]
\end{proposition}

That is, if two agents hold different prior beliefs about the parameter but agree on the set of measure-0 events, then their conditional beliefs merge in a strong sense: For \emph{all} measurable future events, agents eventually assign similar probabilities.

\begin{example} To clarify the difference between this result and the one examined in the previous section, consider the problem of learning the unknown bias of a coin, which is parametrized to $p \in [0,1]$. A coin whose bias is $p$ lands on Heads with probability  $p$ and lands on Tails with probability $1-p$. Two agents have different prior beliefs on $[0,1]$ and each observe $t$ independent flips of this coin.  

Proposition \ref{prop:PosteriorConsistency} says that the two agents will eventually learn the bias of the coin as $t$ grows large. Proposition \ref{prop:BlackwellDubins} says instead: Suppose the two agents have observed $t$ independent flips of the coin; then, their beliefs over all events regarding the future---e.g., that over half of the remaining coin flips will turn up Heads, or that the limiting fraction of Heads realizations is 1/2---must eventually become close (uniformly across such events).
\end{example}

\section{(Expected) Disagreement} \label{sec:KLS}

We now turn to the impact of information on agents' second-order beliefs---i.e., what they think about what others think. \citet{KartikLeeSuen} show that when signals satisfy an MLRP condition, then agents with different beliefs expect information to reduce the extent of disagreement. 

Here we assume the set of parameters $\Theta \subseteq \mathbb{R}$ is finite and ordered. Two signals $X$ and $\widetilde{X}$ respectively take values in $\mathcal{X}$ and $\mathcal{\widetilde{X}}$, and we assume that $X$ is Blackwell more informative than $\widetilde{X}$. There are two agents, Ann and Bob, who have common knowledge of  the conditional distributions $\{f_{X \mid \theta}(x \mid \theta)\}_{\theta \in \Theta}$ and $\{f_{\widetilde{X}\mid \theta}(\widetilde{x}\mid \theta)\}_{\theta \in \Theta}$. But Ann and Bob hold different prior beliefs $f_\theta^A,f_\theta^B \in \Delta(\Theta)$ about the parameter. We use $F^A$ and $F^B$ to denote their perceived joint distributions of $(\theta,X,\widetilde{X})$ (induced by the respective priors and the common knowledge signal distributions), and $\mathbb{E}_A$ and $\mathbb{E}_B$ to denote expectations with respect to these distributions.

\begin{assumption} \label{assp:MLRP} There is an order $\succ$ on $\mathcal{X}$ and an order $\widetilde{\succ}$ on $\mathcal{\widetilde{X}}$ such that the families $\{f_{X\mid \theta}(\cdot \mid \theta)\}_{\theta \in \Theta}$ and $\{f_{\widetilde{X}\mid \theta}(\cdot \mid \theta)\}_{\theta \in \Theta}$ each have MLRP (see Definition \ref{def:MLRP}).
\end{assumption}

\begin{assumption} \label{assp:LR} Bob's prior $f_\theta^B$ likelihood-ratio dominates Ann's prior $f_\theta^A$ (see Definition \ref{def:LRDominance}).
\end{assumption}

The agents' prior expectations of the parameter are $\mu_A \equiv \mathbb{E}_A(\theta)$ and $\mu_B \equiv \mathbb{E}_B(\theta)$. We are interested in Ann's prior expectation of Bob's posterior expectation (updated to $X$), and Bob's prior expectation of Ann's posterior expectation (updated to $X$), respectively denoted by 
\begin{align*}
\mu_{AB}(X) &\equiv \mathbb{E}_A[\mathbb{E}_B(\theta \mid X)] \\
\mu_{BA}(X) &\equiv \mathbb{E}_B[\mathbb{E}_A(\theta \mid X)]
 \end{align*}

\begin{proposition} \label{prop:KLS} Suppose Assumptions \ref{assp:MLRP} and \ref{assp:LR} are satisfied. If $X$ is Blackwell more informative than $\widetilde{X}$, then
\[\mu_A \leq \mu_{AB}(X) \leq \mu_{AB}(\widetilde{X}) \leq \mu_B\]
\[\mu_A \leq \mu_{BA}(\widetilde{X}) \leq \mu_{BA}(X) \leq \mu_B\]
\end{proposition}

\noindent That is, Ann expects that a more informative experiment will, in expectation, bring Bob's posterior mean closer to Ann's prior, and vice versa. These are both subjective statements, and indeed only one of Ann and Bob can be correct.

We'll prove this proposition using the following relationships, which are left as an exercise.

\begin{exercise}[G] \label{exercise:KLS} Prove the following statements:
\begin{itemize}
\item[(a)] $F^B_{\theta \mid X}(\theta \mid X=x)$ first-order stochastically dominates $F^A_{\theta \mid X}(\theta \mid X=x)$ for every signal realization $x \in \mathcal{X}$
\item[(b)] $F^B_{X\mid \widetilde{X}}(X \mid \widetilde{X}=\tilde{x})$ first-order stochastically dominates $F^A_{X \mid \widetilde{X}}(X \mid \tilde{X}=\tilde{x})$ for every signal realization $\tilde{x} \in \widetilde{\mathcal{X}}$
\end{itemize}
\end{exercise}

\begin{proof} Part (a) of Exercise \ref{exercise:KLS} implies $\int \theta dF^A_{\theta \mid X}(\theta \mid x)  \leq \int \theta dF^B_{\theta \mid X}(\theta \mid x)$ for every realization $x$, so also
\begin{equation}\label{eq:KLS1}
\int \int \theta dF^A_{\theta \mid X}(\theta \mid x) dF^A_{X}(x) \leq \int \int \theta dF^B_{\theta \mid X}(\theta \mid x) dF^A_X(x).
\end{equation}
By assumption that $\{f_{X\mid\theta}(\cdot \mid \theta)\}_{\theta \in \Theta}$ has MLRP, the integral $\int \theta dF^B_{\theta \mid X}(\theta \mid x)$ is an increasing function of $x$.  Moreover, Part (b) of Exercise \ref{exercise:KLS} says that $F_X^B$ first-order stochastically dominates $F_X^A$ (taking $\widetilde{X}$ to be any constant signal). Thus
\begin{equation} \label{eq:KLS2}
\int \int \theta dF^B_{\theta \mid X}(\theta \mid x) dF^A_X(x) \leq  \int \int \theta dF^B_{\theta \mid X}(\theta \mid x) dF^B_X(x).
\end{equation}
Together, (\ref{eq:KLS1}) and (\ref{eq:KLS2}) imply
\[\int \int \theta dF^A_{\theta \mid X}(\theta \mid x) dF^A_{X}(x) \leq \int \int \theta dF^B_{\theta \mid X}(\theta \mid x) dF^A_X(x) \leq  \int \int \theta dF^B_{\theta \mid X}(\theta \mid x) dF^B_X(x)\]
which is precisely the desired inequality $\mu_A \leq \mu_{AB}(X) \leq \mu_B.$ It follows by identical arguments that $\mu_A \leq \mu_{BA}(X) \leq \mu_B$. 

To show that $\mu_{AB}(\widetilde{X}) \geq \mu_{AB}(X)$, we use the fact that (since $X$ Blackwell-dominates $\widetilde{X}$) we can generate the two variables in such a way that $\widetilde{X}$ is conditionally independent of $\theta$ conditional on $X$.\footnote{See Remark \ref{remark:Garbling} for further detail. Note also that the correlation between $X$ and $\widetilde{X}$ is irrelevant for the comparison of $\mu_{AB}(X)$ and $\mu_{AB}(\widetilde{X})$.} Then on this probability space
\begin{align*}
\mu_{AB}(\widetilde{X}) &= \mathbb{E}_A \left[\mathbb{E}_B\left(\theta \mid \widetilde{X}\right) \right] \\[2mm]
& = \mathbb{E}_A\left[\mathbb{E}_B\left(\mathbb{E}_B\left(\theta \mid X, \widetilde{X}\right) \mid \widetilde{X}\right)\right] && \mbox{by L.I.E.} \\[2mm]
& = \mathbb{E}_A\left[\mathbb{E}_B\left(\mathbb{E}_B\left(\theta \mid X\right) \mid \widetilde{X}\right) \right] && \mbox{since $\widetilde{X}\indep \theta \mid X$} \\[2mm]
 &= \int \int \mathbb{E}_B(\theta \mid x) dF^B_{X \mid \widetilde{X}}(x \mid \widetilde{x}) dF_A(\widetilde{x}) \\[2mm]
 & \geq \int \int \mathbb{E}_B(\theta \mid x) dF^A_{X \mid \widetilde{X}} (x \mid \widetilde{x}) dF_A(\widetilde{x}) \\[2mm]
 & = \mathbb{E}_A\left[\mathbb{E}_A\left(\mathbb{E}_B\left(\theta \mid X\right) \mid \widetilde{X}\right)\right] \\
 & = \mathbb{E}_A\left[\mathbb{E}_B\left(\theta \mid X\right)\right]  && \mbox{by L.I.E.} \\
 & = \mu_{AB}(X)
\end{align*}
where the crucial inequality follows by observing that $\mathbb{E}_B(\theta \mid x)$ is an increasing function of $x$ (by Assumption \ref{assp:MLRP}) while $F_{X\mid \widetilde{X}}^B(\cdot \mid \widetilde{x})$ first-order stochastically dominates $F_{X\mid \widetilde{X}}^A(\cdot \mid \widetilde{x})$ for every realization of $\widetilde{x}$ (by Part (b) of Exercise \ref{exercise:KLS}).

Since the previous arguments apply to show also that $\mu_{AB}(\widetilde{X}) \leq \mu_B$, we are done. \end{proof}

\section{Common Learning} \label{sec:CommonLearning}

Suppose Assumption \ref{assp:Identifiability} (Identifiability) holds, so that agents eventually learn the true parameter. Does this imply that agents will eventually have \emph{common knowledge} of the true parameter? \cite{CEMS2008} adapt \citet{MondererSamet}'s definition of common $q$-belief for the present learning environment, and show that individual learning does imply common learning when the set of signal realizations is finite, but that this implication may otherwise fail.

In what follows recall that each state $\omega \in \Omega = \Theta \times \mathcal{X}^\infty$ describes both the value of the parameter and the infinite sequence of signal profiles. As before, $P_\theta$ denotes the measure on $\mathcal{X}^\infty$ conditional on parameter $\theta$, and again assume that $\Theta$ is finite. There are two agents $i=1,2$, and (different from the previous sections) we decompose $\mathcal{X} = \mathcal{X}^1 \times \mathcal{X}^2$ where $\mathcal{X}^i$ denotes the set of agent $i$ signal realizations. Each agent privately observes their own signal each period.  We use $h_{it}(\omega)=(x^i_{1}(\omega), \dots, x^i_{t}(\omega))$ for agent $i$'s history at time $t$ when $\omega$ is the realized state, and $\mathcal{H}_{it}$ to denote the filtration induced by agent $i$'s histories. 

\begin{definition}
 For any $q \in [0,1]$ and (measurable) event $F$,  agent $i$ \emph{$q$-believes} in $F$ at time $t$ on
\[B_{it}^q(F) = \{\omega \in \Omega \mid P(F \mid h_{it}(\omega)) \geq q\}\]
\end{definition}

\begin{definition} 
 For any $q \in [0,1]$, there is \emph{common $q$-belief} in $F$ at time $t$ on
\[C_t^q(F) = \bigcap_{n \geq 1} [B_t^q]^n(F)\]
where $B_t^q(F) = B_{1t}^q(F) \bigcap B_{2t}^q(F)$.
\end{definition}

\begin{definition}[Individual Learning] Agent $i$ \emph{learns} $\theta$ if for each $q \in (0,1)$ there exists $T<\infty$ such that
\[P_\theta(B_{it}^q(\{\theta\} \times \mathcal{X}^\infty)) > q \quad \forall t>T\]
Equivalently: $\lim_{t \rightarrow \infty} P_\theta(B_{it}^q(\{\theta\} \times \mathcal{X}^\infty))=1$ for all $q\in (0,1)$. Agent $i$ \emph{individually learns} if the agent learns each $\theta \in \Theta$.

\end{definition}

\begin{definition}[Common Learning] Agents \emph{commonly learn} $\theta$ if for each $q \in (0,1)$ there exists $T<\infty$ such that
\[P_\theta(C_{t}^q(\{\theta\} \times \mathcal{X}^\infty)) > q \quad \forall t>T\]
Equivalently: $\lim_{t \rightarrow \infty} P_\theta(C_{t}^q(\{\theta\} \times \mathcal{X}^\infty))=1$ for all $q\in (0,1)$. Agents \emph{commonly learn} if they commonly learn each $\theta \in \Theta$.
\end{definition}

Clearly if signals are perfectly correlated (or public), so that $P(\theta \mid \mathcal{H}_{1t}) = P(\theta \mid \mathcal{H}_{2t})$ for all $\theta$ and $t$, then individual learning implies common learning. This result also holds at the other extreme of independent signals.

\begin{proposition} \label{prop:CommonLearning} Suppose agents individually learn, and their signals are conditionally independent given the parameter. That is, there exist families $(P^i_\theta)_{\theta \in \Theta}$, with each $P_\theta^i \in \Delta(\mathcal{X}^i)$, such that $P_\theta(A \times B)=P_\theta^1(A)P_\theta^2(B)$ for each $\theta \in \Theta$ and measurable $A \subseteq \mathcal{X}^1$, $B \subseteq \mathcal{X}^2$. Then, agents commonly learn.
\end{proposition}

\cite{CEMS2008} proves this proposition using a result from \citet{MondererSamet} (adapted to the present learning context).

\begin{lemma} \label{lemm:MS} Agents commonly learn if and only if for every $\theta \in \Theta$ and $q \in (0,1)$, there is a sequence of events $F_t$ and a period $T$ such that for all $t>T$,
\begin{enumerate}
\item[(a)] $F_t \subseteq B_t^q(\theta)$ (``$\theta$ is $q$-believed on $F_t$ at time $t$'')
\item[(b)] $P_\theta(F_t)>q$ (``probability of $F_t$ is sufficiently high")
\item[(c)] $F_t \subseteq  B_{it}^q(F_t)$ for $i=1,2$ (``$F_t$ is evident $q$-belief at time $t$")
\end{enumerate}
\end{lemma}

We'll now prove Proposition \ref{prop:CommonLearning}.

\begin{proof} Henceforth write $\{\theta\}$ for the event $\{\theta\} \times \mathcal{X}^\infty$. Define $F_t = \{\theta\} \cap B_t^{\sqrt{q}}(\theta)$ to be the set of states at which $\theta$ is true and both agents $\sqrt{q}$-believe it. We'll verify that the conditions of Lemma \ref{lemm:MS} hold for the sequence of events $(F_t)_{t=1}^\infty$, from which Proposition \ref{prop:CommonLearning} follows.

First observe that
\begin{align*}
F_t & \subseteq B_t^{\sqrt{q}}(\theta) && \mbox{by definition of $F_t$}\\
& \subseteq B_t^q(\theta) && \mbox{since $q <\sqrt{q}$}
\end{align*}
yielding Part (a) of Lemma \ref{lemm:MS}. Part (b) holds since individual learning implies that there exists $T<\infty$ such that for both agents $i=1,2$,
\[P_\theta\left(B_{it}^{\sqrt{q}}(\theta)\right)>\sqrt{q} \quad \forall t>T\]
and thus
\[P_\theta(F_t) = P_\theta\left(B_{1t}^{\sqrt{q}}(\theta)\right)P_\theta\left(B_{2t}^{\sqrt{q}}(\theta)\right) > q \quad \forall t>T\]
 from the assumption of conditional independence. 

It remains to show Part (c). First rewrite the set $B_{1t}^q(F_t) $ as follows:
\begin{align*}
B_{1t}^q(F_t) & =  \left\{\omega \mid \mathbb{E}\left[\mathbbm{1}_{F_t} \mid \mathcal{H}_{1t}\right] \geq q) \right\} && \mbox{by definition of $B_{1t}^q$}\\
& = \left\{\omega \mid \mathbb{E}\left[\mathbbm{1}_{B_{1t}^{\sqrt{q}}(\theta)} \mathbbm{1}_{B_{2t}^{\sqrt{q}}(\theta) \cap \{\theta\}} \mid \mathcal{H}_{1t}\right]\geq q\right\} && \mbox{by definition of $F_t$}\\
& = \left\{\omega \mid \mathbbm{1}_{B_{1t}^{\sqrt{q}}(\theta)} \mathbb{E}\left[ \mathbbm{1}_{B_{2t}^{\sqrt{q}}(\theta) \cap \{\theta\}} \mid \mathcal{H}_{1t}\right] \geq q \right\} && \mbox{since } B_{1t}^{\sqrt{q}}(\theta) \in \mathcal{H}_{1t} \\
& = B_{1t}^{\sqrt{q}}(\theta) \cap B_{1t}^q \left(B_{2t}^{\sqrt{q}}(\theta) \cap \{\theta\}\right)
\end{align*} 
By definition we have that $F_t \subseteq B_{1t}^{\sqrt{q}}(\theta)$. As above, individual learning implies existence of $T$ sufficiently large that $P_\theta\left(B_{2t}^{\sqrt{q}}(\theta)\right)>\sqrt{q}$ for all $t>T$. Since signals are conditionally independent, agent 1's history is uninformative about agent 2's history, implying that
\begin{align}
P_\theta\left(B_{2t}^{\sqrt{q}}(\theta) \mid \mathcal{H}_{1t}\right) \geq \sqrt{q} \label{eq:qBound}
\end{align}
holds uniformly across agent 1 histories (for all $t>T$). So on $F_t$ (for $t>T$) we have
\[P(B_{2t}^{\sqrt{q}}(\theta) \cap \{\theta\} \mid \mathcal{H}_{1t}) = \underbrace{P_\theta(B_{2t}^{\sqrt{q}}(\theta) \mid \mathcal{H}_{1t} )}_{>\sqrt{q} \text{ by } (\ref{eq:qBound})} \underbrace{P(\theta \mid \mathcal{H}_{1t}) }_{>\sqrt{q} \text{ since } F_t \subseteq B_{1t}^{\sqrt{q}}(\theta)} >q.\]
Apply Lemma \ref{lemm:MS} and we are done.
\end{proof}

\begin{remark} This proof extends for arbitrary finite numbers of agents, setting $F_t = \{\theta\} \cap B_t^{\sqrt[n]{q}}(\theta)$.
\end{remark}

\bigskip

Although common learning is implied by individual learning when agents have either perfect information or no information about the other agent's history, intermediate cases of correlation can break this result. \\

\begin{example} (Twist on Rubinstein (1989)'s email game.) \label{ex:EmailGameTwist} The unknown parameter is 
 $\theta \in \{\theta', \theta"\}$, where $0\leq \theta' < \theta'' \leq 1$. Suppose that every period a signal profile is independently drawn according to:
\[ \begin{array}{ccc}
 \mbox{Probability} & \mbox{Agent-1 Signal} & \mbox{Agent-2 Signal} \\
 \theta & 0 & 0  \\[-1mm]
 \eps(1-\theta) & 1 & 0 \\[-1mm]
 (1-\eps)\eps(1-\theta) & 1 & 1  \\[-1mm]
 (1-\eps)^2 \eps(1-\theta) & 2 & 1 \\[-1mm]
 (1-\eps)^3 \eps(1-\theta) & 2 & 2 \\[-1mm]
 (1-\eps)^4 \eps (1-\theta) & 3 & 2 \\[-1mm]
 (1-\eps)^5 \eps (1-\theta) & 3  & 3 \\[-1mm]
 \vdots & \vdots & \vdots
 \end{array}\]

\noindent This signal structure generalizes the information structure  in the email game from Section \ref{sec:emailgame}, where $\theta=1$ corresponds to state $a$ in the email game and $\theta=0$ corresponds to state $b$.

Agents observe repeated independent realizations of the signal. Will they commonly learn the game parameter? When $\theta$ is restricted to values 0 and 1 (as per \citet{Rubinstein}'s email game), the answer is yes.

\begin{exercise}[G] Prove that common learning occurs if $\theta \in \{\theta',\theta''\} \equiv \{0,1\}$.
\end{exercise}

But common learning fails whenever $0 <\theta'<\theta" < 1$ as agents cannot commonly learn $\theta''$, the parameter placing more weight on the lower signal realizations. Intuitively, when 1 sees the signal $k$, then he believes with some probability (that can be uniformly lower bounded across histories) that 2 has also observed at least $k$. And if 2 observes $k$, then he believes with some probability (that again can be uniformly lower bounded) that 1 observed $k+1$. Since the number of signal realizations is infinite, there is unbounded contagion upwards: The agent always believes with some probability that the other agent believes with some probability that he has observed\dots such a large signal that he believes that the state is (very likely to be) $\theta'$. And thus we cannot establish common $q$-belief of $\theta''$ for large $q$.

\end{example}

The main result in \citet{CEMS2008} establishes that infinite signal realizations are critical to the previous counterexample. When the number of signal realizations is finite, then individual learning always implies common learning.

\begin{assumption}[Finite Signal Sets] $\vert \mathcal{X}^1 \vert, \vert \mathcal{X}^2 \vert <\infty$ \label{assp:FiniteSignal}
\end{assumption}

\begin{proposition} If Assumption \ref{assp:FiniteSignal} is satisfied, then individual learning implies common learning.
\end{proposition}

A brief idea of the proof follows. Define $\pi^\theta(ij)$ to be the probability of realization $(x^1_{t},x^2_{t})=(i,j)$ when the parameter is $\theta$, and define 
\[\phi^\theta(i) = \sum_{j \in \mathcal{X}^2} \pi^\theta (ij)\]
to be the marginal probability of signal $i$, with $\phi^\theta \equiv (\phi^\theta(i))_{i \in \mathcal{X}^1}$. Likewise define
\[\psi^\theta(j) = \sum_{i\in \mathcal{X}^1} \pi^\theta(ij)\]
to be the marginal probability of signal $j$, with $\psi^\theta \equiv (\psi^\theta(j))_{j \in \mathcal{X}^2}$. Then (by the results in Section \ref{sec:Doob}), individual learning follows whenever $\phi^\theta \neq \phi^{\theta'}$ and $\psi^\theta \neq \psi^{\theta'}$ for every $\theta \neq \theta'$. 

Define $\hat{\phi}_t$ to be the empirical frequency of agent $1$ signals
and $\hat{\psi}_t$ to be the empirical frequency of agent $2$ signals. Under the assumption of individual learning, empirical frequencies must converge to the theoretical frequencies, i.e., for each parameter $\theta$, $\hat{\phi}_t \rightarrow \phi^\theta$ and $\hat{\psi}_t \rightarrow \phi^\theta$ $P_\theta$-almost surely. Thus each agent eventually assigns a high probability to true $\theta$.

The crucial next step is establishing that when agent 1 assigns a high probability to $\theta$, he believes that agent 2 does as well (and vice versa). To see why this might be the case, let $M_1^\theta$ be the $\vert \mathcal{X}^1 \vert \times \vert \mathcal{X}^2 \vert$ matrix whose $(i,j)$-th entry is $\frac{\pi^\theta(ij)}{\phi^\theta(i)}$, i.e. the conditional probability (under $\theta$) that agent 2 observes $j$ given that agent 1 observed $i$, and define $M_2^\theta$ analogously. Then $\hat{\phi}_t M_1^\theta$ is agent 1's expectation of agent 2's realized frequencies (conditional on $\theta$), and $\hat{\phi}_t M_1^\theta M_2^\theta $ is agent 1's expectation of agent 2's expectation of agent 1's realized frequencies (again conditional on $\theta$). Observe (by algebra) that
\begin{align*}
\phi^\theta M_1^\theta &= \psi^\theta \\
\psi^\theta M_2^\theta &= \phi^\theta
\end{align*}
so $\phi^\theta M_1^\theta M_2^\theta = \phi^\theta$. Indeed the matrix $M_{12}^\theta \equiv M_1^\theta M_2^\theta$ is a Markov transition matrix on $\mathcal{X}^1$ with stationary distribution $\phi^\theta$, and  it is moreover a contraction mapping on $\Delta(\mathcal{X}^1)$. These properties together imply that the higher order beliefs cannot run away from the agent's first-order belief as they did in Example \ref{ex:EmailGameTwist}.

\section{Additional Exercises}

\begin{exercise}[G$^*$]
Let $\theta \sim \mathcal{N}(0,1)$ be an unknown parameter. Each agent $i=1,2$ observes $n$ signals $X_1^i, \dots, X_n^i$ where each 
\[X_m^i = \theta + \eps^i_m\]
with $\eps^i_m \sim \mathcal{N}(0,1)$ independent of $\theta$, independent across agents, and independent across signals. Suppose that the true value of $\theta$ is strictly positive, and let $E_p$ be the event that the two agents have common $p$-belief that $\theta$ is positive, where $p>1/2$. What is the probability of $E_p$ under the actual data-generating process?
\end{exercise}

\chapter{Model Uncertainty and Misspecification}

We have so far assumed that agents' model of the world is \emph{correctly specified}: Their prior belief over $\Theta$ assigns positive probability to the true parameter $\theta$ and they update to information correctly, i.e. with knowledge of the true signal generating distribution $(P_\theta)_{\theta \in \Theta}$. Some reasons to question this model of learning include:

\begin{itemize}
\item We see substantial and persistent disagreement between individuals, but Sections \ref{sec:Doob} and \ref{sec:Merging} imply that agents will eventually hold similar beliefs. 
\item It is unclear how agents came to know $(P_\theta)_{\theta \in \Theta}$.
\item The assumption that agents perceive only one signal-generating distribution $(P_\theta)_{\theta \in \Theta}$ as possible means that agents  never abandon their model, even as evidence accumulates against it. As we discuss  in Section \ref{sec:BinaryACY}, this dogmatism has some strange implications.
\end{itemize}

This section relaxes the standard learning model by allowing for \emph{model uncertainty} (Section \ref{sec:ModelUncertainty}) and \emph{model misspecification} (Section \ref{sec:Misspecification}). In the former class of models, agents hold non-degenerate beliefs over the signal generating distribution. In the second, agents assign probability zero to the true parameter.

\section{Model Uncertainty} \label{sec:ModelUncertainty}

\subsection{Motivation} \label{sec:BinaryACY}

Recall the binary setting from Section \ref{sec:LearningExample}: There is an unknown parameter $\theta \in \{A,B\}$, and each period $t\in \mathbb{Z}_+$ a signal is generated iid according to
\[\begin{array}{ccc}
& a & b\\
A & q& 1-q \\
B & 1-q & q
\end{array}\]
where $q>1/2$. Agents may hold different (non-degenerate) prior beliefs $\pi_i \in \Delta(\Theta)$ about the parameter, but the value of $q$ is common knowledge.

In Section \ref{sec:LearningExample}, we observed that these agents almost surely learn the true parameter as the sample size grows large, and moreover their disagreement about the parameter vanishes. This is because (1) agents assign probability 1 to the event in which the limiting fraction of $a$-realizations is either $q$ or $1-q$, and (2)  the parameter is identified, so for either of these limiting frequencies agents (eventually) assign probability 1 to the correct parameter value. 

What happens along sequences in which the limiting frequency is neither $(q,1-q)$ nor $(1-q,q)$? Although agents assign probability zero to this event, sampling variation can explain any empirical frequency of $a$ and $b$ realizations (however surprising) in finite sequences. Thus Bayes' rule yields well-defined posterior beliefs.

For example, suppose $q \in (1/2, 1)$ and let $\bold{x}$ be the (infinite) sequence of $a$-realizations.  For any $t$, the unconditional probability of the event that all $t$ realizations are $a$ is 
\[\pi^i_A \cdot q^t + (1-\pi^i_A)\cdot (1-q)^t\]
where $\pi^i_A$ denotes the prior probability of $A$. This expression converges to zero as $t$ grows large but is strictly positive for every $t$. The agent's limiting belief along $\bold{x}$ can thus be computed to be 
\[\lim_{t \rightarrow \infty} P^i(\theta = A \mid \bold{x}_t) = \lim_{t \rightarrow \infty} \frac{1}{1 + \frac{1-\pi^i_A}{\pi^i_A} \left(\frac{1-q}{q}\right)^{t}} = 1\]

So the agent is increasingly convinced that the state is $A$, even as the observed sequence grows increasingly unlikely under the agent's model.  Even more striking, as signals accumulate in the frequency $(1,0)$,  the agent becomes increasingly confident that future signals will appear in the frequency $(q,1-q)$! These conclusions are a consequence of the agent's dogmatic view of the signal generating distribution---he is unwilling to abandon this model even as mounting evidence points to its error.

\subsection{Expanded Framework}

We can introduce \emph{model uncertainty} into this learning model by expanding the state space to $\Omega = \Theta \times \Gamma \times \mathcal{X}^\infty$ where the new parameter $\gamma$ indexes the signal-generating distribution, and the parameters $\theta$ and $\gamma$ jointly determine a family $(P_{\theta,\gamma})_{\theta \in \Theta, \gamma \in \Gamma}$ of conditional distributions over signals. The key distinction between $\theta$ and $\gamma$ is that only $\theta$ is payoff-relevant. We'll use $P^i$ to denote agent $i$'s subjective prior belief on $\Omega$, which is common knowledge to all agents.

If people do not in fact have dogmatic beliefs about the signal-generating distribution, a natural question is whether modeling agents in this way is still a good abstraction, in the sense that the qualitative insights of this model are robust to introduction of a small amount of model uncertainty. \citet{ACY2015} demonstrate one important sense in which this is not so.

\subsection{Failure of Asymptotic Agreement}

For any infinite sequence $\bold{x} \in \mathcal{X}^\infty$, write
\[\phi_{\theta,t}^i \equiv P^i(\theta \mid x_1, \dots x_t )\]
for the posterior probability that agent $i$ assigns to $\theta$ following the first $t$ realizations of the sequence $\bold{x}$. Further define
\begin{equation} \label{eq:AsymptoticBelief}
\phi_{\theta,\infty}^i(\bold{x}) = \lim_{t \rightarrow \infty} \phi^i_{\theta,t}(\bold{x})
\end{equation}
to be the asymptotic posterior probability that agent $i$ assigns to $\theta$ along sequence $\bold{x}$. 

\begin{definition} Say that \emph{asymptotic agreement} occurs if for each agent $i$,
\[P^i(\phi^1_{\theta,\infty} = \phi^2_{\theta,\infty}) =1 \quad \forall \theta\in \Theta\]
\end{definition}
\noindent That is, both agents believe their asymptotic beliefs will be identical.

When agents hold a dogmatic belief about the signal-generating distribution, asymptotic agreement occurs whenever the parameter is identified (Proposition \ref{prop:PosteriorConsistency}). But \citet{ACY2015} show that asymptotic agreement can fail when an arbitrarily small amount of model uncertainty is introduced. The basic idea behind this fragility can be seen through this following example from their paper. 

Let $\Theta = \{A,B\}$, with each agent $i$'s prior about the parameter denoted by $\pi^i \equiv (\pi^i_A,\pi^i_B)$. Agent $i$ believes that signals are generated iid from the set $\{a,b\}$ with state-dependent distribution
\[\begin{array}{ccc}
& a & b\\
A & \gamma & 1-\gamma \\
B & 1-\gamma & \gamma
\end{array}\]
where $\gamma$ is unknown and distributed  according to $G^i$ with density
\[g^i(\gamma) = \left\{ \begin{array}{cl}
\eps + \frac{1-\eps}{\lambda} & \mbox{if } \gamma \in (\gamma^i - \lambda/2, \gamma^i + \lambda/2) \\
\eps & \mbox{otherwise}
\end{array} \right.\]
for some $\gamma^i >1/2$. Assume that  $\gamma^1$ and $\gamma^2$ are different from one another. This density is depicted in Figure \ref{fig:DensityACY}.

 \begin{figure}[h]
				\centering
				\includegraphics[scale=.7]{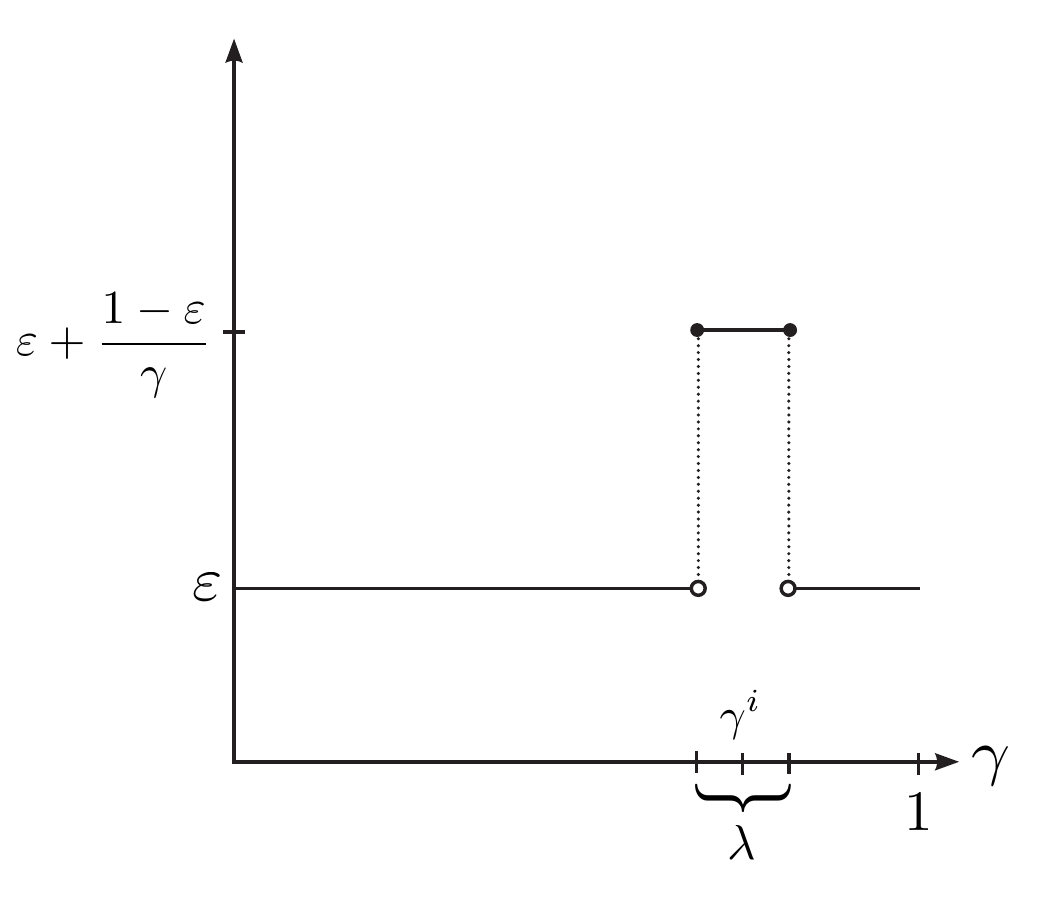}
				\caption{Depiction of $g^i$.}	 \label{fig:DensityACY}			
\end{figure}

The limit as $\eps \rightarrow 0$ and $\lambda \rightarrow 0$ returns the model in which each agent $i$ dogmatically believes the signal structure to be given by
\[\begin{array}{ccc}
& a & b\\
A & \gamma^i & 1-\gamma^i \\
B & 1-\gamma^i & \gamma^i
\end{array}\]
At this limit, asymptotic agreement holds.

Now suppose $\eps$ and $\lambda$ are strictly positive and $\lambda$ is small (specifically, let $\lambda < \vert \gamma^1 - \gamma^2 \vert$ and suppose $\gamma^i - \frac{\lambda}{2} > \frac12$ for each agent $i$). As in Section \ref{sec:LearningExample}, define
\[n_t(\bold{x}) \equiv \#\{ 1 \leq t' \leq t : \bold{x}_{t'} =a\} \quad \forall \bold{x} \in \mathcal{X}^\infty \]
to be the count of $a$-realizations among the first $t$ realizations of $\bold{x}$, and let 
\[\rho(\bold{x})= \lim_{t \rightarrow \infty} n_t(\bold{x})/t \quad \forall \bold{x} \in \mathcal{X}^\infty\]
be the asymptotic frequency of $a$-realizations along $\bold{x}$.

The following lemma provides a simple expression for the agent's asymptotic belief (\ref{eq:AsymptoticBelief}) on the set of sequences $\widetilde{\mathcal{X}}^\infty \subseteq \mathcal{X}^\infty$ where the limiting frequency $\rho(\bold{x})$ exists.

\begin{lemma}[\citet{ACY2015}] \label{lemm:Asymptotic} For every sequence $\bold{x} \in \widetilde{\mathcal{X}}^\infty$, 
\[
\phi^i_{A, \infty}(\bold{x})  = \left(1 + \frac{1-\pi^i_A}{\pi^i_A} \cdot \frac{f^i_B(\rho(\bold{x}),1-\rho(\bold{x}))}{f^i_A(\rho(\bold{x}),1-\rho(\bold{x}))}\right)^{-1}
\]
where $\frac{f^i_B(\rho(\bold{x}),1-\rho(\bold{x}))}{f^i_A(\rho(\bold{x}),1-\rho(\bold{x}))}$ is the asymptotic likelihood ratio under agent $i$'s subjective model.
\end{lemma}

In the running example of this section, the asymptotic likelihood ratio can be simplified to
\begin{align*}
\frac{f_B^i(\rho, 1-\rho)}{f_A^i(\rho,1-\rho)} = \frac{g^i(1-\rho)}{g^i(\rho)}
\end{align*}
This ratio takes on either of three possible values. For any $\rho \in (\gamma^i - \lambda/2, \gamma^i + \lambda/2)$, 
\[\frac{g^i(1-\rho)}{g^i(\rho)} = \frac{\eps \lambda}{1-\eps(1-\lambda)} \]
which converges to zero as $\eps$ and $\lambda$ grow small (implying $\phi^i_{A,\infty} \rightarrow 1$). By a mirror argument, if the limiting frequency of $a$-realizations is some $\rho \in (1-\gamma^i - \lambda/2, 1-\gamma^i + \lambda/2)$, then 
\[\frac{g^i(1-\rho)}{g^i(\rho)} = \frac{1-\eps(1-\lambda)}{\eps \lambda} \]
which converges to $\infty$ as $\eps$ and $\lambda$ grow small (implying $\phi_{A,\infty} \rightarrow 0$). 
For all other limiting frequencies, the asymptotic likelihood ratio is simply $\frac{g^i(1-\rho)}{g^i(\rho)}
=1$. These unlikely signal sequences are considered possible but uninformative about the parameter.

 Applying Lemma \ref{lemm:Asymptotic}, Figure \ref{fig:AsymptoticPosterior} depicts agent $i$'s asymptotic posterior as a function of the limiting signal frequency.

\begin{figure}[H]
				\centering
				\includegraphics[scale=0.75]{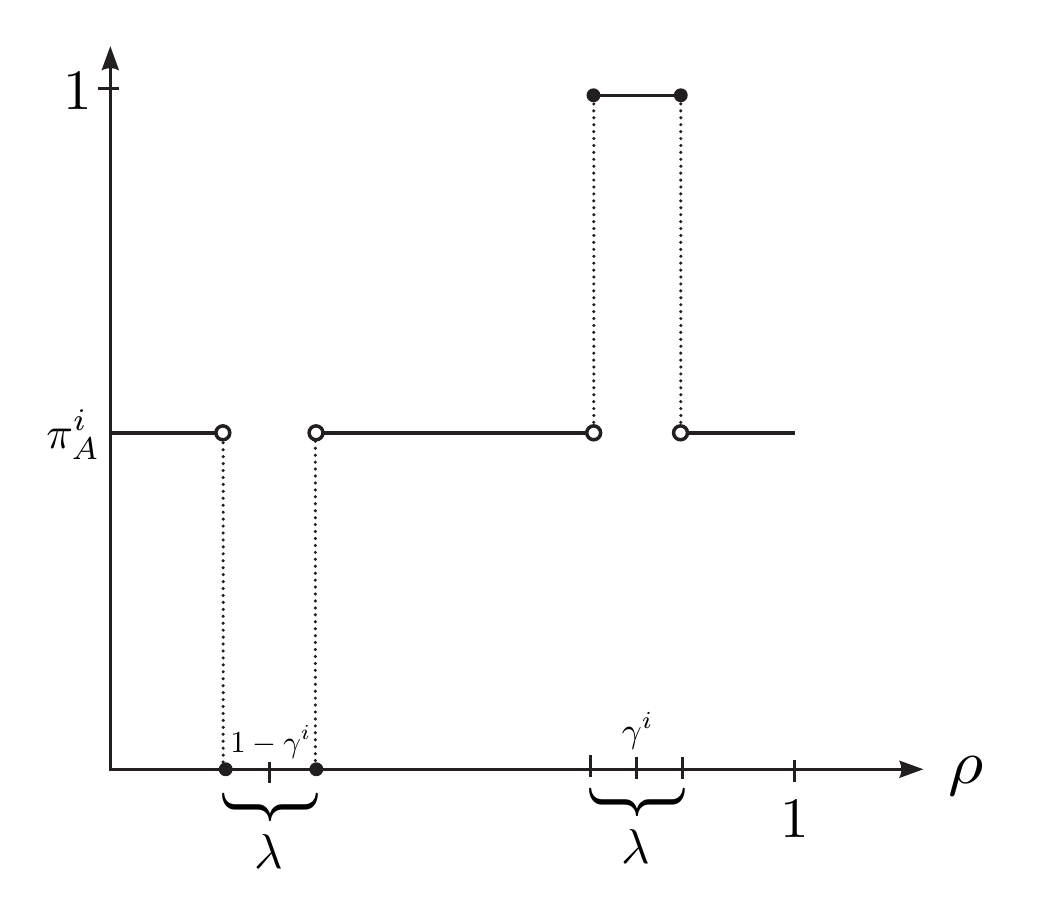}
				\caption{Agent $i$'s asymptotic posterior in the limit as $\eps \rightarrow 0$.}	 \label{fig:AsymptoticPosterior}			
\end{figure}

 In the limit as $\eps \rightarrow 0$ and $\lambda \rightarrow 0$, each agent $i$ is increasingly sure that the limiting frequency $\rho$ will either be close to $\gamma^i$ or $1-\gamma^i$, so he believes that he will (approximately) learn the parameter.  But when a sequence of signals has a long-run frequency that leads agent 1 to learn $\theta=A$ or $\theta=B$, agent 1 knows that this sequence has led agent 2 to consider the signal uninformative, in which case agent 2's limiting belief is the same as his prior.  Likewise whenever agent 2 believes the signal sequence to be informative about $\theta$, he knows that agent 1 considers the signal sequence to be uninformative. So not only does asymptotic agreement fail, but we have the stronger conclusion that the limiting beliefs $\phi^1_\infty$ and $\phi^2_\infty$ are different on \emph{all} sample paths. Figure \ref{fig:AsymptoticDisagreement} depicts $\vert \phi^1_{A,\infty} - \phi^2_{A,\infty}\vert$ as a function of the limiting signal frequency.
 
 To summarize, asymptotic agreement holds in the limiting model $\eps=0,\lambda=0$ (with no model uncertainty), but fails when the model is perturbed to include an arbitrarily small amount of model uncertainty via $\eps>0,\lambda>0$. 

\begin{remark} As in Section \ref{sec:Learning}, there is no ground truth---whether asymptotic agreement does or doesn't hold is determined solely with respect to the agents' subjective beliefs.
\end{remark}

\begin{remark} In this example, the two agents' prior beliefs on $\Theta \times \Gamma$ are absolutely continuous with respect to one another.  So Proposition \ref{prop:BlackwellDubins} tells us that their beliefs about future signal realizations will eventually merge. But $(\theta,\gamma)$ is not identified: For example, $(A,1)$ and $(B,0)$ identically lead to a degenerate distribution on the infinite sequence of $a$-realizations. Thus asymptotic agreement about the expanded parameter $(\theta,\gamma)$ is not guaranteed from the results of Sections \ref{sec:Doob} and \ref{sec:Merging}.
\end{remark}

\begin{figure}[h]
				\centering
				\includegraphics[scale=.8]{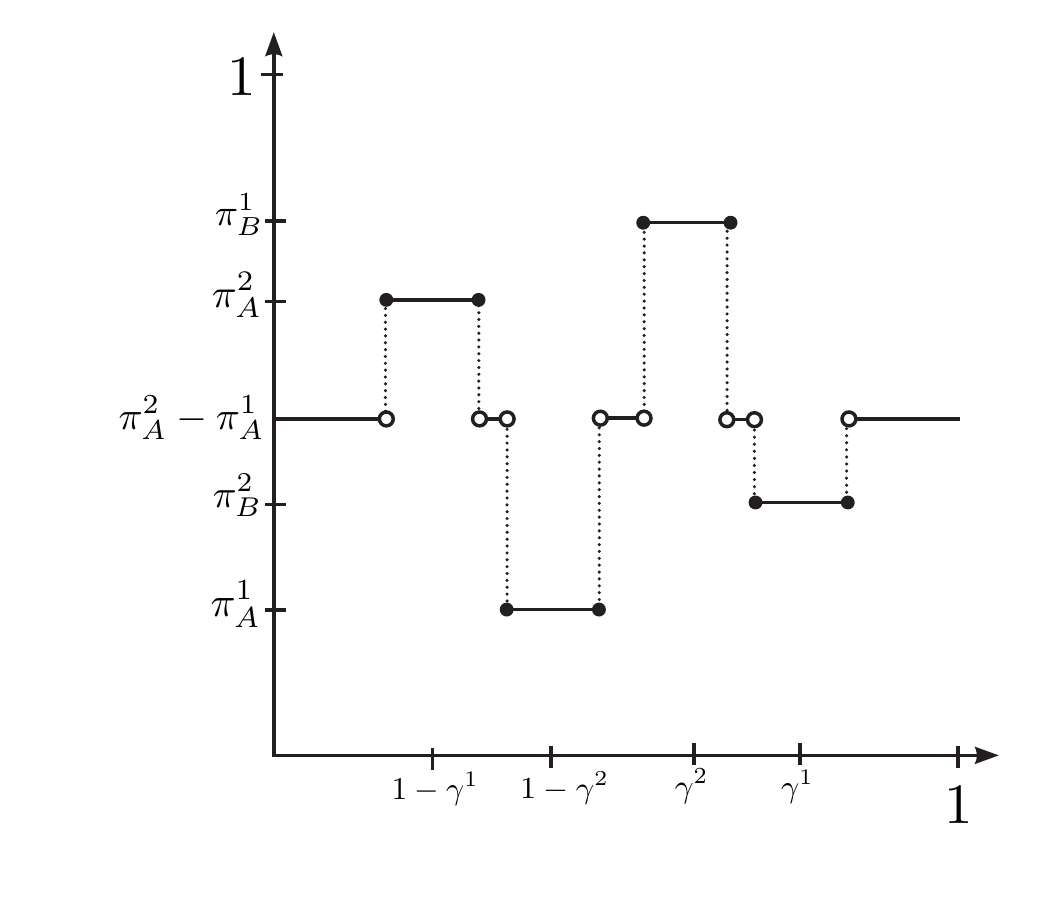}
				\caption{Asymptotic disagreement $\vert \phi^1_{A,\infty} - \phi^2_{A,\infty}\vert$ in the limit as $\eps \rightarrow 0$, for parameter values $ \pi_B^1 > \pi_A^2 > \pi_B^2 > \pi_A^1$.}	 \label{fig:AsymptoticDisagreement}					
\end{figure}

\section{Misspecified Learning} \label{sec:Misspecification}

Next suppose the agent is not simply uncertain about the signal-generating distribution, but in fact rules out the true distribution. 

\begin{example} \label{ex:Misspecified} Let $\Theta = \{A,B,C\}$ where the conditional distributions over signal realizations $\{a,b\}$ are given as follows:
\[\begin{array}{ccc}
& a & b\\
A & 4/5 & 1/5 \\
B & 1/2 & 1/2 \\
C & 2/3 & 1/3
\end{array}\]
The agent has a uniform prior on $\{A,B\}$, but the true parameter is $C$. Given repeated independent observations from the distribution $(2/3,1/3)$, will the agent's beliefs converge and if so to what limiting belief? \end{example}

\subsection{Role of KL Divergence} \label{sec:Berk}
Intuitively, we may expect that the agent's beliefs converge to certainty of the parameter whose distribution is ``closer" to the true distribution. The right notion of closeness here turns out to be KL Divergence (Section \ref{sec:KL}).

Here is a heuristic argument for how KL divergence emerges. Suppose the agent only considers parameter values  $\theta =A $ and $\theta=B$ to be possible, where the prior probability of $\theta=A$ is $\pi \in (0,1)$. We'll use $f_\theta(x)$ to denote the conditional probability of signal realization $x$ when the parameter is $\theta$. The agent observes a sequence of signals drawn iid according to $f_{\theta^*}$, where the ``true" parameter value $\theta^*$ may be different from both $A$ and $B$.

For any signal sequence $\bold{x}_t = (x_1, \dots, x_t)$, the conditional probability of $A$ can be rewritten
\begin{align*}
\mathbb{P}&(\theta = A  \mid \bold{x}_t)  \\
& =  \left(1 +  \frac{1-\pi}{\pi}  \left(\prod_{i=1}^t  \frac{f_B(x_i)}{f_A(x_i)}\right)\right)^{-1}  \\
& = \left(1 + \frac{1-\pi}{\pi} \left(\prod_{i=1}^t \frac{f_B(x_i)/f_{\theta^*}(x_i)}{f_A(x_i)/f_{\theta^*}(x_i)}\right)\right)^{-1}  \\
&= \left(1 + \frac{1-\pi}{\pi} \exp\left(  - \log \left(\prod_{i=1}^t  \frac{f_{\theta^*}(x_i)}{f_B(x_i)}\right) +  \log \left(\prod_{i=1}^t  \frac{f_{\theta^*}(x_i)}{f_A(x_i)}\right)\right)\right)^{-1}  \\
 &= \left(1 + \frac{1-\pi}{\pi} \exp\left(-n \cdot \left(\frac1t \sum_{i=1}^t \log \left(\frac{f_{\theta^*}(x_i)}{f_B(x_i)}\right) -  \frac1t \sum_{i=1}^t\log \left(\frac{f_{\theta^*}(x_i)}{f_A(x_i)}\right)\right)\right)\right)^{-1}  
 \end{align*}
 and for large $t$ this final display is approximately equal to
 \begin{equation}
  \left(1 + \frac{1-\pi}{\pi} \exp\left(-t \cdot \left(D(f_{\theta^*} \| f_B) - D(f_{\theta^*} \| f_A)\right)\right)\right)^{-1} \label{eq:KLBerk}
\end{equation}
If $\theta^* \in \{A,B\}$, then either $D(f_{\theta^*}\| f_A) =0 < D(f_{\theta^*}\| f_B)$ (in which case the expression in (\ref{eq:KLBerk}) converges to 1) or $D(f_{\theta^*}\| f_B)=0 < D(f_{\theta^*}\| f_A)$ (in which case the expression in (\ref{eq:KLBerk}) converges to 0). In either case beliefs converge to certainty of the true parameter, as previously implied by Proposition \ref{prop:PosteriorConsistency} (Section \ref{sec:Doob}).

Suppose now that $\theta^* \notin \{A,B\}$. Proposition \ref{prop:PosteriorConsistency} no longer applies:  \citet{Doob}'s consistency result is with respect to a $P$-measure 1 set of sequences, (where $P$ is the agent's prior on $\Theta \times \mathcal{X}^\infty$), but in this example $\theta^*$ falls in the $P$-measure zero set on which consistency is not guaranteed. Indeed, in Section \ref{sec:Doob} we made no reference to a ``true" distribution---consistency was demonstrated within the agent's subjective model.

But (\ref{eq:KLBerk}) is useful even when $\theta^*$ has zero probability under the agent's prior. Specifically, when $D(f_{\theta^*}\| f_A) < D(f_{\theta^*}\| f_B)$, then (\ref{eq:KLBerk}) converges to 1 as $t\rightarrow \infty$, yielding certainty of  $\theta=A$, and when $D(f_{\theta^*}\| f_A) > D(f_{\theta^*}\| f_B)$, then (\ref{eq:KLBerk}) converges to zero as $t \rightarrow \infty$. So the agent's beliefs concentrate on the parameter that induces a distribution over signals that is closest in Kullback-Liebler divergence to the true distribution.

\citet{Berk1966} establishes this result more generally. We'll use the notation of Section \ref{sec:LearningFramework}, introducing $\theta^*$ as new notation for the true parameter, and assuming that the observed signals are drawn iid according to the density $f_{\theta^*}$ (with $P_{\theta^*}$ denoting the induced measure on $\mathcal{X}^\infty$). To simplify exposition, assume that $\Theta$ is finite.

\begin{proposition}[\citet{Berk1966}] \label{prop:Berk} 
Let
\[A \equiv \argmin_{\theta \in Supp(P)} D(f_{\theta^*} \| f_\theta)  \]
be the set of parameters in the support of the agent's prior that minimize KL divergence to the true distribution.
Then
\[\lim_{t \rightarrow \infty} P(A \mid X_1, \dots, X_t) =1 \quad P_{\theta^*}\mbox{-a.s.}\]
\end{proposition}

\begin{example} Returning to Example \ref{ex:Misspecified}, since  
\begin{align*}
D(f_C \| f_A ) = (2/3) \cdot \log\left(\frac{2/3}{4/5}\right) +(1/3) \cdot \log\left(\frac{1/3}{1/5}\right) \approx 0.021 \\ D(f_C \| f_B) = (2/3) \cdot \log\left(\frac{2/3}{1/2}\right) +(1/3) \cdot \log\left(\frac{1/3}{1/2}\right) \approx 0.025 \end{align*}
Proposition \ref{prop:Berk} implies that the agent's beliefs converge to certainty of $\theta=A$.
\end{example}

\subsection{Berk Nash Equilibrium} 
Standard equilibrium concepts in game theory assume that players best-respond to correct and common beliefs. \citet{EspondaPouzo} proposes a new equilibrium concept (modifying Nash equilibrium) that allows players to be misspecified.  As this definition can be applied also within a single-agent setting, and as the notation is substantially lighter in this case, we start by defining Berk Nash equilibrium with one agent.

 \subsubsection{Single Agent Settings} 
 There is a finite set of payoff-relevant states $\Omega$, a finite set of signal realizations $\mathbb{S}$, and a finite set of actions $\mathbb{A}$. The agent holds a prior $p$ over $\Omega \times \mathbb{S}$. Additionally, there is a finite set of consequences $\mathbb{Y}$, which are determined by the agent's action and the state via a feedback function $f : \mathbb{A} \times \Omega \rightarrow \mathbb{Y}$. The agent's payoff function is $u : \mathbb{A} \times \mathbb{Y} \rightarrow \mathbb{R}$.

The timing is as follows. First the agent chooses a strategy $\sigma: \mathbb{S} \rightarrow \Delta(\mathbb{A})$ mapping the observed signal into a distribution over actions. Then, the state and signal $(\omega, s)$ are drawn according to $p$, and the action $\sigma(s)$ is implemented. Finally, the consequence $y$ is determined given the action and state $(a,\omega)$, and the agent obtains payoff $u(a,y)$.

There is an  \emph{objective} mapping $Q: \mathbb{S} \times \mathbb{A} \rightarrow \Delta(\mathbb{Y})$ from actions and signals into distributions over consequences, where
\[Q(y \mid s,a) = \sum_{\omega : f(\omega,a) =y} p(\omega \mid s) \quad \forall (y,s,a).\]
This is the conditional distribution over consequences that a Bayesian agent with knowledge of $f$, the action $a$, and the signal realization $s$ would expect.

The agent does not know $Q$ (or $f$). His \emph{subjective model} $\mathcal{Q}= \langle \Theta, (Q_\theta)_{\theta \in \Theta} \rangle$ is a parametrized family of mappings $Q_{\theta} : \mathbb{S} \times \mathbb{A} \rightarrow \Delta(\mathbb{Y})$. 

\begin{definition} The agent is \emph{correctly specified} if there exists $\theta \in \Theta$ such that $Q_{\theta}(\cdot \mid s,a) = Q(\cdot \mid s,a)$ for all $(s,a) \in \mathbb{S} \times \mathbb{A}$; otherwise the agent is \emph{misspecified}. 
\end{definition}

The following example is adapted from \citet{EspondaPouzo}:

\begin{example} A monopolist  chooses a price $a$, which together with a random shock $\omega \sim \mathcal{N}(0,1)$ determines demand
\[y = f(a,\omega) = \phi(a) + \omega.\]
The monopolist's payoff is $u(a,y) = a\cdot y$. Under the objective mapping $f$ , the conditional distribution $Q(\cdot \mid a)$ is normal with mean $\phi(a)$ and variance 1. The monopolist's subjective model is instead the family $Q_\theta(\cdot \mid a)$ of normal distributions indexed to $\theta =(\theta_0,\theta_1) \in \mathbb{R} \times \mathbb{R}$, where each $Q_\theta(\cdot \mid a)$ is normal with mean $\theta_0 + \theta_1 a$  and variance 1, corresponding to a perceived feedback function 
\[f_\theta(a,\omega) = \theta_0 + \theta_1 a.\]
If $\phi$ is not in fact affine in $a$, then the monopolist is misspecified. (This example did not include a signal.)
\end{example}

For any agent strategy $\sigma: \mathbb{S} \rightarrow \Delta(\mathbb{A})$, define 
\[q_\sigma(s,a) \equiv p_S(s) \sigma(a\mid s)\]
to be the distribution on $\mathbb{S} \times \mathbb{A}$ induced by the strategy $\sigma$ and the agent's prior $p$. Further define 
\[K(\sigma, \theta) = \sum_{(s,a) \in \mathbb{S} \times \mathbb{A}} \left(\mathbb{E}_{Q(Y \mid s, a)} \left[ \ln \frac{Q(Y \mid s, a)}{Q_{\theta}(Y \mid s, a)} \right]\right) q_\sigma(a,s)\]
to be the expected Kullback-Leibler divergence between $Q_\theta(\cdot \mid s,a)$ and the objective distribution $Q(\cdot \mid s,a)$, weighted by $q_\sigma \in \Delta(\mathbb{S} \times \mathbb{A})$. 

Given the agent's strategy $\sigma$, the set of closest parameters (in weighted KL divergence) is
\[\Theta^*(\sigma) = \arg \min_{\theta \in \Theta} K(\sigma,\theta) \]

\begin{definition} \label{def:BerkNash} A strategy profile $\sigma$ is a \emph{Berk-Nash equilibrium} if there exists a $\mu \in \Delta(\Theta)$ such that
\begin{itemize}
\item[(a)] $\mu \in \Delta(\Theta^*(\sigma))$; i.e., $\mu$ has support on the set of KL-minimizers.
\item[(b)] $\sigma$ is optimal given $\mu$; namely, $\sigma(a \mid s)>0$ implies that
\[a \in \arg \max_{a' \in \mathbb{A}} \mathbb{E}_{\overline{Q}_{\mu}( y \mid s, a')} [u(a', y)]\]
where $\overline{Q}_{\mu}(y \mid s, a) = \int_{\Theta} Q_{\theta}(y \mid s, a) \mu(\theta) d\theta$ is the conditional distribution over consequences that is induced by $\mu$. 
\end{itemize}
\end{definition}

\begin{example} \label{ex:BerkNash} A researcher's project is either good or bad, $\Omega = \{g,b\}$. The researcher observes a reaction to the project, which is either positive or negative, $\mathbb{S} = \{+,-\}$ where $(\omega,s)$ are jointly distributed according to:
\[\begin{array}{ccc}
& s=+ & s=-\\
\omega = g & 1/3 & 1/6 \\
\omega = b & 1/6 & 1/3 \\
\end{array}\]
The researcher observes the signal $s \in \mathbb{S}$ and decides whether to exert high or low effort towards developing the project, $A = \{H, L\}$. The unknown true quality of the project, and the researcher's effort, jointly determine a journal outcome in $\mathbb{Y} = \{A,R\}$ (accept or reject) according to the following function
\[f(a,\omega) =  \left\{ \begin{array}{cc}
A & (a,\omega)=(H,g) \\
R & otherwise
\end{array}\right.\]
That is, the project is accepted if it is good and also the researcher's effort is high, and it is rejected otherwise. The researcher's payoff is
\[u(a,y) = \left\{ \begin{array}{cc}
1 & (a,y)=(H, A) \\
-1 & (a,y)=(H, R)\\
2 & (a,y)=(L,A) \\
0 & (a,y)=(L,R)
\end{array} \right.\]
The true distribution $Q(y \mid a,s)$ is described by $Q(A \mid +,L)= Q(A \mid -, L) = 0$ (since the paper will not be accepted if effort is low) and
\begin{align*}
Q(A \mid +,H) & =  p(\{\omega : f(H,\omega) =A\} \mid +)  = p(g \mid +) = 2/3 \\
Q(A \mid -,H) & =  p(\{\omega : f(H,\omega) =A\}  \mid -) = p(g \mid -) = 1/3
\end{align*} 
since conditional on high effort, the probability of acceptance is equal to the probability that the paper is good. These conditional distributions are summarized as follows:
\[\begin{array}{ccc}
& A & R\\
(+,H) & 2/3 & 1/3 \\
(-,H)& 1/3 & 2/3 \\
(+,L)& 0 & 1 \\
(-,L) & 0 & 1
\end{array}\]
Suppose the researcher's subjective model allows only for the parameters $\theta_1$ and $\theta_2$ which are indexed to the following conditional distributions:
\[\begin{array}{ccc}
& A & R\\
(+,H) & 3/4 & 1/4 \\
(-,H)& 1/2 & 1/2 \\
(+,L)& 0 & 1 \\
(-,L) & 0 & 1
\end{array} \quad \quad \begin{array}{ccc}
& A & R\\
(+,H) & 2/3 & 1/3 \\
(-,H)& 1/3 & 2/3 \\
(+,L)& 1/10 & 9/10 \\
(-,L) & 1/10 & 9/10
\end{array}\]
The distribution on the left, $Q_{\theta_1}$, overestimates the value of hard work, and the distribution on the right, $Q_{\theta_2}$, is overly optimistic about the probability of acceptance given low effort. Is the strategy profile $\sigma(+)=H$, $\sigma(-)=L$ (in which the research exerts high effort after a positive signal and low effort after a low signal) a Berk Nash equilibrium?

The distribution $q_\sigma$ assigns probability $1/2$ to $(+,H)$ and to $(-,L)$. So
\begin{align*}
K(\sigma, \theta) & = \frac12 \left(\sum_{y \in \{A,R\}} Q(y \mid +,H) \cdot \ln\left( \frac{Q(y\mid +,H)}{Q_{\theta}(y\mid +,H) } \right)\right) \\
& \quad \quad \quad + \frac12 \left(\sum_{y \in \{A,R\}} Q(y \mid -,L) \cdot \ln\left( \frac{Q(y\mid -,L)}{Q_{\theta}(y\mid -,L) } \right)\right)
\end{align*}
and thus
\begin{align*}
K(\sigma, \theta_1) & = \frac12 \cdot \left( \frac{2}{3} \ln\left(\frac{2/3}{3/4}\right) +  \frac{1}{3} \ln\left(\frac{1/3}{1/4}\right) \right)  \approx 0.0038 \\
K(\sigma, \theta_2) & = \frac12 \cdot  \ln\left(\frac{1}{9/10}\right)   \approx 0.02 \end{align*}
Hence $\theta_1$ is the unique minimizer of KL divergence to the true distribution, i.e.,  $\Theta^*(\sigma) = \{\theta_1\}$. 

Only $\mu = \delta_{\theta_1}$ (a point mass at $\theta_1$) satisfies Part (a) of Definition \ref{def:BerkNash}, and the  distribution $\overline{Q}_\mu$ in Part (b) of Definition \ref{def:BerkNash} simplifies to $Q_{\theta_1}$. To determine whether $\sigma$ is a Berk Nash equilibrium, it remains to verify that $\sigma$ satisfies the optimality condition in Part (b) of Definition \ref{def:BerkNash}.

Suppose the signal realization is $s=+$. Then the action $H$ yields an expected payoff of 
\[\mathbb{E}_{Q_{\theta_1}(y \mid +,H)}[u(H,y)] = 1 \cdot \frac{3}{4} - 1 \cdot \frac{1}{4} = \frac{1}{2}\]
while the action $L$ yields an expected payoff of
\[\mathbb{E}_{Q_{\theta_1}(y \mid +,L)}[u(L,y)] =  0\]
so $a=H$ is indeed optimal.

Suppose the signal realization is $s=-$. Then the action $H$ yields an expected payoff of 
\[\mathbb{E}_{Q_{\theta_1}(y \mid -,H)}[u(H,y)] = 1 \cdot \frac{1}{2} - 1 \cdot \frac{1}{2} = 0\]
while the action $L$ yields an expected payoff of
\[\mathbb{E}_{Q_{\theta_1}(y \mid -,L)}[u(L,y)] =  0.\]
So $a=L$ is a best reply, and we conclude that $\sigma$ is a Berk Nash equilibrium.

In sum, we have shown that the strategy $\sigma$ is a best reply to a point mass on the unique parameter that minimizes KL divergence to the distribution over consequences induced by $\sigma$. In this sense the strategy $\sigma$ is internally consistent with respect to the agent's misspecified model. 
\end{example}

\begin{exercise}[G] Solve for all remaining pure-strategy Berk Nash equilibria in Example \ref{ex:BerkNash}, or prove that there are none other.
\end{exercise}

\subsubsection{Simultaneous-Move Games}

We turn now to the definition of Berk Nash equilibrium in simultaneous-move games. There is a set of players $I$, a set of payoff-relevant states $\Omega$, a set of signal profiles $\mathbb{S} = \times_i \mathbb{S}_i$, and a probability distribution $p$ over $\Omega \times \mathbb{S}$ whose marginals have full-support. There is a set of action profiles $\mathbb{A} = \times_i \mathbb{A}_i$, a set of \emph{consequence} profiles $\mathbb{Y} = \times_i \mathbb{Y}_i$, and a profile of \emph{feedback functions} $f = (f_i)_{i \in \mathcal{I}}$ where each $f_i : \mathbb{A} \times \Omega \rightarrow \mathbb{Y}_i$ maps outcomes in $\Omega \times \mathbb{A}$ into consequences for player $i$. Agents have payoff functions $u_i : \mathbb{A}_i \times \mathbb{Y}_i \rightarrow \mathbb{R}$.

The timing of the game is as follows: First, the state and signal $(\omega, s)$ are drawn according to $p$. Then each player $i$ privately observes his own signal $s_i$ and chooses an action $a_i$. The profile of consequences is determined via $f$ as a function of the action profile and the state, and payoffs are realized.

For any player $i$, action $a_i \in \mathbb{A}_i$, and consequence $y_i \in \mathbb{Y}_i$, let
\[\Lambda^i(a_i,y_i) = \{(\omega,a_{-i}): f_i(a_i,a_{-i},\omega) = y_i\}\]
be the state and opponent action profiles that induce consequence $y_i$ given player $i$'s choice of $a_i$. 
 The \emph{objective distribution} over player $i$'s consequences is $Q_\sigma^i : \mathbb{S}_i \times \mathbb{A}_i \rightarrow \Delta(\mathbb{Y}_i)$, where
\[Q_\sigma^i(y_i \mid s_i, a_i) = \sum_{(\omega,a_{-i}) \in \Lambda^i(a_i,y_i) } \sum_{s_{-i} \in S_{-i}} \left(\prod_{j \neq i} \sigma^j (a^j \mid s^j)\right) \cdot p_{\Omega \times S_{-i} \mid S_i} (\omega, s_{-i} \mid s_i) \]
for all $(s_i, a_i, y_i) \in \mathbb{S}_i \times \mathbb{A}_i \times \mathbb{Y}_i$. This is the conditional distribution over consequences that a Bayesian agent with knowledge of $f$, the strategy profile $\sigma$, and the signal realization $s_i$ would expect.

\bigskip

 The subjective model $\mathcal{Q}= \langle \Theta, (Q_\theta)_{\theta \in \Theta} \rangle$, with $\Theta = \prod_{i \in \mathcal{I}} \Theta^i$ and $Q_\theta = (Q^i_{\theta_i})_{i \in \mathcal{I}}$, describes the set of distributions over consequences that each player considers possible.
 Each player's parameter set $\Theta_i$ indexes distributions $Q^i_{\theta_i} : \mathbb{S}_i \times \mathbb{A}_i \rightarrow \Delta(\mathbb{Y}_i)$.

\begin{definition} A game is \emph{correctly specified given $\sigma$} if for all players $i$, there exists $\theta_i \in \Theta_i$ such that $Q_{\theta_i}^i(\cdot \mid s_i,a_i) = Q_\sigma^i(\cdot \mid s_i,a_i)$ for all $(s_i,a_i) \in \mathbb{S}_i \times \mathbb{A}_i$; otherwise the game is \emph{misspecified given $\sigma$}. A game is \emph{correctly specified} if it is correctly specified for all $\sigma$; otherwise it is \emph{misspecified}.
\end{definition}

For any strategy profile $\sigma$, define 
\[q_{\sigma_i}(s_i,a_i) \equiv \sigma_i(a_i \mid s_i) p_{S_i}(s_i)\]
For any strategy profile $\sigma$, define
\[K_i(\sigma, \theta_i) = \sum_{(s_i,a_i) \in \mathbb{S}_i \times \mathbb{A}_i} \left(\mathbb{E}_{Q^i_\sigma(\cdot \mid s_i, a_i)} \left[ \ln \frac{Q_\sigma^i(Y_i \mid s_i, a_i)}{Q_{\theta_i}^i(Y_i \mid s_i, a_i)} \right]\right) q_{\sigma_i}(s_i,a_i)\]
to be the expected Kullback-Leibler divergence between $Q_{\theta_i}(\cdot \mid s_i,a_i)$ and the objective distribution $Q_\sigma^i(\cdot \mid s_i,a_i)$, weighting $(s_i,a_i)$ pairs according to $q_{\sigma_i}(s_i,a_i)$. 

The set of closest parameters is
\[\Theta_i(\sigma) = \arg \min_{\theta_i \in \Theta_i} K_i(\sigma,\theta_i) \]

\begin{definition} A strategy profile $\sigma$ is a \emph{Berk-Nash equilibrium} if for all players $i$, there exists a $\mu_i \in \Delta(\Theta_i)$ such that
\begin{itemize}
\item[(a)] $\mu_i \in \Delta(\Theta_i(\sigma))$; i.e., $\mu$ has support on the set of KL-minimizers.

\item[(b)] $\sigma_i$ is optimal given $\mu_i$; namely, $\sigma_i(a_i \mid s_i)>0$ implies that
\[a_i \in \arg \max_{\overline{a}_i \in \mathbb{A}_i} \mathbb{E}_{\overline{Q}^i_{\mu_i}(\cdot \mid s_i, \overline{a}_i)} [u_i(\overline{a}_i, Y_i)]\]
where $\overline{Q}^i_{\mu_i}(\cdot \mid s_i, \overline{a}_i) = \int_{\Theta_i} Q^i_{\theta_i}(\cdot \mid s_i, a_i) \mu_i(\theta_i)d\theta_i$ is the distribution over consequences of player $i$, conditional on $(s_i,a_i) \in \mathbb{S}_i \times \mathbb{A}_i$, induced by $\mu_i$. 
\end{itemize}
\end{definition}

\begin{remark} This definition is equivalent to Nash equilibrium when (a) is replaced with the condition that players have correct beliefs; i.e., $\overline{Q}^i_{\mu_i} = Q_\sigma^i$.
\end{remark}

\begin{proposition}[\citet{EspondaPouzo}] A Berk-Nash equilibrium exists.
\end{proposition}

Building on Proposition \ref{prop:Berk}, several authors have examined convergence of  misspecified learning processes where---different from \citet{Berk1966}'s setting---signals are endogenous to actions chosen by agents  \citep{Nyarko,FudenbergRomanyukStrack,HeidhuesKoszegiStrack2021}. The stable outcomes under many of these processes turn out to correspond to Berk Nash equilibria or a refinement of this set. Some recent works on this topic include \citet{EspondaPouzo}, \citet{EspondaPouzoYamamoto}, \citet{BohrenHauser}, \citet{FudenbergLanzaniStrack}, \citet{EspondaPouzoYamamoto} and \citet{FrickIijimaIshii}.

\section{Additional Exercises}

\begin{exercise}[G] There are two states of the world, $\theta \in \{A,B\}$. A news source receives an infinite sequence of signals about this state of the world drawn iid according to the following signal structure
\[\begin{array}{ccc}
& a & b \\
\theta = A & 3/4 & 1/4 \\
\theta = B & 1/4 & 3/4
\end{array}\]
This news source is biased. When it observes the signal realization $a$, it  reports $a$, but conditional on observing the signal realization $b$, it reports this $b$ with probability $1-\lambda$ and otherwise falsely reports $a$ (where $\lambda$ is constant across time). You are aware that the news source is biased and dogmatically believe that $\lambda = 1/2$. 

Suppose the true state is $\theta=B$, and you observe the infinite sequence of news reports. Provide a condition (potentially empty) on the true value of $\lambda$ such that your asymptotic belief is that the state is $\theta=A$. Interpret this result.
\end{exercise}

\chapter{Information Design}

\section{Bayesian Persuasion}

\subsection{Example} \label{sec:BPexample}
A judge and a prosecutor are involved in a court case. The unknown payoff-relevant state is whether the defendant in this case (who will not take an action) is  \emph{innocent} ($I$) or \emph{guilty} ($G$). The judge and the prosecutor share  a common prior that the defendant is guilty with probability 0.3. 

The prosecutor cannot falsify or distort evidence, but can selectively choose what kind of information to present to the court (e.g., deciding who to subpoena or which forensic tests to conduct). Formally, the prosecutor chooses an information structure $\sigma: \{G,I\} \rightarrow \Delta(S)$ for some set of signal realizations $S$.
The judge observes the outcome of the signal $\sigma$, updates his beliefs, and chooses whether to \emph{acquit} or \emph{convict} the defendant.

The judge and prosecutor's payoffs are determined by the judge's action and by the unknown state. The judge receives a payoff of 1 from convicting a guilty defendant or from acquitting an innocent defendant, and otherwise receives a payoff of zero. The prosecutor receives a payoff of 1 if the judge convicts the defendant and a payoff of 0 if the judge acquits the defendant, independent of the defendant's guilt. What information structure should the prosecutor choose, and what is the best expected payoff he can achieve?

Let's start with some benchmarks. One possibility is to send a completely uninformative signal.  Since innocence is more likely than guilt under the judge's prior, the judge chooses to acquit given no information, yielding a payoff of zero for the prosecutor. Alternatively, the prosecutor can choose a perfectly informative signal that reveals the defendant's guilt. The judge convicts precisely when the defendant is guilty, yielding an expected payoff (under the prior) of 0.3 for the prosecutor. 

Can the prosecutor do better? The perfectly revealing signal splits defendants into two bins---one labeled ``convict" and one labelled ``acquit" (Figure \ref{fig:BPreveal}).

\begin{figure}[H]
\begin{center}
\includegraphics[scale=0.25]{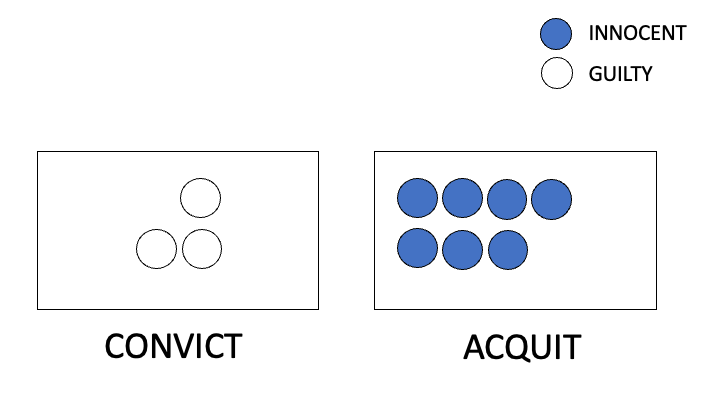}
\end{center}
\caption{Depiction of the perfectly revealing signal, where each circle represents $1/10$ of the population.} \label{fig:BPreveal}
\end{figure}

The judge's posterior for individuals labeled ``convict" is that they are guilty with probability 1, so he optimally convicts any individual with this label. Likewise his posterior for individuals labeled ``acquit" is that they are innocent with probability 1, so he acquits any individual with this label.

Now consider moving one unit of innocent individuals from the acquit bin to the convict bin (Figure \ref{fig:BPdeviate}).

\begin{figure}[H]
\begin{center}
\includegraphics[scale=0.25]{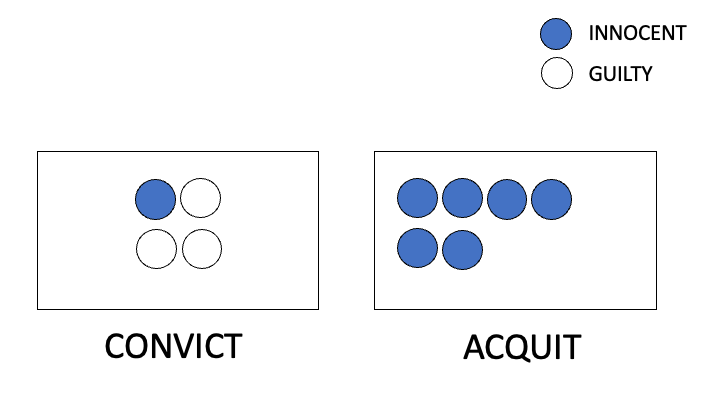}
\end{center}
\caption{Deviation from the perfectly revealing signal.} \label{fig:BPdeviate}
\end{figure}

 \begin{remark} Every ``bin representation" as shown in Figures \ref{fig:BPreveal} and \ref{fig:BPdeviate} corresponds to a unique signal. For each $\theta \in \Theta$ and $s \in \{\mbox{convict}, \mbox{acquit}\}$, let $P(\theta,s)$ be the mass of $\theta$-type units in bin $s$ (interpreting each circle as $1/10$ of the population). Then $P$ is a probability measure on $\Theta \times S$, and the corresponding signal $\sigma: \Theta \rightarrow \Delta(S)$ can be derived by Bayes' rule. As we see in the proof of Proposition \ref{prop:BP}, every signal also admits a bin representation.\footnote{In particular, every signal admits a ``bin representation" that consists of two bins---a convict bin, and an acquit bin---where the judge optimally convicts all individuals in the convict bin and acquits all individuals in the acquit bin.} 
\end{remark}

Following this modification on the perfectly revealing signal, the posterior probability of guilt in the acquit bin is unchanged. The posterior probability of guilt for individuals labeled ``convict" drops to $3/4$---but crucially, the judge's optimal action remains the same. Intuitively, by pooling innocent defendants with guilty defendants (but maintaining sufficiently guilty defendants that the judge still wants to convict), the prosecutor is able to induce the judge to wrongly convict a larger number of defendants.

Iterating this logic, we can continue moving units of innocent individuals into the convict bin, up until the judge is indifferent between convicting and acquitting (Figure \ref{fig:BPoptimal}).

\begin{figure}[h]
\begin{center}
\includegraphics[scale=0.25]{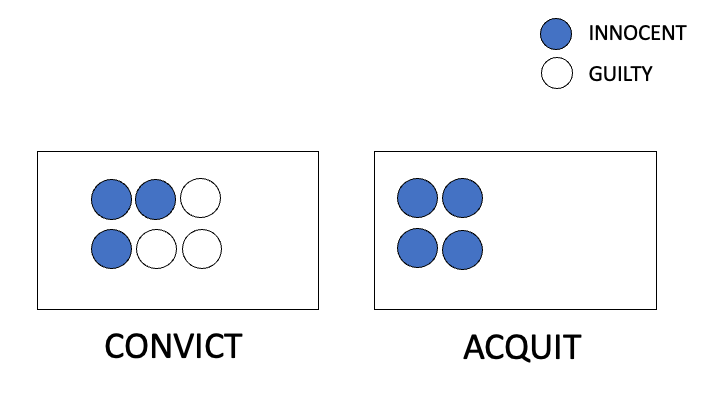}
\end{center}
\caption{Depiction of the prosecutor-optimal signal structure.} \label{fig:BPoptimal}
\end{figure}

These bins correspond to the following signal structure:
\begin{equation} \label{eq:OptimalSignal}
\begin{array}{ccc}
& \emph{convict} & \emph{acquit} \\
G & 1 & 0 \\
I & 3/7 & 4/7
\end{array}
\end{equation}

\noindent That this signal structure is optimal will follow from the results in the subsequent section. Strikingly, although the judge knows that only 30\% of defendants are guilty, he ends up convicting 60\% of them.

\subsection{Model} 

There are two agents, a Sender and a Receiver. The unknown parameter $\theta$ takes values in the finite set $\Theta$, and agents share a common prior $\mu_0 \in \Delta(\Theta)$. A signal is any mapping $\sigma: \Theta \rightarrow \Delta(S)$ from the set of states into distributions over a finite set of signal realizations $S$.

The Receiver chooses from a compact set of actions $A$. Both agents' payoffs depend on the Receiver's action and the unknown state. We'll denote the Receiver's utility function by  $u_R: A \times \Theta \rightarrow \mathbb{R}$ and the Sender's utility function by $u_S: A \times \Theta \rightarrow \mathbb{R}$, where both are assumed to be continuous.

The timeline is as follows: First, the Sender chooses a signal $\sigma$. The realization of this signal is then observed by the Receiver, who updates his beliefs and chooses an action $a \in A$. Finally payoffs are realized. The solution concept is Sender-Preferred subgame perfect equilibrium; that is, the Receiver chooses an action to maximize his expected payoffs, breaking ties between optimal actions by maximizing Sender's payoffs.\footnote{If there are multiple such actions, the Receiver chooses any action between them.}

\subsection{Solution and Geometric Representation}

Consider any Sender-Preferred subgame perfect equilibrium, and let $\hat{a}(\mu)$ denote the Receiver's action given belief $\mu \in \Delta(\Theta)$ in this equilibrium. That is, 
\begin{equation} \label{eq:ReceiverAction}
\hat{a}(\mu) \in \argmax_{a \in A(\mu)} \mathbb{E}_\mu \left[u_S(a, \theta)\right]
\end{equation}
where
\[A(\mu) = \argmax_{a \in A} \mathbb{E}_\mu \left[u_R(a, \theta)\right]\]
is the set of actions that maximize the Receiver's expected payoff given belief $\mu$. (If the RHS of (\ref{eq:ReceiverAction}) is non-empty, set $\hat{a}(\mu)$ to be any action in this set.) Let
\begin{equation} \label{hatv}
\hat{v}(\mu) := \mathbb{E}_\mu \left[u_S(\hat{a}(\mu), \theta)\right]
\end{equation}
be the Sender's expected payoff given belief $\mu$ and Receiver-action $\hat{a}(\mu)$.  A signal's \emph{value} is the Sender's (ex-ante) expected payoff given choice of that signal.

\begin{proposition}[\citet{KamenicaGentzkow}] The following are equivalent:
\begin{itemize}
\item[(i)] There exists a (finite-valued) signal with value $v^*$.
\item[(ii)] There exists a (finite-valued) signal taking realizations in $S\subseteq A$ with value $v^*$.
\item[(iii)] There exists a Bayes-plausible distribution over posterior beliefs, $\tau \in \Delta(\Delta(\Theta))$, such that $\mathbb{E}_\tau\left[\hat{v}(\mu)\right] = v^*$.
\end{itemize} \label{prop:BP}
\end{proposition}

\begin{proof}
The implication (ii) $\Rightarrow$ (i) is immediate.  The implication (ii) $\Rightarrow$ (iii) follows from Fact \ref{fact:Martingale} (every signal induces a Bayes-plausible distribution over posterior beliefs).

To show (i) $\Rightarrow$ (ii), observe that for any signal $\sigma: \Theta \rightarrow \Delta(S)$ with value $v^*$, we can define a new signal $\widetilde{\sigma}: \Theta \rightarrow \Delta(A)$ that maps types into the recommended action under $\sigma$. That is,
\[\widetilde{\sigma}(a \mid \theta) = \sum_{s :  \hat{a}(\mu_s) = a} \sigma(s \mid \theta)\]
for every $a \in A$ and $\theta \in \Theta$, where $\mu_s$ denotes the Receiver's posterior given signal realization $s$ under $\sigma$. (The number of distinct action recommendations cannot exceed the size of $S$ and so is finite.) Clearly the optimal action given recommendation of $a$ remains the action $a$, so the distribution of optimal actions induced by $\widetilde{\sigma}$ and $\sigma$ are the same. 

The direction (iii) $\Rightarrow$ (i) is nearly immediate from Proposition \ref{prop:BayesPlausible} (every  Bayes-plausible distribution over posterior beliefs can be induced by a signal), but we need to show that it is possible to construct a \emph{finite-valued} signal for arbitrary $\tau$ (even ones with infinite support).\footnote{The construction in Section \ref{sec:BayesPlausibility} chooses $S$ to be the set of all beliefs in the support of $\tau$, which need not be finite.}

We'll use the following result from convex analysis.

\begin{proposition}[Caratheodory's Theorem] Let $X \subseteq \mathbb{R}^n$ be a nonempty subset of finite-dimensional Euclidean space. Let $conv(X)$ denote the convex hull of $X$. Then every vector in $conv(X)$ can be represented as a convex combination of at most $n+1$ vectors from $X$.
\end{proposition}

Fix any $v^*$ and Bayes-plausible $\tau$ such that $\mathbb{E}_\tau [\hat{\nu}(\mu)] = v^*$. Define
\[C = \{(\mu, \hat{v}(\mu)) \mid \mu \in \Delta(\Theta)\}\]
to be the set of all beliefs and valuations of those beliefs, noting that $C \subseteq \mathbb{R}^n$ where $n \equiv \vert \Theta \vert$.\footnote{The simplex $\Delta(\Theta)$ is a subset of $\mathbb{R}^{n-1}$ and the valuation belongs to $\mathbb{R}$, hence $C \subseteq \mathbb{R}^n$.} Moreover, by assumption that $v^* = \mathbb{E}_\tau [\hat{\nu}(\mu)]$ for some Bayes-plausible distribution $\tau$ over posterior beliefs, the vector $(\mu_0, v^*)$ belongs to the convex hull of $C$.

Then by Caratheodory's Theorem, there exists a sequence of beliefs $(\mu_i)_{i=1}^{n+1}$ and a sequence of nonnegative weights $(\alpha_i)_{i=1}^{n+1}$ summing to 1, such that
\[(\mu_0,v^*)  = \sum_{i=1}^{n+1} \alpha_i \cdot (\mu_i,\hat{v}(\mu_i))\] 
Let $\tau^*$ be the distribution over posterior beliefs that assigns probability $\alpha_i$ to each belief $\mu_i$, $1\leq i \leq n+1$. Then
\[\mathbb{E}_{\tau^*}[\hat{v}(\mu)] = \sum_{i=1}^{n+1} \alpha_i \cdot \hat{v}(\mu_i) = v^*\]
as desired. Follow the construction in Section \ref{sec:BayesPlausibility} (setting the set of signal realizations $S$ to be the posterior beliefs in the support of $\tau^*$) to complete the proof.
\end{proof}

\bigskip

Proposition \ref{prop:BP} tells us that we can determine when the Sender benefits from persuasion by studying how $\mathbb{E}_\tau\left[\hat{\nu}(\mu)\right]$ varies over the set of Bayes-plausible distributions.

\begin{corollary} The Sender benefits from persuasion if and only if there exists a Bayes-plausible distribution $\tau$ such that $\mathbb{E}_\tau \left[\hat{\nu} (\mu)\right] > \hat{\nu} (\mu_0)$.
\end{corollary}

\begin{corollary}
The value of an optimal signal is
\begin{align*}
\max_{\tau \in \Delta(\Theta)} \mathbb{E}_\tau \left[\hat{\nu}(\mu)\right] \quad  \mbox{s.t. } \int \mu d\tau(\mu) = \mu_0
 \end{align*}
\end{corollary}

The value of information for the Sender at any prior $\mu$ can be represented geometrically using the upper concave envelope of $\hat{\nu}$.

\begin{definition} \label{def:ConcaveClosure} Define
\[V(\mu) \equiv \sup \{z \mid (\mu,z) \in Conv(\hat{\nu})\} \quad \forall \mu \in \Delta(\Theta)\]
where $Conv(\hat{\nu})$ denotes the convex hull of the graph $\hat{\nu}$. That is, $V$ is the smallest concave function that is everywhere weakly greater than $\hat{\nu}$. 
\end{definition}

\begin{figure}[H]
				\centering
				\includegraphics[scale=1]{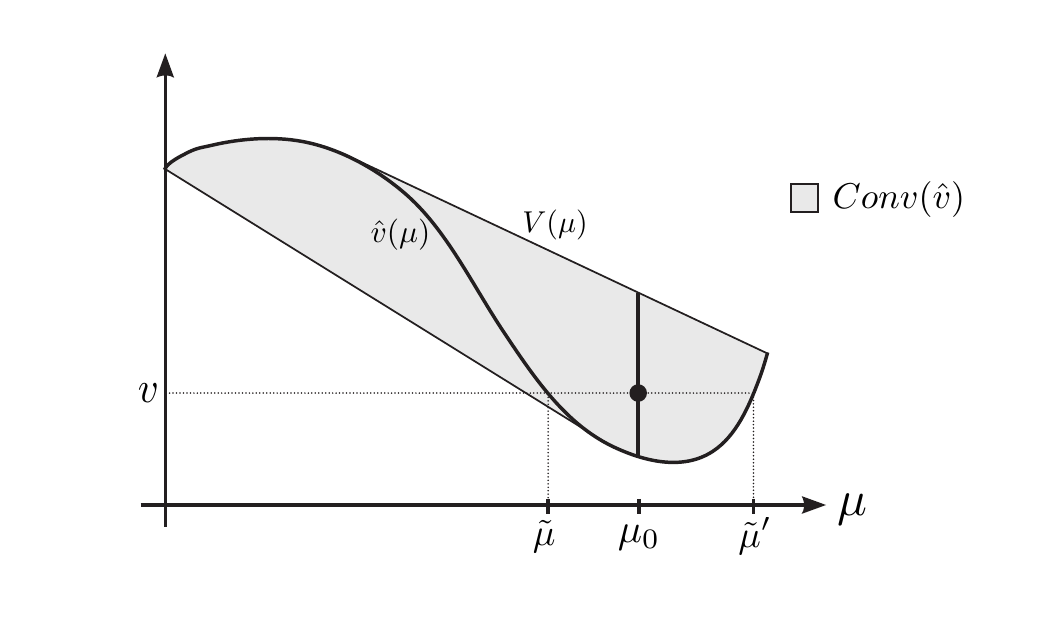}
				\caption{Illustration of Definition \ref{def:ConcaveClosure}.}
			\end{figure}

By Proposition \ref{prop:BP}, the set $\{ z \mid (\mu_0, z) \in Conv(\hat{\nu})\}$ is precisely those expected payoffs that the Sender can achieve when the prior $\mu_0$. For example, in Figure \ref{def:ConcaveClosure}, the value $v$ is achievable from the prior $\mu_0$ via a signal that splits the prior into two posterior $\tilde{\mu}$ and $\tilde{\mu}'$ (setting the weights so that the expected posterior equals the prior). So $V(\mu_0) = \sup \{z \mid (\mu_0,z) \in Conv(\hat{\nu})\}$ is the largest payoff Sender can achieve when the prior is $\mu_0$, and the Sender strictly benefits from persuasion if and only if $V(\mu_0) > \hat{\nu}(\mu_0)$.

The following corollary is immediate from the previous analysis.
 
\begin{corollary} If $\hat{v}$ is concave, then the Sender does not benefit from persuasion for any prior. If $\hat{v}$ is strictly convex, the Sender benefits from persuasion for every prior. 
\end{corollary}

\subsection{Back to the Example} \label{solutionBP}
Returning to the setting of Section \ref{sec:BPexample}, observe that in any Sender-preferred subgame equilibrium, the judge's action  given probability of guilt $\mu$ is
\[\hat{a}(\mu) = \left\{\begin{array}{cc}
\mbox{\emph{convict}} & \mbox{if } \mu \geq 0.5 \\
\mbox{\emph{acquit}} & \mbox{if } \mu < 0.5 \\
\end{array}\right.\]
where the tie at $\mu=0.5$ is broken in favor of the prosecutor. So the prosecutor's expected payoff is
\[\hat{v}(\mu) = \left\{\begin{array}{cc}
1 & \mbox{if } \mu \geq 0.5 \\
0 & \mbox{if } \mu < 0.5 \\
\end{array}\right.\]
as depicted in Panel (a) of Figure \ref{fig:v}.

\begin{figure}[h]
				\centering
				\includegraphics[scale=0.9]{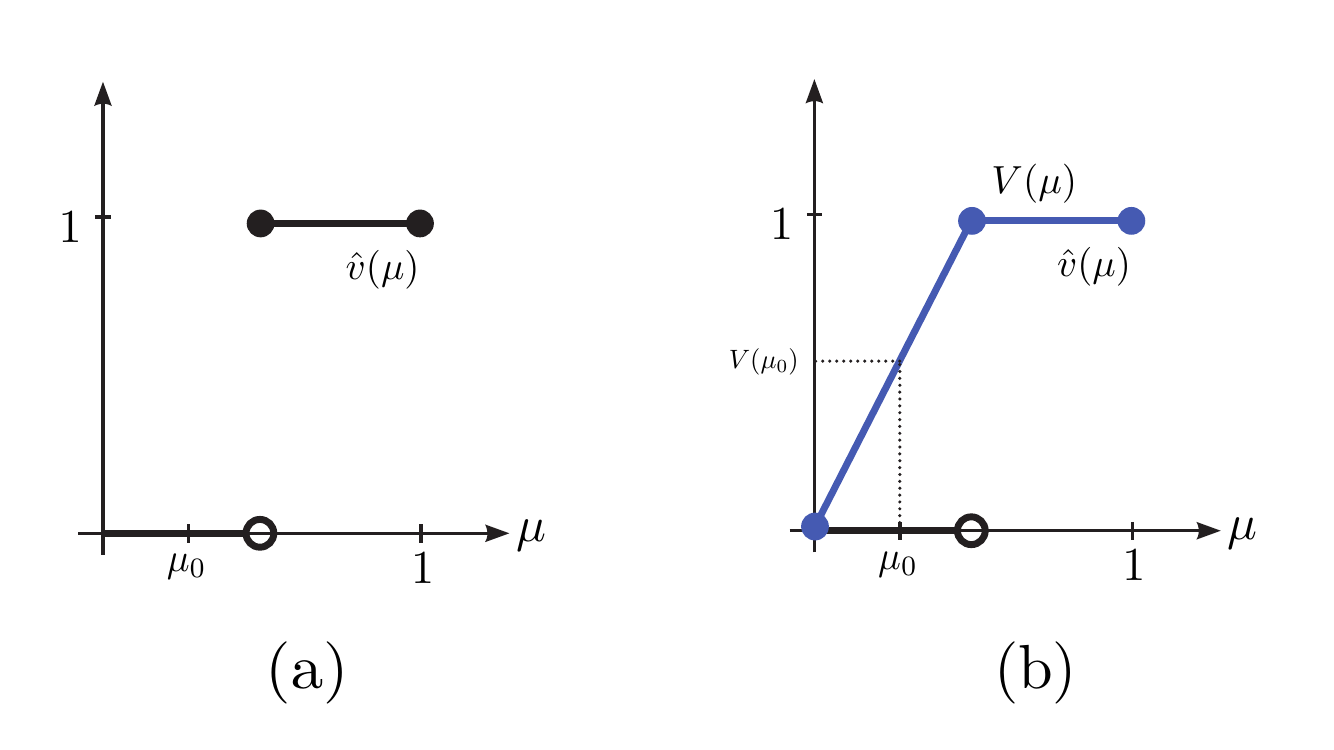}
				\caption{Depiction of $\hat{v}(\mu)$ in the prosecutor-judge example.} \label{fig:v}
			\end{figure}

The upper concave envelope of $\hat{v}$ is
\[V(\mu) = \left\{\begin{array}{cc}
1 & \mbox{if } \mu \geq 0.5 \\
2\mu & \mbox{if } \mu < 0.5 \\
\end{array}\right.\]
as depicted in Panel (b) of Figure \ref{fig:v}. At the prior belief of $\mu_0 = 0.3$, we have $V(0.3)=0.6$, confirming that the signal structure in (\ref{eq:OptimalSignal}) delivers the best possible expected payoff for the prosecutor. We see moreover that the prosecutor benefits from persuasion whenever $\mu_0 <0.5$ (i.e., whenever the judge would optimally acquit under the prior), but cannot improve his expected payoff through choice of any signal structure when $\mu_0 \geq 0.5$.

\section{Additional Exercises}

\begin{exercise}[U] A student (Sender)'s quality is $\theta \in \{L,H\}$. The employer chooses an action from $A = \{l,m,h\}$ where $l$ is a low-responsibility position, $m$ is a medium-responsibility position, and $h$ is a high-responsibility position. The employer's payoffs are:
\[u_E(a,\theta) = \left\{ \begin{array}{cc}
1 & \mbox{ if } (a,\theta) = (H,h) \\
0 & \mbox{ if } (a,\theta) \in \{(H,m), (H,l),(L,l)\} \\
-1 & \mbox{ if } (a,\theta) \in \{(L,m),(L,h)\}
\end{array}\right.\]
The student's (state-independent) payoff function $u_S$ takes value 1 if $a=h$, 0 if $a=m$, and $-1$ if $a=l$.
\begin{itemize}
\item[(a)] Suppose the employer's beliefs are described as $(p,1-p)$, where $p$ is the probability of $\theta=L$. Let
\[\hat{a}(p) = \arg \max_{a\in \{l,m,h\}} \mathbb{E}_{(p,1-p)}[u_E(a,\theta)].\]
(This is the same as in (\ref{eq:ReceiverAction}), except we simplify notation by writing $\hat{a}(p)$ instead of $\hat{a}(p,1-p)$.) Solve for $\hat{a}(p)$ on the domain $p\in [0,1]$, assuming that the employer breaks ties in favor of the action that maximizes the student's payoffs. 
\item[(b)] Suppose the student's beliefs are described as $(p,1-p)$, where $p$ is the probability of $\theta=L$, and the student knows that the employer chooses action $\hat{a}(p)$.  Let 
\[\hat{v}(p) =  \mathbb{E}_{(p,1-p)}[u_S(\hat{a}(p),\theta)]\]
denote the student's expected payoff at this belief. Solve for $\hat{v}(p)$ on the domain $p\in [0,1]$ and plot it. 
\item[(e)] Let $V(p)$ be the smallest concave function that is everywhere above $\hat{v}(p)$. Reproduce your plot from part (d) with $V(p)$ and $\hat{v}(p)$ depicted in the same figure. Clearly label $V(p)$ and $\hat{v}(p)$. 
\item[(f)] Identify all $p\in [0,1]$ such that $V(p) > \hat{v}(p)$. These are the prior beliefs at which the student can strictly benefit from design of the signal structure. 

\end{itemize}

\end{exercise}

\begin{exercise}[G] Fix an arbitrary finite set of states $\Theta$ and finite set of actions $A$. Suppose that the Sender and Receiver's payoff functions satisfy
\[u_S(a,\theta) = -u_R(a,\theta)\]
for every $a\in A$ and $\theta \in \Theta$. Prove that $V(\mu)=\hat{v}(\mu)$ for every belief $\mu$, where $\hat{v}$ is as defined in (\ref{hatv}) and $V$ is as given in Definition \ref{def:ConcaveClosure}. Interpret this result.
\end{exercise}


\end{document}